\def\draft{0}
\documentclass[11pt]{article}

\usepackage{amsmath}
\usepackage{amssymb}
\usepackage{amsthm}
\usepackage{bbm}
\usepackage{fullpage}
\usepackage{tikz}
\usepackage[most]{tcolorbox}
\usepackage{xcolor}

\usepackage{lineno}

\newcommand*{\myfont}{\fontfamily{bch}\selectfont}
\DeclareTextFontCommand{\textmyfont}{\myfont}

\usepackage[linesnumbered,ruled,vlined]{algorithm2e} 
\SetKwComment{Comment}{/* }{ */}

\SetCommentSty{mycommfont}

\usepackage[hidelinks]{hyperref}
\usepackage[nameinlink]{cleveref}
\usepackage{url}            

\newtheorem{theorem}{Theorem}[section]
\newtheorem{corollary}[theorem]{Corollary}
\newtheorem{lemma}[theorem]{Lemma}
\newtheorem{observation}[theorem]{Observation}

\newtheorem{definition}{Definition}
\newtheorem{claim}[theorem]{Claim}
\newtheorem{fact}[theorem]{Fact}
\newtheorem{remark}[theorem]{Remark}

\newcommand{\prob}[2]{\mathop{\mathrm{Pr}}_{#1}[#2]}

\newcommand{\poly}{\mathop{\mathrm{poly}}}

\newcommand{\F}{\mathbb{F}}
\newcommand{\R}{\mathbb{R}}
\newcommand{\mc}[1]{\mathcal{#1}}

\newcommand{\Boo}{\{0,1 \}}

\newcommand{\bigO}{\mathcal{O}}

\newcommand{\paren}[1]{\left( #1 \right)}
\newcommand{\brac}[1]{\left[ #1 \right]}
\newcommand{\set}[1]{\left\{ #1 \right\}}

\newcommand{\setcond}[2]{\left\{ #1 \;\middle\vert\; #2 \right\}}

\newcommand{\innerprod}[1]{\left\langle#1\right\rangle}

\newcommand{\abs}[1]{\left\lvert#1\right\rvert}
\newcommand{\Real}{\textnormal{Re}}

\DeclareMathOperator*{\E}{\mathbb{E}}

\newcommand{\x}{\mathbf{x}}

\newcommand{\cF}{\mathcal{F}}
\newcommand{\cP}{\mathcal{P}}
\newcommand{\Z}{\mathbb{Z}}

\newtcolorbox{algobox}{colback=lightgray!5!white,colframe=lightgray!75!black}

\usepackage{setspace}
\usepackage{enumerate}
\usepackage[parfill]{parskip}
\usepackage{changepage}
\usepackage{framed}

\allowdisplaybreaks
\hypersetup{
	colorlinks,
	linkcolor={blue},
	citecolor={blue},
	urlcolor={blue}
}

\ifnum\draft=1
\newcommand{\anote}[1]{{\color{brown} [Amik: #1]}}
\newcommand{\mnote}[1]{{\color{red} [Madhu: #1]}}
\newcommand{\mpnote}[1]{{\color{pink} [Manaswi: #1]}}
\newcommand{\pnote}[1]{{\color{blue} [Prashanth: #1]}}
\newcommand{\snote}[1]{{\color{green} [Srikanth: #1]}}
\else  
\newcommand{\anote}[1]{}
\newcommand{\mnote}[1]{}
\newcommand{\mpnote}[1]{}
\newcommand{\pnote}[1]{}
\newcommand{\snote}[1]{}
\fi

\date{\today}

\begin{document} 
	\title{Local Correction of Linear Functions over the Boolean Cube}

    \author{Prashanth Amireddy\thanks{School of Engineering and Applied Sciences, Harvard University, Cambridge, Massachusetts, USA. Supported in part by a Simons Investigator Award and NSF Award CCF 2152413 to Madhu Sudan and a Simons Investigator Award to Salil Vadhan. Email: \texttt{pamireddy@g.harvard.edu}} \and
    Amik Raj Behera\thanks{Department of Computer Science, Aarhus University, Aarhus, Denmark. Supported by Srikanth Srinivasan's start-up grant from Aarhus University. Email: \texttt{bamikraj@cs.au.dk} } \and
    Manaswi Paraashar\thanks{Department of Computer Science, University of Copenhagen, Denmark. Supported by Srikanth Srinivasan's start-up grants from Aarhus University and the University of Copenhagen. Email: \texttt{manaswi.isi@gmail.com} } \and
     Srikanth Srinivasan \thanks{Department of Computer Science, University of Copenhagen, Denmark. Also partially employed by Aarhus University, Denmark. This work was supported by start-up grants from Aarhus University and the University of Copenhagen. Email: \texttt{srsr@di.ku.dk} } \and 
     Madhu Sudan\thanks{School of Engineering and Applied Sciences, Harvard University, Cambridge, Massachusetts, USA. Supported in part by a Simons Investigator Award and NSF Award CCF 2152413. Email: \texttt{madhu@cs.harvard.edu}}}

	\maketitle
        \pagenumbering{arabic}

\begin{abstract}
We consider the task of locally correcting, and locally list-correcting, multivariate linear functions over the domain $\{0,1\}^n$ over arbitrary fields and more generally Abelian groups. Such functions form error-correcting codes of relative distance $1/2$ and we give local-correction algorithms correcting up to nearly $1/4$-fraction errors making $\widetilde{\bigO}(\log n)$ queries. This query complexity is optimal up to $\poly(\log\log n)$ factors.  We also give local list-correcting algorithms correcting $(1/2 - \varepsilon)$-fraction errors with $\widetilde{\bigO}_\varepsilon(\log n)$ queries. 

These results may be viewed as natural generalizations of the classical work of Goldreich and Levin whose work addresses the special case where the underlying group is $\mathbb{Z}_2$. By extending to the case where the underlying group is, say, the reals, we give the first non-trivial locally correctable codes (LCCs) over the reals (with query complexity being sublinear in the dimension (also known as message length)). 

Previous works in the area mostly focused on the case where the domain is a vector space or a group and this lends to tools that exploit symmetry. Since our domains lack such symmetries, we encounter new challenges whose resolution may be of independent interest. 
The central challenge in constructing the local corrector is constructing ``nearly balanced vectors'' over $\{-1,1\}^n$ that span $1^n$ --- we show how to construct $\bigO(\log n)$ vectors that do so, with entries in each vector summing to $\pm1$. The challenge to the local-list-correction algorithms, given the local corrector, is principally combinatorial, i.e., in proving that the number of linear functions within any Hamming ball of radius $(1/2-\varepsilon)$ is $\bigO_\varepsilon(1)$. Getting this general result covering every Abelian group requires integrating a variety of known methods with some new combinatorial ingredients analyzing the structural properties of codewords that lie within small Hamming balls.
\end{abstract} 

        \newpage 
        
\tableofcontents

        \newpage

\section{Introduction}\label{sec:intro}

In this paper we consider the class of ``linear'' functions mapping $\{0,1\}^n$ to an Abelian group and give ``local correction'' and ``local list-correction'' algorithms for this family (of codes). We describe our problems and results in detail below. We start with some basic notation.

We denote the space of functions mapping $\{0,1\}^n$ to an Abelian group $G$ by $\mc{F}(\{0,1\}^n, G)$. Given two functions $f,g$ from this set, we denote by $\delta(f,g)$ the fractional Hamming distance between them, i.e. the fraction of points in $\{0,1\}^n$ on which $f$ and $g$ disagree. In other words,
\[
\delta(f,g) = \prob{x\sim \{0,1\}^n}{f(x)\neq g(x)}.
\]
We say that $f,g$ are $\delta$-close if $\delta(f,g) \leq \delta$ and that $f,g$ are $\delta$-far otherwise. Given a set of functions $\mc{F}\subseteq \mc{F}(\{0,1\}^n, G)$, we denote by $\delta(f,\mc{F})$ the minimum distance between $f$ and a function $P\in \mc{F}.$  The function $f$ is said to be $\delta$-close if $\delta(f,\mathcal{F}) \leq \delta$ and otherwise $\delta$-far from $\mc{F}.$ We denote by $\delta(\mc{F})$ the minimum distance between two distinct functions in $\mc{F}.$ 

The main thrust of this paper is getting efficient local correcting algorithms for some basic classes of functions $\cF$ that correct close to $\delta(\cF)/2$ fraction of errors uniquely (i.e. given $f:\Boo^n\to G$ determine $P \in \cF$ such that $\delta(f,P) < \delta(\cF)/2$), and to list-correct close to $\delta(\cF)$ fraction of errors with small sized lists (i.e., given $f$ output a small list $P_1,\ldots,P_L \in \cF$ containing all functions $P$ such that $\delta(f,P)< \delta(\cF) -\varepsilon$). We start by describing our class of functions. 

\paragraph{Group valued polynomials.} The function spaces we are interested in are defined by polynomials of low-degree over the Boolean cube $\{0,1\}^n$ with coefficients from an Abelian group $G$, where we view $\{0,1\} \subseteq \mathbb{Z}$. (Thus a monomial function is given by a group element $g \in G$ and subset $S \subseteq [n]$ and takes the value $g$ at points $a \in \{0,1\}^n$ such that $a_i = 1$ for all $i \in S$. The degree of a monomial is $|S|$ and a degree $d$ polynomial is the sum of monomials of degree at most $d$.)  We let  $\mc{P}_d(\{0,1\}^n,G)$ denote the space of polynomials of degree at most $d$ in this setting. (If $n$ and $G$ are clear from context we refer to this class as simply $\mc{P}_d$.) The standard proof of the ``Schwartz-Zippel Lemma'' \cite{ore1922hohere,DL78,Zippel79,Schwartz80} extends to this setting (see \Cref{thm:basic}) and shows that two distinct degree $d$ polynomials differ on at least a $2^{-d}$ fraction of $\Boo^{n}$. Therefore $\delta(\mc{P}_d) = 2^{-d}$ and our goal is to correct half this fraction of errors uniquely and close to this fraction of errors with small-sized lists. (We do so for $d = 1$, though many results apply to general values of $d$.) Next, we describe our notion of efficiency for correction. 
\paragraph{Local Correction.} In this work we are interested in a particularly strong notion of decoding, namely ``local correction''.
 Informally, $\mathcal{F}$ is locally correctable if there exists a probabilistic algorithm $C$ that, given a point $\mathbf{x} \in \{0,1\}^n$ and oracle access to a function $f$ that is close to $\mathcal{F}$, computes $P(\mathbf{x})$ with high probability while making few queries to the oracle $f$, where $P$ is function in $\mathcal{F}$ closest to $f$. In contrast to the usual notion of decoding which would require explicitly outputting a description of $P$ (say the coefficients of $P$ when $\cF = \cP_d$) this notion thus only requires us to compute $P(\mathbf{x})$ for a given $\mathbf{x}$. The main parameters associated with a local correction algorithm are the fraction of errors it corrects and the number of queries it makes to the oracle $f$. Formally, we say $\mathcal{F}$ is $(\delta, q)$-locally correctable if there exists a probabilistic algorithm that given $\mathbf{a} \in\Boo^n$ and oracle access to the input polynomial $f$ that is $\delta$-close to $P \in \mathcal{F}$, outputs $P(\mathbf{a})$ correctly with probability at least $3/4$ by making at most $q$ queries to $f$.\newline

The question of local correction of low-degree polynomials has been widely studied \cite{BeaverF, GLRSW, GemmellS, STV-list-decoding}. These works have focused on the setting when the domain has an algebraic structure such as being a vector space over a finite field. In contrast the ``Schwartz-Zippel'' lemma only requires the domain to be a product space. Kim and Kopparty \cite{kim-kopparty-prod-decoding} were the first to study the decoding of low-degree multivariate polynomials when the domain is a product set, though they do not study local correctibility. Bafna, Srinivasan, and Sudan \cite{bafna2017local} were the first to study the problem of local correctibility of linear polynomials, though their result was mainly a negative result. 
They showed that if the underlying field $\F$ is large, for example, $\F = \R$, then any $(\Omega(1),q)$-local correction algorithm for $\mathcal{P}_{1}$ with constant $\delta$ requires at least $\Tilde{\Omega}(\log n)$ queries. In this work, we consider the task of designing local correction algorithms with nearly matching performance. 

\paragraph{Local List Correction.} When considering a fraction of errors larger than $\delta(\cF)/2$, the notion of correction that one usually appeals to is called ``list-decoding'' or ``list-correction'' as we will refer to it, to maintain a consistent distinction between the notion of recovering the message (decoding) and recovering its encoding (correction). Here the problem comes in two phases: First a combinatorial challenge of proving that for some parameter $\rho \in [\delta(\cF)/2,\delta(\cF)]$ we have an a priori bound $L = L(\rho)$ such that for every function $f:\Boo^n\to G$ there are at most $t \leq L$ functions
$P_1,\ldots,P_t \in \cF$ satisfying $\delta(f,P_i) \leq \rho$. We define $\cF$ to be $(\rho,L)$-list-decodable if it satisfies this property. Next comes the algorithmic challenge of finding the list of size of most $L$ that includes all such $P_i$'s. In the non-local setting, this is referred to as the list correction task. In the local setting, the task is subtler to define and was formalized by Sudan, Trevisan, and Vadhan \cite{STV-list-decoding} as follows:\newline
A $(\rho,L,q)$-local-list-corrector for $\cF$ is a probabilistic algorithm $C$ that takes as input an index $i \in [L]$ and $\mathbf{x} \in \Boo^n$ along with oracle access to $f:\Boo^n\to G$ such that $C^f(i,\mathbf{x})$ is computed with at most $q$ oracle calls to $f$ and $C$ satisfies the following property: For every polynomial $P \in \cF$ such that $\delta(f,P) \leq \rho$ there exists an index $i \in [L]$ such that for every $\mathbf{x} \in \Boo^n$, $\Pr[C^f(i,\mathbf{x}) = P(\mathbf{x})] \geq 3/4$. (In other words $C^f(i,\cdot)$ provides oracle access to $P$ and ranging over $i \in [L]$ gives oracle access to every $P_1,\ldots,P_t$ that are $\rho$-close to $f$, while some $i$ may output spurious functions.) \newline

The notion of list-decoding dates back to the work of Elias \cite{Elias}. The seminal work of Goldreich and Levin \cite{GoldreichL} produced the first non-trivial list-decoders for any non-trivial class of functions. (Their work happens to consider the class $\cF = \cP_1(\{0,1\}^n,\Z_2)$ and presents local list-decoders, though the argument also yields local list-correctors.) List-decoding of Reed-Solomon codes was studied by Sudan \cite{Sudan} and Guruswami and Sudan \cite{gs-list-rs}. Local list-correction algorithms for functions mapping $\F_q^n$ to $\F_q$ for polynomials of degree $d \ll q$ were given in \cite{GRS00, STV-list-decoding}. The setting of $d > q$ has been considered by Gopalan, Klivans and Zuckerman \cite{gkz-list-decoding} and Bhowmick and Lovett \cite{BhowmickL}. In a somewhat different direction Dinur, Grigorescu, Kopparty, and Sudan \cite{GrigorescuKS,DinurGKS-ECCC} consider the class of homomorphisms from $G$ to $H$ for finite Abelian groups $G$ and $H$, and give local list-correction algorithms for such classes of functions. 

In this work, we give local list-correction algorithms for the class $\cP_1(\Boo^n,G)$ for every Abelian group $G$. We explain the significance of this work in the broader context below.

\subsection{Motivation for Our Work}

The problem of decoding linear polynomials over $\{0,1\}^n$ over an arbitrary Abelian group is a natural generalization of the work of Goldreich and Levin, who consider this problem over $\mathbb{Z}_2$. However, the error-correcting properties of the underlying code (rate and distance) remain the same over any Abelian group $G$. Further, standard non-local algorithms~\cite{Reed} over $\mathbb{Z}_2$ work almost without change over any $G$ (see \Cref{app:non-local}). The local correction problem is, therefore, a natural next step.

There are also other technical motivations for our work from the limitations of known techniques. Perhaps the clearest way to highlight these limitations is that to date we have no $(\Omega(1),O(1))$-locally correctable codes over the reals with growing dimension. This glaring omission is highlighted in the results of~\cite{BDWY, DSW1,DSW2}. The work of Barak, Dvir, Wigderson and Yehudayoff \cite{BDWY} and the followup of Dvir, Saraf, and Wigderson~\cite{DSW1} show that there are no $(\Omega(1),2)$-locally correctable codes of super-constant dimension, while another result of Dvir, Saraf and Wigderson~\cite{DSW2} shows that any $(\Omega(1),3)$-locally correctable codes in $\mathbb{R}^n$ must have dimension at most $n^{1/2-\Omega(1)}.$  Indeed, till our work, there has been very little exploration of code constructions over the reals. While our work does not give an $(\Omega(1),O(1))$-locally correctable code either, ours is the first to give $n$-dimensional codes that are $(\Omega(1),\widetilde{O}(\log n))$-locally correctable. (These are obtained by setting $G = \R$ in our results.)   

A technical reason why existing local correction results do not cover any codes over the reals is that almost all such results rely on the underlying symmetries of the domain. Typical results in the area (including all the results cited above) work over a domain that is either a vector space or at least a group. Automorphisms of the domain effectively play a central role in the design and analysis of the local correction algorithms; but they also force the range of the function to be related to the domain and this restricts their scope. In our setting, while the domain $\Boo^n$ does have some symmetries, they are much less rich and unrelated to the structure of the range. Thus new techniques are needed to design and analyze local-correction algorithms in this setting. Indeed we identify a concrete new challenge --- the design of ``balanced'' vectors in $\{-1,1\}^n$ (i.e., with sum of entries being in $\{-1,1\}$) that span the vector $1^n$ --- that enables local correction, and address this challenge. We remark that correcting $\cP_d$ for $d>1$ leads to even more interesting challenges that remain open.

Another motivation is just the combinatorics of the list-decodability of this space. For instance for the class $\cF = \cP_1(\Boo^n,G)$ for any $G$, it is straightforward to show $\delta(\cF) = 1/2$ and so the unique decoding radius is $1/4$, but the list-decoding radius was not understood prior to this work. The general bound in this setting would be the Johnson bound which only promises a list-decoding radius of $1 - 1/\sqrt{2} \approx 0.29$.  Indeed a substantial portion of this work is to establish that the list-decoding radius of all these codes approaches $\delta(\cF) = 1/2$. 

Finally, we note that the complementary problem of \emph{local testing} has been quite successfully studied in grids such as $\{0,1\}^n.$ Here, we are given oracle access to a function $f$ and the problem is to determine if $f$ is close to an element of $\mathcal{{F}}.$ The problem of testing closeness to linearity (e.g. $\cF = \cP_1(\{0,1\}^n,\Z_2)$) goes back to the work of Blum, Luby and Rubinfeld~\cite{BLR} and has a long history of its own. More recently, David, Dinur, Goldenberg, Kindler and Shinkar~\cite{DDGKS17} showed how to test linearity (over $\mathbb{Z}_2$)  when the domain is a Hamming slice of $\{0,1\}^n.$ The  works~\cite{bafna2017local,dinur2019direct,bogdanov2021direct,AmireddySS} show how to test for closeness to higher-degree polynomials over groups in the setting of $\{0,1\}^n$ or other grids.

We describe our results below before turning to the proof techniques.


\subsection{Our Main Results}

Our first result is an almost optimal local correction algorithm for degree 1 polynomials up to an error slightly less than $1/4$, which is the unique decoding radius. 
\begin{theorem}[Local correction algorithms for $\mc{P}_1$ up to the unique decoding radius]
    \label{thm:uniquedeg1}
    The space $\mc{P}_1$ has a $(\delta,q)$-local correction algorithm where $\delta = \frac{1}{4}-\varepsilon$ for any constant $\varepsilon > 0$ and $q = \tilde{O}(\log n).$
\end{theorem}
We remark that \Cref{thm:uniquedeg1} is tight up to $\poly(\log \log n)$ factors due to a lower bound of $\Omega(\log n/\log \log n)$ by earlier work of Bafna, Srinivasan, and Sudan \cite{bafna2017local}.
Using further ideas, we also extend the algorithm from Theorem~\ref{thm:uniquedeg1} to the list decoding regime. For this, we need first a combinatorial list decoding bound.

\begin{theorem}[Combinatorial list decoding bound for $\mc{P}_1$]
    \label{thm:comblistdeg1}
    For any constant $\varepsilon > 0$, the space $\mc{P}_1$ over any Abelian group $G$ is $(1/2 - \varepsilon, \poly(1/\varepsilon))$-list correctable.
\end{theorem}

Using the combinatorial list decoding bound, we also give a local list correction algorithm for degree $1$ polynomials. We state the result below. For a formal definition of local list correction, refer to \Cref{defn:local-list-algo}.

\begin{theorem}[Local List Correction for degree-$1$]
\label{thm:listdecoding}
There exists a fixed polynomial $L$ such that for all Abelian groups $G$ and for every $\varepsilon>0$ and $n \in \Z^+$ we have that $\cP_1(\Boo^n,G)$ is $(1/2-\varepsilon, L(\varepsilon), \widetilde{O}(\log n))$-locally list correctable. 

Specifically, 
there is a randomized algorithm $\mathcal{A}$ that, when given oracle access to a polynomial $f$ and a parameter $\varepsilon > 0$, outputs with probability at least $3/4$ a list of randomized algorithms $\phi_1,\ldots, \phi_L$ ($L\leq \poly(1/\varepsilon)$) such that the following holds:

For each 
$P \in \mathcal{P}_1$ that is $(1/2 - \varepsilon)$-close to $f$, there is at least one algorithm $\phi_i$ that, when given oracle access to $f$, computes $P$ correctly on every input with probability at least $3/4.$

The algorithm $\mathcal{A}$ makes $O_{\varepsilon}(1)$ queries to $f$, while each $\phi_i$ makes $\tilde{O}_\varepsilon(\log n)$ queries to $f.$ 
\end{theorem}

Using known local testing results~\cite{bafna2017local, AmireddySS}, one can show that the local list-correction \Cref{thm:listdecoding} actually implies \Cref{thm:uniquedeg1}. Nevertheless, we present the proof of \Cref{thm:uniquedeg1} in its entirety, since it is a simpler self-contained proof than the one that goes through \Cref{thm:listdecoding}, and introduces some of the same ideas in an easier setting. A weak version of \Cref{thm:uniquedeg1} is also required for \Cref{thm:listdecoding}.

\subsection{Proof Overview}

\subsubsection{Local Correction - \Cref{thm:uniquedeg1}}
We prove \Cref{thm:uniquedeg1} in three steps. The first step, which is specific to the space of linear polynomials, shows how to locally correct from any oracle $f$ that $\bigO(1/\log n)$-close to a degree-$1$ polynomial in $n$ variables using $O(\log n)$ queries. In the second and third steps, we show how to increase the radius of decoding from $\bigO(1/\log n)$ to an absolute constant and then to (nearly) half the unique decoding radius. The latter two steps also work for polynomials of degree greater than $1.$

\paragraph{\underline{First step}} To motivate the proof of the first step, it is worth recalling the idea behind the lower bound result of \cite{bafna2017local} mentioned above. For simplicity, let us assume that we are working with \emph{homogeneous} linear polynomials over $\mathbb{R}$. In this situation, we can equivalently replace the domain with $\{-1,1\}^n.$ Now, assume the given oracle $f:\{-1,1\}^n\rightarrow\mathbb{R}$ agrees with some homogeneous linear polynomial $P$ at all points of Hamming weight in the range $[\frac{n}{2}-n^{0.51},\frac{n}{2}+n^{0.51}]$, and takes adversarially chosen values outside this set. It is easy to check that $f$ is $\exp\paren{-n^{\Omega(1)}}$-close to $P$. Further, assume that we only want to correct the polynomial $P$ at the point $1^n.$

Over $\mathbb{R}$, the space of homogeneous linear polynomials forms a vector space. Hence, given access to an oracle $f$ that is close to a codeword $P$, it is natural to correct $P$ using a `linear algorithm' in the following sense. To correct $P$ at $1^n$, we choose points $\mathbf{x}^{(1)},\ldots, \mathbf{x}^{(q)}\in \{-1,1\}^n$ and output 
\[
c_1 f(\mathbf{x}^{(1)}) + \cdots + c_q f(\mathbf{x}^{(q)})
\]
for some coefficients $c_1,\ldots, c_q\in \mathbb{R}.$ (Indeed, it is not hard to argue that if any strategy works, then there must be a linear algorithm that does \cite{bafna2017local}.) 

Since this strategy must work when given $P$ itself as an oracle, it must be the case that
\[
P(1^n) = c_1 P(\mathbf{x}^{(1)}) + \cdots + c_q P(\mathbf{x}^{(q)})
\]
for any linear polynomial $P$. Further, as the space of homogeneous linear polynomials is a vector space spanned by the coordinate (i.e. dictator) functions, it is necessary and sufficient to have
\begin{equation}
\label{eq:proofoverview-step1}
1^{n} = c_1 \cdot 
    \mathbf{x}^{(1)} + \cdots + c_q\cdot 
    \mathbf{x}^{(q)}.
\end{equation}
Finally, for such a correction algorithm to work for $f$ as given above, it must be the case that each of the `query points' $\mathbf{x}^{(1)},\ldots, \mathbf{x}^{(q)}$ has Hamming weight in the range $[n/2 - n^{0.51}, n/2 + n^{0.51}]$. 

Note that \Cref{eq:proofoverview-step1} cannot hold for \emph{perfectly balanced} (i.e. Hamming weight exactly $n/2$) query points, no matter what $q$ we choose: this is because the query points lie in a subspace not containing $1^{n}.$ The work of~\cite{bafna2017local} showed a robust version of this: for any set of `nearly-balanced' vectors with Hamming weight in the range $[n/2 - n^{0.51}, n/2 + n^{0.51}]$ that satisfy \Cref{eq:proofoverview-step1}, it must hold that $q = \Omega(\log n/\log \log n).$ At a high level, this lower bound holds because if \Cref{eq:proofoverview-step1} is true, then the coefficients can be taken to be at most $q^{\bigO(q)}$ in magnitude (via a suitable application of Cramer's rule). The lower bound then follows by adding up the entries of the vectors on both sides of \Cref{eq:proofoverview-step1}.

The first step of the algorithm is based on showing that this lower bound is essentially tight. More formally, we show that we can find $q = \bigO(\log n)$ nearly-balanced\footnote{In fact, the vectors we construct have Hamming weight $n/2 \pm 1$.} vectors $\mathbf{x}^{(1)},\ldots,\mathbf{x}^{(q)}\in \{0,1\}^n$ such that the following (more general) equation holds.
\begin{equation}
\label{eq:proofoverview-step1-2}
1^{n+1} = c_1 \cdot 
\left(\begin{array}{c}1\\ \mathbf{x}^{(1)}\end{array}\right) + \cdots + c_q\cdot \left(\begin{array}{c}1\\ \mathbf{x}^{(q)}\end{array}\right).
\end{equation}
This identity allows us to correct any linear (not just homogeneous) polynomial. Moreover, we show that we can take the coefficients to be \emph{integers}, which allows us to apply this algorithm over any Abelian group.\footnote{It makes sense to multiply a group element $g$ with an integer $k$, since it amounts to adding either the element $g$ or its inverse $-g$, $|k|$ times.}

Finally, we show that this construction also implies a similar algorithm to compute $P(1^n)$ from \emph{any} $f$ that is $\bigO(1/\log n)$-close to $P$ (and not just the special $f$ given above). This is done by constructing random query points $\mathbf{y}^{(1)},\ldots, \mathbf{y}^{(q)}$ where the $i$th bit of these vectors is picked by choosing a \emph{random} bit of $\mathbf{x}^{(1)},\ldots, \mathbf{x}^{(q)}$. The fact that each $\mathbf{x}^{(j)}$ is nearly balanced implies that each $\mathbf{y}^{(j)}$ is nearly uniform over $\{0,1\}^n$ and hence likely not an error point of $f.$ Intuitively, the distance requirement is because we make $q$ (nearly) random queries to $f$ and the algorithm succeeds if none of the query points is in error. So, the algorithm correctly computes $P(1^n)$ when $\delta(f,P)$ is sufficiently smaller than $1/q$. By a suitable `shift', we can also correct at points other than $1^n.$

This construction of the points $\mathbf{x}^{(1)},\ldots,\mathbf{x}^{(q)}\in \{0,1\}^n$ is based on ensuring that the coefficients $c_1,\ldots, c_q$ must be exponentially large in $q$ (to ensure that the argument of~\cite{bafna2017local} is tight). This leads to the natural problem of finding a hyperplane whose Boolean solutions cannot be described by an equation with small coefficients. This is a topic that has received much interest in the study of threshold circuits and combinatorics~\cite{GHR,Hastad,AlonVu,Podolskii,BabaiHansenPodolskii}.

For the result stated above, we require only a simple construction. Consider the following equation over $\{0,1\}^{q}$ where $q = 2k$. The first $k$ bits describe an integer $i\in \{0,\ldots, 2^k-1\}$ and the last $k$ bits describe an integer $j.$ The hyperplane expresses the constraint that $j=i-1.$ This hyperplane can easily be described using coefficients that are exponentially large in $k$ and one can easily show that this is in fact necessary. After some modification to ensure that the coefficients sum to $1,$ we get \Cref{eq:proofoverview-step1-2}. See \Cref{lemma:decode-1n} for more details.

Using a more involved construction due to H\r{a}stad~\cite{Hastad} and its extension due to Alon and Vu~\cite{AlonVu}, it is possible to show that we can achieve $q= O(\log n/\log \log n)$, showing that the lower bound of~\cite{bafna2017local} is in fact tight up to constant factors (see \Cref{app:improved-uniquedeg1-smallerror}). However, in this case, we don't know how to ensure that the coefficients $c_1,\ldots, c_q$ are integers, meaning that the algorithm does not extend to general Abelian groups.\footnote{Moreover, we also lose $\poly(\log\log n)$ factors in query complexity in the subsequent steps, and so the final algorithm is only tight up to $\poly(\log\log n)$ factors, no matter which construction we use in the first step.}

\paragraph{\underline{Second and third steps}} To obtain an algorithm resilient to a larger fraction of errors, we use a process of error reduction. Specifically, we show that, given an oracle $f:\{0,1\}^n\rightarrow G$ that is $\delta$-close to a polynomial $P\in \mathcal{P}_1$, we can obtain (with high probability) an oracle $g:\{0,1\}^n\rightarrow G$ that is $\bigO(1/\log n)$-close to $P$; we can then apply the above described local correction algorithm to $g$ to correct $P$ at any given point. The oracle $g$ makes $\poly(\log \log n)$ queries to $f$ and hence the overall number of queries to $f$ is $\tilde{\bigO}(\log n).$ 

Interestingly, the error-reduction step is not limited to linear polynomials. We show that this also works for the space of degree-$d$ polynomials, where the number of queries now also depends on the degree parameter $d.$ In general, $\delta$ can be arbitrarily close to the unique decoding radius of $\mathcal{P}_d$, which is $2^{-(d+1)}.$

We use two slightly different error-reduction algorithms to handle the case when $\delta$ is a small constant, and the case when $\delta$ is close to $2^{-(d+1)}$ respectively. We reduce the latter case to the former case and the former to the case of error $\bigO(1/\log n).$ It is simpler to describe the error reduction algorithm when the error is large, i.e. close to the unique decoding radius, so we start there.

The process of error-reduction may be viewed as an \emph{average-case} version of the correction problem, where we are only required to compute $P$ on \emph{most} points in $\{0,1\}^n$ with high probability. Assume, therefore, that we are given a \emph{random} point $\mathbf{a}\in \{0,1\}^n$ and we are required to output $P(\mathbf{a})$ (with high probability).

In the setting where the domain is not $\{0,1\}^n$ but rather a vector space like $\F_q^n$, a natural strategy going back to the works of Beaver and Feigenbaum~\cite{BeaverF} and Lipton~\cite{Lipton} is to choose a random subspace $V$ of appropriate constant\footnote{Here, we think of all parameters except $n$ as constants.} dimension $k$ containing $\mathbf{a}$ and then find the closest $k$-variate degree-$d$ polynomial to the restriction $f|_V$ of $f$ to this subspace. The reason this works is that the points in a random subspace come from a pairwise independent distribution and hence standard second-moment methods show that $\delta(f|_V,P|_V)\approx \delta$ with high probability, in which case $\delta(f|_V,P|_V)$ is also less than the unique-decoding radius of $\mathcal{P}_d.$ A brute-force algorithm (or better ones, such as the Welch-Berlekamp algorithm (see e.g. \cite[Chapter 15]{GRS-codingbook})) can now be used to find $P|_V$, which also determines $P(\mathbf{a}).$

To adapt this idea to the setting of $\{0,1\}^n$, we note that random subspaces are not available to us since most constant-dimensional subspaces don't have points in $\{0,1\}^n$. However, we observe that we can apply the above idea to a \emph{random subcube} in $\{0,1\}^n.$ More specifically, we identify variables randomly into $k$ buckets via a random hash function $h:[n]\rightarrow [k]$, reducing the original set of $n$ variables $x_1,\ldots, x_n$ to a set of $k$ variables $y_1,\ldots, y_k.$ Further, to ensure that the given point $\mathbf{a}$ is in the chosen subcube, we start by replacing $x_i$ by $x_i \oplus a_i$ before the identification process.\footnote{The process of XORing a variable $x$ by a Boolean value $b$ is equivalent to either leaving the variable as is when $b=0$, or replacing $x$ by $1-x$ when $b=1$. This does not affect the degree of the polynomial $P.$} This gives rise to a random subcube $\mathsf{C}$ containing $\mathbf{a}$ (obtained by setting $y_1 = \cdots = y_k= 0$). We define a random subcube formally in \Cref{defn:random-embedding}. We can now apply the above idea by restricting the given $f$ to this subcube. 

Having defined a subcube $\mathsf{C}$ as above, the non-trivial part of the argument is to show that $\delta(f|_\mathsf{C}, P|_\mathsf{C})\approx \delta$. This is not obvious as the points of the subcube $\mathsf{C}$ are not pairwise independent. Nevertheless, for random $\mathbf{a}$, the points of $\mathsf{C}$ are `noisy' copies of one another (\Cref{def:noisedistribution}). Using this fact and standard hypercontractivity estimates, we can show that most pairs of points of $C$ are `approximately' pairwise independent (see \Cref{thm:hypercontractivity} below) as long as $k$ is a large enough constant. This allows us to use the second-moment method to recover $P(\mathbf{a})$ as before, for all but a  small fraction $\delta'$ of possible inputs $\mathbf{a}$ (with high probability). The parameter $k$ is $\poly(1/\delta')$ and hence the query complexity is constant as long as the required error $\delta'$ is constant. This step is proved in \Cref{subsec:error-reduction-small-constant}.

To reduce the error further down to $\bigO(1/\log n)$, we modify the above idea. We repeat the above process\footnote{Actually, we need to slightly modify the process to ensure that we only query `balanced' points on the subcube $\mathsf{C}$. We postpone this detail to the proof.} on three randomly chosen subcubes of dimension $k'$ each containing $\mathbf{a}$ and take a plurality vote of their outputs. The probability of error in this algorithm is bounded by the probability that at least two of the iterations query a point of error, which would be $\leq \bigO_{k'}((\delta')^2)$ if the repetitions were entirely independent. However, the iterations here have some dependency - each iteration uses the \emph{same} random input $\mathbf{a}.$ Nevertheless, using hypercontractivity, we can again argue that if $k'$ is a large enough constant \emph{depending only on $d$}, the probability of error is at most $\bigO_{k'}((\delta')^{1.5})\leq (\delta')^{1.1}$ for small enough $\delta'.$ Repeating this process $t$ times, gives an error that is \emph{double-exponentially} small in $t$, at the expense of $\bigO_{k'}(1)^t$ many queries. Choosing $t$ to be $\bigO(\log \log \log n)$ gives us an oracle that is $\bigO(1/\log n)$-close to $P.$ This step is proved in \Cref{subsec:error-reduction-close-radius}.

\subsubsection{Combinatorial List Decoding Bound - \Cref{thm:comblistdeg1}}

We first note that the list size can indeed be as large as $\poly(1/\varepsilon)$, no matter the underlying group $G.$ This is shown by the following example. Fix an integer parameter $t$ and any non-zero elment $g\in G$. Let $f=\mathrm{Maj}_G^t(x_1,\ldots, x_t)$ denote the function of the first $t$ input variables that takes the value $g$ when it's input has Hamming weight greater than $t/2$ and the value $0$ otherwise. A standard calculation (see e.g.~\cite[Theorem 5.19]{odonnellbook}) shows that $\mathrm{Maj}_G^t$ agrees with the linear functions $g\cdot x_i$ ($i\in [t]$) on a $(\frac{1}{2}+\frac{O(1)}{\sqrt{t}}$)-fraction of inputs. Setting $t = \Theta(1/\varepsilon^2),$ we see that this agreement can be made $\frac{1}{2} + \varepsilon.$ This implies that for $f$ as defined above, the list size at distance $\frac{1}{2}-\varepsilon$ can be as large as $\Omega((1/\varepsilon)^2)).$

To motivate the proof of \Cref{thm:comblistdeg1}, it is helpful to start with the case when $G$ is a group $\mathbb{Z}_p$ of prime order. There are two extremes in this case: $p=2$ and large $p$ (say a large constant or even growing with $n$).

\paragraph{\underline{Case 1: $p=2$}} The case $p=2$ is the classical setting that has been intensively investigated in the literature, starting with the foundational work of Goldreich and Levin~\cite{GoldreichL} (see also the work of Kushilevitz and Mansour~\cite{KM}). In this setting, it is well-known that the bound of $1/\varepsilon^2$ is tight. This follows from the standard Parseval identity from basic Fourier analysis of Boolean functions (see e.g. \cite{odonnellbook}) or as a special case of the binary Johnson bound (see e.g. the appendix of \cite{DinurGKS-ECCC}). At a high level, this is because the Boolean Fourier transform identifies each $f:\{0,1\}^n\rightarrow \mathbb{Z}_2$ with a real unit vector $\mathbf{v}_f$ such that distinct linear polynomials are mapped to an orthonormal basis. Moreover, if $f$ is $(\frac{1}{2}-\varepsilon)$-close to a linear polynomial $P$, then the length of projection of the vector $\mathbf{v}_f$ on $\mathbf{v}_P$ is at least $\varepsilon$. Pythagoras' theorem now implies the list bound. 

\paragraph{\underline{Case 2: Large $p$}} For $p > 2$, it is unclear if we can map distinct linear polynomials to orthogonal real or complex vectors in the above way. Nevertheless, we do expect the list-size bound to hold, as the distance $\delta(\mathcal{P}_1)$ is the same as over $\mathbb{Z}_2$, i.e. $1/2$. Moreover, a \emph{random} pair of linear polynomials have a distance much larger than $1/2$ for large $p.$ This latter fact is a consequence of \emph{anti-concentration} of linear polynomials, which informally means the following. Let $P(\mathbf{x})$ be a non-zero polynomial with many (say, at least $100$) non-zero coefficients. Then, on a random input $\mathbf{a}$, the random variable $P(\mathbf{a})$ does not take any given value $b\in \mathbb{Z}_p$ with good probability (say, greater than $1/5$).

In the case of large $p$, we crucially use anti-concentration to argue the upper bound on the list size. At a high level, in this case, we can show that if a function $f:\{0,1\}^n\rightarrow \mathbb{Z}_p$ is $(\frac{1}{2}-\varepsilon)$-close to many (say $L$) linear polynomials, then there is a large subset (size $L' = L^{\Omega(1)}$) that `look' somewhat like the example of the $\mathrm{Maj}_G^t$ example mentioned above. More precisely, the coefficient vectors of the linear polynomials in this subset are at most a \emph{constant} (independent of $p$, order of the underlying group) Hamming distance from one another. By shifting the polynomials by one of the linear functions in the subset, we can assume without loss of generality that all the linear functions in fact have a constant number of \emph{non-zero coefficients}, as in the case of the list of polynomials corresponding to $\mathrm{Maj}_G^t.$ It now suffices to bound the size $L'$ of this subset by $\poly(1/\varepsilon).$

The bound now reduces to a case analysis based on the number of variables that appear in the coefficients of the $L'$ polynomials in the subset. The case analysis is based on carefully interpolating between the following two extreme cases.
\begin{itemize}
    \item The first is that all the $L'$ polynomials and also the function $f$ itself depend on the \emph{same} set of variables. Assume this set is $S = \{x_1,\ldots, x_\ell\}$ where $\ell$ must be a constant (based on the previous paragraph). In this case, we note that each polynomial $P$ in the list is specified completely by the subset of $A\subseteq \{0,1\}^\ell$ where $P$ agrees with $f$ (by the fact that $\delta(\mathcal{P}_1) = 1/2$). Since the number of such $A$ is at most $2^{2^\ell}$, this bounds $L'$ by the same quantity.

    \item The other extreme case is that the polynomials in the list all depend on \emph{disjoint} sets of variables. In this case, on a random input $\mathbf{a}\in \{0,1\}^n$, the polynomials in the list all output \emph{independent} random values in $\mathbb{Z}_p$. By a Chernoff bound, it is easy to argue that the chance that significantly more than $\frac{L'}{2} + \sqrt{L'}$ of these polynomials agree with $f(\mathbf{a})$ is very small. By an averaging argument, this implies that $L'$ is $\tilde{\bigO}(1/\varepsilon^2)$, and we are done.
\end{itemize}

\paragraph{\underline{Putting it together}} We sketch here how to handle general finite Abelian groups. In the proof, we show that this also implies the same bound for infinite groups such as $\mathbb{R}$.

Recall that any finite group $G$ is a direct product of cyclic groups, each of which has a size that is a prime power. We collect the terms in this product to write $G = G_1 \times G_2 \times G_3$ where $G_1$ is the product of the factors of sizes that are powers of $2,$ $G_2$ is the same with powers of $3$, and $G_3$ is the product of the rest.\footnote{There is nothing very special about this decomposition. Essentially, we have one argument that works for `small' $p$ and another that works for `large' $p$. To combine them, we need some formalization of these notions. Here, `large' could be defined to be larger than any constant $C \geq 5.$} A simple observation shows that it suffices to bound the size of the list in each of these cases by $\poly(1/\varepsilon).$

For $G_3,$ the argument of large $p$ sketched above works without any change (with some care to ensure that we can handle all the primes greater than or equal to $5$). The only part of the argument that is sensitive to the choice of the group is the initial use of anti-concentration, and this works over $G_3$ since the order of any non-zero element is large (i.e. at least $5$).

The argument for $G_1$ needs more work. While the use of Parseval's identity works over $\mathbb{Z}_2$, it is not clear how to extend it to groups of size powers of $2$, such as $\mathbb{Z}_4.$ For inspiration, we turn to a different extension of the $\mathbb{Z}_2$-case proved by Dinur, Grigorescu, Kopparty, and Sudan \cite{DinurGKS-ECCC}. They deal with the list-decodability of the space of \emph{group homomorphisms} from a group $H$ to a group $G$. Setting the group $H$ to be $\{0,1\}^n$ (with addition defined by the XOR operation) and $G$ to be $\mathbb{Z}_2,$ we recover again the setting of (homogeneous) linear polynomials over $\mathbb{Z}_2.$ The work of \cite{DinurGKS-ECCC} show how to extend this result to larger groups $G$ that have order a power of $2$. Note that it is not immediately clear that this should carry over to the setting of linear polynomials: for groups of order greater than $2$, the space of polynomials is different from the space of homomorphisms. However, we show that the technique of \cite{DinurGKS-ECCC} does work in our setting as well. 

Finally, the proof for $G_2$ is a combination of the ideas of the two proofs above. We omit the details here and refer the reader to the actual proof.

\subsubsection{Local List Correction - \Cref{thm:listdecoding}}

Like the proof of the second and third steps of \Cref{thm:uniquedeg1} described above, at the heart of our local list correction algorithm lies an error-reduction algorithm. More precisely, we design an algorithm $\mathcal{A}_1^{f}$ which, using oracle access to $f$, produces a list of algorithms $\psi_1,\ldots, \psi_L$ such that, with high probability, for each linear polynomial $P$ that is $(\frac{1}{2}-\varepsilon)$-close to $f$, there is at least one algorithm $\psi_{j}$ in the list that agrees with $P$ on \emph{most} inputs, i.e. $\psi_{j}$ ``approximates'' $P$. Here, $L \leq L(\varepsilon)$ denotes the list-size bound proved in \Cref{thm:comblistdeg1}. Further, each $\psi_i$ makes at most $O_L(1) = O_\varepsilon(1)$ queries to $f$. 

We can now apply the algorithm from the unique correction setting with oracle access to the various $\psi_{j}$ to produce the desired list $\phi_1,\ldots, \phi_L$ as required.


The proof is motivated by a local list-decoding algorithm for low-degree polynomials over $\F_q^n$ due to Sudan, Trevisan, and Vadhan~\cite{STV-list-decoding}. In that setting, we are given oracle access to a function $f:\F_q^n \rightarrow \F_q$ and we are required to produce a list as above that approximates the set $S = \{P_1,\ldots, P_L\}$ of degree-$d$ polynomials (say $d= o(q)$) that have significant (say $\Omega(1)$) agreement with $f$. It follows from the Johnson bound that $L= \bigO(1)$ in this case (see, e.g. \cite[Chapter 7]{GRS-codingbook}). The corresponding algorithm $\mathcal{A}_{\mathrm{STV}}$ chooses a \emph{random} point $\mathbf{a}$ and gets as advice the values of $\alpha_i = P_i(\mathbf{a})$ for each $i\in [L]$. (We can easily get rid of this advice assumption, but let us assume for now that we have it.) 

Now, we want to produce an algorithm that approximates $P_i$. Given a random point $\mathbf{b}\in \F_q^n$, the algorithm constructs the random line $\ell$ passing through $\mathbf{a}$ and $\mathbf{b}$ and produces the list of univariate polynomials that have significant agreement with the restriction $f|_\ell$ of $f$ to the line. This can be done via brute force with $O(d)$ queries (if one only cares about query complexity) or in $\poly(d,\log q)$ time using Sudan's list decoding algorithm for univariate polynomials \cite{Sudan}. By pairwise independence and standard second-moment estimates, it is easy to argue that for each $j\in [L]$, $P_j|_\ell$ is in this list of univariate polynomials. However, to single out $P_i|_\ell$ in this list, we use advice $\alpha_i = P_i(\mathbf{a})$. Since $\mathbf{a}$ is a \emph{random} point on $\ell$ (even given $\ell$), it follows that, with high probability, $\alpha_i$ uniquely disambiguates $P_i|_\ell$ from the ($\bigO(1)$ many) other polynomials in the list. In particular this also determines $P_i(\mathbf{b})$, since $\mathbf{b}$ lies on $\ell$.

Let us now turn to our local list correction algorithm. We use similar ideas to \cite{STV-list-decoding} but, as in the proof of \Cref{thm:uniquedeg1}, with \emph{subcubes} instead of lines. More precisely, the algorithm $\mathcal{A}_1^{f}$ produces a random $\mathbf{a}$ and a random hash function $h:[n]\rightarrow [k]$ ($k = \bigO_{\varepsilon}(1)$ suitably large), and uses them to produce a random subcube $\mathsf{C}$ as in the proof sketch of \Cref{thm:uniquedeg1}. The advice in this case is the restriction $P|_{\mathsf{C}}$ for each polynomial $P$ in the set $S=\{P_1,\ldots, P_L\}$ of degree-$1$ polynomials that are $(\frac{1}{2}-\varepsilon)$-close to $f$. 

Now, given a random point $\mathbf{b}\in \{0,1\}^n$, we correct $P_i(\mathbf{b})$ as follows. We first construct the smallest subcube $\mathsf{C'}$ that contains both $\mathsf{C}$ and the point $\mathbf{b}$. With high probability, this is a subcube of dimension $2k$. Using a simple brute-force algorithm that uses $2^{2k}$ queries to $f$, we can find the set $S'$ of all $2k$-variate linear polynomials that are $(\frac{1}{2}-\frac{\varepsilon}{2})$-close to $f|_{\mathsf{C'}}$. Note that $|S'|\leq L(\varepsilon).$ By a hypercontractivity-based argument (as we did in the error reduction algorithms), we can show that, with high probability, each $P_j|_{\mathsf{C'}}$ is in this list $S'$. To single out $P_i|_{\mathsf{C'}}$, we use advice $P_i|_\mathsf{C}$. The proof that this works needs an understanding of the distribution of $\mathsf{C}$ given $\mathsf{C'}$: it turns out that the $k$-dimension subcube $\mathsf{C}$ is obtained by randomly pairing up variables in $\mathsf{C'}$ and identifying them with a single variable. We show that, if $k$ is large enough in comparison to the list bound $L$, then with high probability, this process does not identify any two distinct elements in the list.\footnote{There is a small subtlety in the argument that is being hidden here for simplicity.} Thus, we are able to single out $P_i|_{\mathsf{C'}}$ and this allows us to compute $P_i(\mathbf{b})$ correctly, with high probability over the choice of $\mathbf{b}$ and the randomness of the algorithm (which includes $\mathbf{a}$ and the hash function $h$).

Finally, to get rid of the advice, we note that a similar hypercontractivity-based argument also shows that each $P_i|_\mathsf{C}$ is $(\frac{1}{2}-\frac{\varepsilon}{2})$-close to $f|_\mathsf{C}$. So by applying a similar brute-force algorithm on $\mathsf{C}$, we find, with high probability, a set $\tilde{S}$ of polynomials containing $P_i|_\mathsf{C}$ for each $i\in [L]$. This is good enough for the argument above. The algorithm $\mathcal{A}_1^{f}$ first computes $\tilde{S}$ and then outputs the descriptions of the algorithm in the previous paragraph for each $P\in \tilde{S}$ (treating it as a restriction of one of the $P_i$).

\section{Preliminaries}
\label{sec:prelims}
\subsection{Notation}
Let $(G, +)$ denote an Abelian group $G$ with addition as the binary operation. For any $g \in G$, let $-g$ denote the inverse of $g \in G$. For any $g \in G$ and integer $a \geq 0$, $a \cdot g$ (or simply $ag$) is the shorthand notation of $\underbrace{g + \ldots + g}_{a \; \mathrm{times}}$ and $-ag$ denotes $a\cdot (-g).$

For a natural number $n$, we consider functions $f :\Boo^{n} \to G$. We denote the set of functions that can be expressed as a multilinear polynomial of degree $d$, with the coefficients being in $G$ by $\mathcal{P}_{d}(n,G)$. We will simply write $P_{d}$ when $n$ and $G$ are clear from the context. For $\mathbf{x}, \mathbf{y} \in \Boo^{n}$, let $\delta(\mathbf{x},\mathbf{y})$ denote the relative Hamming distance between $\mathbf{x}$ and $\mathbf{y}$, i.e. $\delta(\mathbf{x}, \mathbf{y}) = |\setcond{i \in [n]}{x_{i} \neq y_{i}}|/n$.

For any $\mathbf{x} \in \Boo^{n}$, $|\mathbf{x}|$ denotes the Hamming weight of $\mathbf{x}$. $\tilde{O}(\cdot)$ notation hides factors that are poly-logarithmic in its argument.
For a polynomial $P(\mathbf{x})$, let $\mathrm{vars}(P)$ denote the variables on which $P$ depends, i.e. the variables that appear in a monomial with non-zero coefficient in $P$.

For any natural number $n$, $U_{n}$ denotes the uniform distribution on $\Boo^{n}$. 

\subsection{Basic Definitions and Tools}

\paragraph{Probabilistic notions.} For any distribution $X$ on $\Boo^{n}$, let $\mathrm{supp}(X)$ denote the subset of $\Boo^{n}$ on which $X$ takes non-zero probability. For two distributions $X$ and $Y$ on $\Boo^{n}$, the statistical distance between $X$ and $Y$, denoted by $\mathrm{SD}(X,Y)$ is defined as 
\begin{align*}
    \mathrm{SD}(X,Y) = \max_{T\subseteq \{0,1\}^n} |\Pr[X \in T] - \Pr[Y \in T]|
\end{align*}
We say $X$ and $Y$ are $\varepsilon$-close if the statistical distance between $X$ and $Y$ is at most $\varepsilon$.

\paragraph{Coding theory notions.} Fix an Abelian group $G$. We use $\mathcal{P}_d$ to denote the space of multilinear polynomials from $\{0,1\}^n$ to $G$ of degree at most $d.$ More precisely, any element $P\in \mathcal {P}_d$ can be described as
\[
P(x_1,\ldots, x_n) = \sum_{I\subseteq [n] \, : \, |I|\leq d} \alpha_I \prod_{i\in I}x_i
\]
where $\alpha_I\in G$ for each $I.$ On an input $\mathbf{a}\in \{0,1\}^n$, each monomial evaluates to a group element in $G$ and the polynomial evaluates to the sum of these group elements. 

The following is a summary of standard facts about multilinear polynomials, which also hold true in the setting when the range is an arbitrary Abelian group $G.$ The proofs are standard and omitted.
\begin{theorem}
    \label{thm:basic}
    \begin{enumerate}
        \item (M\"{o}bius Inversion) Any function $f:\{0,1\}^n\rightarrow G$ has a unique representation as a multilinear polynomial 
 in $\mathcal{P}_n$. Moreover, we have $f= \sum_{I\subseteq [n]} c_I \prod_{i\in I}x_i$ where for any $I\subseteq [n],$ we have 
        \[
        c_I  = \sum_{J\subseteq I} (-1)^{|I\setminus J|} f(1_J) 
        \]
        where $1_J$ is the indicator vector of the set $J.$
        \item (DeMillo-Lipton-Schwartz-Zippel) Any non-zero polynomial $P\in \mathcal{P}_d$ is non-zero with probability at least $2^{-d}$ at a uniformly random input from $\{0,1\}^n.$ Equivalently, $\delta(\mathcal{P}_d)\geq 2^{-d}.$
    \end{enumerate}
\end{theorem}

We now turn to the kinds of algorithms we will consider. Below, let $\mathcal{F}$ be any space of functions mapping $\{0,1\}^n$ to $G.$

\begin{definition}[Local Correction Algorithm]
We say that $\mc{F}$ has a $(\delta,q)$-local correction algorithm if there is a probabilistic algorithm that, when given oracle access to a function $f$ that is $\delta$-close to some $P\in \mc{F}$, and given as input some $\mathbf{a} \in \{0,1\}^n$, returns $P(\mathbf{a})$ with probability at least $3/4.$ Moreover, the algorithm makes at most $q$ queries to its oracle.
\end{definition}

\begin{definition}[Local List-Correction Algorithm]\label{defn:local-list-algo}
We say that $\mc{F}$ has a $(\delta,q_1,q_2, L)$-local list correction algorithm if there is a randomized algorithm $\mc{A}$ that, when given oracle access to a function $f$, produces a list of randomized algorithms $\phi_1,\ldots, \phi_L$, where each $\phi_{i}$ has oracle access to $f$  and have the following property: with probability at least $3/4$, for each codeword $P$ that is $\delta$-close to $f$, there exists some $i\in [L]$ such that the algorithm $\phi_i$ computes $P$ with error at most $1/4$, i.e. on any input $\mathbf{a}$, the algorithm $\phi_i$ outputs $P(\mathbf{a})$ with probability at least $3/4$. \newline
Moreover, the algorithm $\mc{A}$ makes at most $q_1$ queries to $f$, while the algorithms $\phi_1,\ldots, \phi_L$ each make at most $q_2$ queries to $f$.
\end{definition}

\begin{remark}
    \label{rem:algos}
    Our algorithms can all be implemented as standard Boolean circuits with the added ability to manipulate elements of the underlying group $G$. Specifically, we assume that we can store group elements, perform group operations (addition, inverse) and compare two group elements to check if they are equal.
\end{remark}

\begin{definition}[Combinatorial List Decodability]
We say that $\mc{F}$ is $(\delta,L)$-list decodable if for any function $f$, the number of elements of $\mc{F}$ that are $\delta$-close to $f$ is at most $L.$
\end{definition}

The questions of decoding polynomial-based codes over groups become much more amenable to known techniques if we drop the locality constraint. In \Cref{app:non-local}, we show how to modify the standard Majority-logic decoding algorithm to obtain non-local unique and list-decoding algorithms for $\mathcal{P}_d.$

\paragraph{Hypercontractivity.} Next we are going to state a consequence of the standard Hypercontracitivity theorem (Refer to \cite[Chapter 9]{odonnellbook}).
\begin{definition}[Noise distribution]
\label{def:noisedistribution}
Let $\rho \in [-1,1]$. For a fixed $\mathbf{x} \in \{0,1\}^{n}$, $\mathbf{y} \sim \mathcal{N}_{\rho}(\mathbf{x})$ denotes a random variable defined as follows: For each $i \in [n]$ independently,
\begin{align*}
    y_{i} := \begin{cases}
        x_{i}, & \text{with prob. } \, (1+\rho)/2 \\
        \neg x_{i}, & \text{with prob. } \, (1-\rho)/2
    \end{cases}
\end{align*}
In other words, to sample from the distribution $\mathcal{N}_{\rho}(\mathbf{x}))$, we flip each bit of $\mathbf{x}$ independently with probability $(1-\rho)/2$, and keeping it unchanged with probability $(1+\rho)/2$.
\end{definition}

\begin{theorem}[{\protect \cite[Section 9.5]{odonnellbook}}]\label{thm:hypercontractivity}
Let $E \subseteq \Boo^{n}$ be a subset of density $\delta$, i.e. $|E|/2^{n} = \delta$. Then for any $\rho \in [-1,1]$,
\begin{align*}
    \Pr_{\substack{\mathbf{x} \sim \Boo^{n} \\ \mathbf{y} \sim \mathcal{N}_{\rho}(\mathbf{x})}}[\mathbf{x} \in E \; \text{and} \; \mathbf{y} \in E] \leq \delta^{2/(1+|\rho|)}
\end{align*}
\end{theorem}
In particular, if $\rho$ is close to $0$, then \Cref{thm:hypercontractivity} tells us that the probability that $\mathbf{x}$ and $\mathbf{y}$ are in $E$ is close to the probability in the case that  $\mathbf{x}$ and $\mathbf{y}$ are sampled independently and uniformly from $U_{n}$.

\paragraph{Subcubes of $\{0,1\}^n$.} It will be very useful in our algorithms to be able to restrict the given function to a small-dimensional subcube and analyze this restriction. We construct such subcubes by first negating a subset of the variables and then identifying them into a smaller set of variables. More precisely, we have the following definition.

\begin{definition}[Embedding a smaller cube into $\{0,1\}^n$]\label{defn:random-embedding}
Fix any $k \in \mathbb{N}$ and $k \leq n$.  Fix a point $\mathbf{a} \in \Boo^{n}$ and a function $h: [n] \to [k]$. For every $\mathbf{y} \in \Boo^{k}$, $x(\mathbf{y})$ is defined with respect to $\mathbf{a}$ and $h$ as follows:
\begin{align*}
    x(\mathbf{y})_{i} = y_{h(i)} \oplus a_{i} = \begin{cases}
        a_{i}, & \text{if } \, y_{h(i)} = 0 \\
        1 \oplus a_{i}, & \text{if } \, y_{h(i)} = 1
     \end{cases}
\end{align*}
$C_{\mathbf{a},h}$ is the subset in $\Boo^{n}$ consisting of $x(\mathbf{y})$ for every $\mathbf{y} \in \Boo^{k}$, i.e. $C_{\mathbf{a},h} := \setcond{x(\mathbf{y})}{\mathbf{y} \in \Boo^{k}}$.

Given any polynomial $P(x_1,\ldots, x_n)$ and any subcube $C_{\mathbf{a},h}$  as above, $P$ restricts naturally to a degree-$d$ polynomial $Q(y_1,\ldots, y_k)$ on $C_{\mathbf{a},h}$ obtained by replacing each $x_i$ by $y_{h(i)}\oplus a_i$. We use $P|_{C_{\mathbf{a},h}}$ to denote the polynomial $Q$.
\end{definition}

\paragraph{Random subcubes.} Now assume that we choose a subcube $C_{\mathbf{a},h}$ by sampling  $\mathbf{a} \sim \{0,1\}^n$ and sampling a random hash function $h:[n]\rightarrow [k]$. For any $\mathbf{y} \in \Boo^{k}$, $x(\mathbf{y})$ is the image of $\mathbf{y}$ in $\Boo^{n}$ under $\mathbf{a}$ and $h$ and $C_{\mathbf{a}, h}$ is the subcube consisting of all $2^{k}$ such images. From the \Cref{defn:random-embedding}, we can derive following two observations:
\begin{enumerate}
    \item For any $\mathbf{y} \in \Boo^{k}$, distribution of $x(\mathbf{y})$ is the uniform distribution over $\Boo^{n}$. This is because $\mathbf{a}$ is uniformly distributed over $\Boo^{n}$.
    \item Fix $\mathbf{y}, \mathbf{y}' \in \Boo^{k}$. Recall that $\delta(\mathbf{y}, \mathbf{y}')$ denotes the fractional Hamming distance between $\mathbf{y}$ and $\mathbf{y}'$. A simple calculation shows the following: For all $i \in [n]$,
    \begin{align*}
        x(\mathbf{y})_{i} \oplus x(\mathbf{y}')_i = \begin{cases}
            0, & \text{with probability } \, 1 - \delta(\mathbf{y}, \mathbf{y}') \\
            1, & \text{with probability } \, \delta(\mathbf{y}, \mathbf{y}').
        \end{cases}
    \end{align*}
    Since this is true for any choice of $x(\mathbf{y}),$ this means that the distribution of the random variable $x(\mathbf{y}) \oplus x(\mathbf{y'})$ is independent of $x(\mathbf{y})$. 
    In particular, using also our observation in the previous item, we see that the pair $(x(\mathbf{y}), x(\mathbf{y}'))$ has the same distribution as $(\mathbf{z}, \mathbf{z}')$ where $\mathbf{z}$ is chosen uniformly at random from $\{0,1\}^n$ and $\mathbf{z}'$ is sampled from the distribution $\mathcal{N}_\rho(\mathbf{z}),$ where $\rho = 1-2\delta(\mathbf{y},\mathbf{y}').$
\end{enumerate}

Building on the above observation, we have the following sampling lemma for subcubes that will be useful at multiple points in the paper.

\begin{lemma}[Sampling lemma for random subcubes]
\label{lemma:sampling-subcube}
Sample $\mathbf{a}$ and $h$ uniformly at random, and let $\mathsf{C} = C_{\mathbf{a}, h}$ be the subcube of dimension $k$ as described in \Cref{defn:random-embedding}. Fix any $T\subseteq \{0,1\}^n$ and let $\mu:= |T|/2^n.$ Then, for any $\varepsilon, \eta > 0$
\begin{align*}
    \Pr_{\mathbf{a},h}\left[\left|\frac{|T\cap \mathsf{C}|}{2^k} - \mu\right| \geq \varepsilon \right] < \eta
\end{align*}
as long as $k\geq \frac{A}{\varepsilon^4\eta^2}\cdot \log\left(\frac{1}{\varepsilon\eta}\right)$ for a large enough absolute constant $A > 0.$
\end{lemma}

\begin{proof}
    The proof is an application of the second moment method with a consequence of the Hypercontractivity theorem (\Cref{thm:hypercontractivity}) being used to bound the variance.

    More formally, for each $\mathbf{y}\in \{0,1\}^k$, let $Z_{\mathbf{y}}\in \{0,1\}$ be the indicator random variable that is $1$ exactly when $x(\mathbf{y})\in T.$ Let $Z$ denote the sum of all $Z_{\mathbf{y}}$ ($\mathbf{y}\in \{0,1\}^n)$. The statement of the lemma may be equivalently stated as 
    \begin{equation}
        \label{eq:sampling-subcube}
        \Pr_{\mathbf{a},h}\left[\left|Z - \mu\cdot 2^k\right| \geq \varepsilon\cdot 2^k \right] < \eta
    \end{equation}
    for $k$ as specified above.

    Since each $x(\mathbf{y})$ is uniformly distributed over $\{0,1\}^n$, it follows that each $Z_{\mathbf{y}}$ is a Bernoulli random variable that is $1$ with probability $\mu$. In particular, the mean of $Z$ is $\mu\cdot 2^k.$

    We now bound the variance of $Z.$ We have, for any $\gamma \in [0,1]$,
    \begin{align*}
        \mathrm{Var}(Z) &= \sum_{\mathbf{y}, \mathbf{y}'} \mathrm{Cov}(Z_{\mathbf{y}}, Z_{\mathbf{y}'})\\
        &= \sum_{\substack{\mathbf{y}, \mathbf{y}':\\ \delta(\mathbf{y}, \mathbf{y}') \in (\frac{1}{2}-\frac{\gamma}{2}, \frac{1}{2}+\frac{\gamma}{2})}} \mathrm{Cov}(Z_{\mathbf{y}}, Z_{\mathbf{y}'}) + \sum_{\substack{\mathbf{y}, \mathbf{y}':\\ \delta(\mathbf{y}, \mathbf{y}') \not\in (\frac{1}{2}-\frac{\gamma}{2}, \frac{1}{2}+\frac{\gamma}{2})}} \mathrm{Cov}(Z_{\mathbf{y}}, Z_{\mathbf{y}'})\\
        &\leq \sum_{\substack{\mathbf{y}, \mathbf{y}':\\ \delta(\mathbf{y}, \mathbf{y}') \in (\frac{1}{2}-\frac{\gamma}{2}, \frac{1}{2}+\frac{\gamma}{2})}} \mathrm{Cov}(Z_{\mathbf{y}}, Z_{\mathbf{y}'}) + \sum_{\substack{\mathbf{y}, \mathbf{y}':\\ \delta(\mathbf{y}, \mathbf{y}') \not\in (\frac{1}{2}-\frac{\gamma}{2}, \frac{1}{2}+\frac{\gamma}{2})}} 1\\
        &\leq \sum_{\substack{\mathbf{y}, \mathbf{y}':\\ \delta(\mathbf{y}, \mathbf{y}') \in (\frac{1}{2}-\frac{\gamma}{2}, \frac{1}{2}+\frac{\gamma}{2})}} \mathrm{Cov}(Z_{\mathbf{y}}, Z_{\mathbf{y}'}) + 2^{2k}\cdot \exp(-\Omega(\gamma^2 k))
    \end{align*}
    where the final inequality is an application of the Chernoff bound. On the other hand, for any $\mathbf{y}, \mathbf{y}'$ such that $\delta(\mathbf{y}, \mathbf{y}') \in (\frac{1}{2}-\frac{\gamma}{2}, \frac{1}{2}+\frac{\gamma}{2})$, we have seen above that the pair $(x(\mathbf{y}),x(\mathbf{y}'))$ have the same distribution as a pair of random variables $(\mathbf{z},\mathbf{z}')$ where $\mathbf{z}$ is chosen uniformly at random from $\{0,1\}^n$ and $\mathbf{z}'$ is sampled from the distribution $\mathcal{N}_\rho(\mathbf{z}),$ where $\rho = 1-2\delta(\mathbf{y},\mathbf{y}')\in [-\gamma,\gamma].$

    By \Cref{thm:hypercontractivity}, we see that for such a pair $(\mathbf{y},\mathbf{y}')$ and for any $\gamma \leq 1/4,$ we have
    \begin{align*}
        \mathrm{Cov}(Z_{\mathbf{y}}, Z_{\mathbf{y}'}) &= 
        \Pr_{\mathbf{a},h}[x(\mathbf{y})\in T \text{ and } x(\mathbf{y}')\in T] -\mu^2\\
        &\leq \mu^{2/1+\gamma} - \mu^2 \leq \mu^{2(1-\gamma)}-\mu^2\\
        &\leq \min\{\mu^{1.5}, \mu^2\cdot (\exp(O(\gamma\cdot \log (1/\mu))-1)\}.
    \end{align*}
    Plugging this into the computation above and setting $\gamma = C_1\cdot \sqrt{\frac{\log k}{k}}$ for a large enough absolute constant $C_1$, we get the following inequalities:
    \begin{align*}
    \mathrm{Var}(Z) &\leq  2^{2k}\cdot \mu^{1.5} + 2^{2k} \cdot \frac{1}{k} \leq 2^{2k}\cdot O\left(\frac{1}{k}\right)\ \ \ \left(\text{if $\mu\leq \frac{1}{k}$}\right)\\
    \mathrm{Var}(Z) &\leq  2^{2k}\cdot \mu^2 \cdot O\left(\sqrt{\frac{\log k}{k}}\cdot \log(1/\mu)\right) + 2^{2k} \cdot \frac{1}{k} \leq 2^{2k}\cdot O\left(\sqrt{\frac{\log k}{k}}\right)\ \ \ \left(\text{if $\mu > \frac{1}{k}$}\right)\\
    \end{align*}
    where we used the fact that $e^x \leq 1+2x$ for $|x|\leq 1/2$ for the first inequality and the fact that $\mu\leq 1$ for the second. 

    Finally, using Chebyshev's inequality, we get
    \begin{align*}
       \Pr_{\mathbf{a},h}\left[\left|Z - \mu\cdot 2^k\right| \geq \varepsilon\cdot 2^k \right] 
       &= \Pr_{\mathbf{a},h}\left[\left|Z - \E[Z]\right| \geq \varepsilon\cdot 2^k \right] \\
       &\leq \frac{\mathrm{Var}(Z)}{\varepsilon^2 2^{2k}}
        \leq \frac{1}{\varepsilon^2}\cdot O\left(\sqrt{\frac{\log k}{k}} \right)< \eta
    \end{align*}
    using the lower bound on $k$ in the statement of the lemma.
\end{proof}

\section{Local Correction in the Unique Decoding Regime}\label{sec:deg-1-decoding}
In this section, we will prove \Cref{thm:uniquedeg1}, i.e. we will give a local correction algorithm for degree $1$ polynomials with only $\Tilde{\bigO}(\log n)$ queries.

We will do this in three steps. 
\begin{itemize}
    \item Step 1: We start by proving a slightly weaker statement: we will give a local correction algorithm that can correct $\mathcal{P}_{1}$ with the error-parameter $\delta \leq 1/\bigO(\log n)$ (see \Cref{thm:uniquedeg1-smallerror}).
    \item Step 2: We will show how to handle some $\delta = \Omega(1)$ by reducing to the small error case (see \Cref{subsec:error-reduction-small-constant} and \Cref{lem:error-reduction-small-constant-main}).
    \item Step 3: Using a similar argument to the second step, we prove \Cref{thm:uniquedeg1}, which is a local correction algorithm with $\delta$ arbitrarily close to the unique decoding radius (see \Cref{subsec:error-reduction-close-radius} and \Cref{lem:error-reduction-main}).
\end{itemize}
The first step works only for linear polynomials, while the latter two reductions also work for higher-degree polynomials.

\subsection{Sub-Constant Error}\label{subsec:sub-constant-error}
In this subsection, we describe the first step towards proving \Cref{thm:uniquedeg1}. We give a local correction algorithm for $\mathcal{P}_{1}$ that can correct for $\delta < 1/\bigO(\log n)$. The main result of this section is the following.\\

\begin{theorem}[Local correction algorithms for $\mathcal{P}_{1}$ up to error $1/\bigO(\log n)$]\label{thm:uniquedeg1-smallerror}
Let $\mathcal{P}_{1}$ be the set of degree $1$ polynomials from $\Boo^{n}$ to $G$. Then $\mathcal{P}_{1}$ has a $(\delta, q)$-local correction algorithm for any $\delta < \bigO(1/\log n)$ and $q =\bigO(\log n)$.\\
\end{theorem}

\begin{remark} In \cite[Theorem 5.3]{bafna2017local}, a lower bound of $q\ge \Omega(\log n/\log log n)$ was shown on the number of queries required to locally correct in the setting where $G$ is the additive group of a field of large characteristic (the lower bound even holds in the regime $\delta < \exp(-n^{\Omega(1)})$). Our theorem above implies that this lower bound is tight up to a $\log \log n$ factor in the setting when $\delta < \bigO(1/\log n)$. In fact, over the reals, we can obtain an upper bound of $q=\bigO(\log n/\log\log n)$, thus matching the lower bound of~\cite{bafna2017local} up to a constant factor. We refer the reader to~\Cref{app:improved-uniquedeg1-smallerror} for this improvement.
\end{remark}

We first describe the general framework of the algorithm, which is applicable more generally. In the following subsection, we will use this framework for linear polynomials and construct a local corrector for $\mathcal{P}_{1}$.

\subsubsection{Framework of Local Correction Algorithm}

We will now give a formal definition of how we construct a local correction algorithm, namely, via a \emph{correction gadget}. This will be useful in the regime where the distance of the input function to the codeword (in our case, a linear polynomial) is small.

Let $\mathcal{F}$ be a class of functions from $\Boo^{n}$ to an Abelian group $G$. Let $P_{1}, \ldots, P_{D}$ be functions from $\Boo^{n}$ to $\mathbb{Z}$ satisfying the following property: for any $P \in \mathcal{F}$, there exist coefficients $\alpha_{1},\ldots,\alpha_{D} \in G$ such that for any $\mathbf{a}\in \{0,1\}^n$
\begin{align*}
    P(\mathbf{a}) = \alpha_{1} P_{1}(\mathbf{a}) + \ldots + \alpha_{D} P_{D}(\mathbf{a}).
\end{align*}
In the case when $G$ is a finite field $\F_{p}$ for a prime $p$ and $\mathcal{F}$ is a vector space of functions, $\set{P_{1},\ldots,P_{D}}$ is a standard spanning set for $\mathcal{F}$ in the linear algebraic sense. We extend this definition to the case when $\mathcal{F}$ is defined over Abelian groups and say that $\{P_1,\ldots,P_D\}$ is a spanning set for $\mathcal{F}.$

\begin{definition}[Local Correction Gadget]\label{defn:decoding-gadget}
Let $\mathcal{F}$ be a set of functions from $\Boo^{n}$ to an Abelian group $G$ with spanning set $\{P_1,\ldots, P_D\}$. For any $\mathbf{a} \in \Boo^{n}$, an $(\varepsilon, q)$-correction gadget for $\mathbf{a}$ is a distribution $\mathcal{D}$ over $(\Boo^{n})^{q}$ satisfying the following two properties:
\begin{enumerate}
    \item There exists $c_{1}, \ldots, c_{q}\in \mathbb{Z}$\footnote{We require that the coefficients are $\mathbb{Z}$ because we are working with Abelian group, and a rational number times a group element is not well defined.} such that for any $(\mathbf{y}^{(1)}, \ldots, \mathbf{y}^{(q)}) \in \mathrm{supp}(\mathcal{D})$, the following holds: for each element of the spanning set $P_j$ ($j\in [D]$).
    \begin{equation}\label{eqn:step1-vectors-condition}
    P_j(\mathbf{a}) = c_{1}P_j(\mathbf{y}^{(1)}) + \ldots + c_{q}P_j(\mathbf{y}^{(q)})
\end{equation}

    \item For any $i \in [q]$, the distribution of $\mathbf{y}^{(i)}$ is $\varepsilon$-close to $U_{n}$.
\end{enumerate}
\end{definition}

The next claim shows that if we have an $(\varepsilon, q)$-correction gadget for sufficiently small $\varepsilon$, that immediately gives us a $(\delta, q)$-local correction algorithm for small enough $\delta$. We will use the same notation in \Cref{defn:decoding-gadget}.

\begin{claim}[Correction gadget gives local correction algorithm]\label{claim:gadget-to-algo}
 If there is an $(\varepsilon, q)$-correction gadget for any $\mathbf{a}\in \{0,1\}^n$ where $q (\delta + \varepsilon) < 1/4$, then there is a $(\delta, q)$-local correction algorithm for $\mathcal{F}$.
\end{claim}
\begin{proof}[Proof of \Cref{claim:gadget-to-algo}]
The existence of a correction gadget for each $\mathbf{a}\in \{0,1\}^n$ gives rise to a natural local correction algorithm for $\mathcal{F}.$ Given access to a function $f:\{0,1\}^n\rightarrow G$ that is promised to be $\delta$-close to some $P\in \mathcal{F}$ and an input $\mathbf{a}\in \{0,1\}^n$, we sample $(\mathbf{y}^{(1)}, \ldots, \mathbf{y}^{(q)})$ from the correction gadget $\mathcal{D}$ for $\mathbf{a}$ and return 
\[
c_1 f(\mathbf{y}^{(1)}) + \cdots + c_q f(\mathbf{y}^{(q)})
\]
where $c_1,\ldots, c_q$ are the coefficients corresponding to the correction gadget.

Since $P_1,\ldots, P_D$ form a spanning set for $\mathcal{F}$, it follows from \Cref{eqn:step1-vectors-condition} and linearity that for any $P\in \mathcal{F}$ and any $\mathbf{a}\in \{0,1\}^n$
\[
P(\mathbf{a}) = c_{1}P(\mathbf{y}^{(1)}) + \ldots + c_{q}P(\mathbf{y}^{(q)}).
\]
In particular, the correction algorithm outputs the correct answer $P(\mathbf{a})$ as long as $f$ agrees with $P$ on \textit{each} of $\mathbf{y}^{(1)}, \ldots, \mathbf{y}^{(q)}$. We now upper bound the probability of the correction algorithm outputting an incorrect value.

For any $i \in [q]$, the distribution of $\mathbf{y}^{(i)}$ is $\varepsilon$-close to $U_{n}$. In other words, for any set $T \subseteq \Boo^{n}$,
\begin{align*}
    \left|\Pr_{(\mathbf{y}^{(1)}, \ldots, \mathbf{y}^{(q)}) \sim \mathcal{D}}[\mathbf{y}^{(i)} \in T] - \Pr_{\mathbf{y}^{(i)} \sim U_{n}}[\mathbf{y}^{(i)} \in T]\right| \leq \varepsilon
\end{align*}
If $T$ is the set of points where $f$ and $P$ disagree, i.e. $|T|/2^{n} \leq \delta$, then we have,
\begin{align*}
    \Pr_{(\mathbf{y}^{(1)}, \ldots, \mathbf{y}^{(q)}) \sim \mathcal{D}}[f(\mathbf{y}^{(i)}) \; \text{is incorrect}] \leq \delta + \varepsilon
\end{align*}
Thus the probability that $f$ is incorrect on at least one of $\mathbf{y}^{(1)}, \ldots, \mathbf{y}^{(q)}$, when $(\mathbf{y}^{(1)}, \ldots, \mathbf{y}^{(q)}) \sim \mathcal{D}$ is at most $q (\delta + \varepsilon) < 1/4$.
\end{proof}

In this subsection, we defined a local correction gadget, which is a distribution with suitable properties. In \Cref{claim:gadget-to-algo}, we showed that to construct a local correction algorithm, it suffices to construct a local correction gadget. In the following subsections, we will focus on constructing local correction gadgets with, and then \Cref{claim:gadget-to-algo} would imply that we get a local correction algorithm too.

\subsubsection{Local Correction Algorithm for Linear Polynomials}

We now prove \Cref{thm:uniquedeg1-smallerror}. The main technical step in the proof is the following lemma, which is to construct a local correction gadget for $1^{n}$.\\

\begin{lemma}[Correction gadget for $1^n$]\label{lemma:decode-1n}
Fix any odd positive integer $q$. For any $n$, there is a choice of $c_1,\ldots, c_q \in \mathbb{Z}$ and a distribution $\mc{D}$ over $(\{0,1\}^n)^q$ such that the following properties hold for $c_1,\ldots, c_q$ and  any sample $(\mathbf{y}^{(1)},\ldots, \mathbf{y}^{(q)})$ from $\mathcal{D}$.
\begin{itemize}
    \item $c_{1} + \ldots + c_{q} = 1$ and for all $i \in [n]$,
    \begin{align*}
        c_{1}y^{(1)}_{i} + \ldots + c_{q}y^{(q)}_{i} = 1
    \end{align*}
    \item For each $j \in [q]$, $\mathbf{y}^{(j)}$ is $(1/2^{\Omega(q)}\cdot \sqrt{n})$-close to the $U_{n}$.\\
\end{itemize}
\end{lemma}

We first show how to prove \Cref{thm:uniquedeg1-smallerror} assuming this lemma. The idea is that since the space of linear polynomials $\mathcal{P}_{1}$ is closed under affine-shift, we can shift the query points to correct any point $\mathbf{a}$. \Cref{lemma:decode-1n} is proved subsequently.

\begin{proof}[Proof of \Cref{thm:uniquedeg1-smallerror}]
The space $\mathcal{P}_1$ of linear polynomials over $G$ has as a spanning set the constant function $P_0(\mathbf{x}) = 1$ and the co-ordinate functions $P_j(\mathbf{x}) = x_j$ for each $j\in [n].$

From \Cref{claim:gadget-to-algo}, it suffices to give a $(\varepsilon, q)$-correction gadget for any $\mathbf{a} \in \Boo^{n}$, where $\varepsilon = 1/n$. Note that \Cref{lemma:decode-1n} directly yields a correction gadget $\mathcal{D}$ at the point $1^n$ for $q = \bigO(\log n).$

To get a correction gadget at a point $\mathbf{a}\neq 1^n$, we simply shift this correction gadget by $\mathbf{b} = 1^n \oplus \mathbf{a}$ and use the fact that the space of linear polynomials is preserved by such shifts.

More precisely, consider the distribution $\mc{D}_\mathbf{b}$ obtained by sampling $(\mathbf{y}^{(1)},\ldots, \mathbf{y}^{(q)})$ from $\mathcal{D}$ and shifting each element by $\mathbf{b}$ to get 
\[
(\mathbf{z}^{(1)},\ldots, \mathbf{z}^{(q)}) = (\mathbf{y}^{(1)}\oplus \mathbf{b},\ldots, \mathbf{y}^{(q)}\oplus\mathbf{b}).
\]
We retain the same coefficients $c_1,\ldots, c_q$ as in \Cref{lemma:decode-1n}.

To prove that $(\mathbf{z}^{(1)},\ldots, \mathbf{z}^{(q)})$ is an $(\varepsilon,q)$-correction gadget for $\mathbf{a}$, it remains to verify
\begin{align}
    c_1+\cdots + c_q &= 1\notag\\
    c_1 z^{(1)}_i + \cdots + c_q z^{(q)}_i &= a_i \text{ for each $i\in [n]$}\label{eq:shift-gadget}
\end{align}
The first of the above follows from \Cref{lemma:decode-1n}. The second equality \Cref{eq:shift-gadget} is also easily verified for $i$ such that $a_i = 1$ since $z^{(j)}_i = y^{(j)}_i$ in this case. For $i$ such that $a_i = 0$, we see that $z^{(j)}_i = 1-y^{(j)}_i$ for each $j\in [q]$ and hence
\[
\sum_{j\in [q]} c_j z^{(j)}_i = \sum_{j\in [q]}c_j - \sum_{j\in [q]} c_j y^{(j)}_i = 1-1 = 0 = a_i.
\]
We have thus shown that \Cref{eq:shift-gadget} holds for all $i\in [n]$. Further, since $\mathbf{y}^{(j)}$ is $1/n$-close to uniform for each $j\in [q]$, so is $\mathbf{z}^{(j)}.$ Overall, this implies that $\mathcal{D}_\mathbf{b}$ is a correction gadget for $\mathbf{a}.$

\Cref{claim:gadget-to-algo} then gives us the desired local correction algorithm.
\end{proof}
So now we have shown that constructing a local correction gadget for $1^{n}$ is sufficient to get a local correction gadget for any $\mathbf{a}$. In the next subsection, we give a local correction gadget for $1^{n}$.

\subsubsection{Correction Gadget for all 1s Vector}
In this subsection, we prove \Cref{lemma:decode-1n}. We first construct a Boolean matrix with some interesting combinatorial and algebraic properties. The distribution $\mathcal{D}$ in \Cref{lemma:decode-1n} is obtained later by sampling $n$ rows of this matrix independently and uniformly at random.

The heart of our construction of a local correction gadget is the following technical lemma. It shows that we can find a small number of nearly balanced Boolean vectors, whose integer span contains the all 1s vector.

\begin{lemma}[Construction of a matrix]\label{lemma:matrix-construction}
For any natural number $k$, there exists an integer matrix $A_{k}$ of dimension $(2^{k}-1) \times (2k-1)$ with entries in $\Boo$ and a vector $\mathbf{c} \in \mathbb{Z}^{2k-1}$ such that $A_{k} \mathbf{c} = 1^{2^{k}-1}$ and there is exactly one row in $A_{k}$ that is $(1,\ldots,1)$. Additionally, for any column of $A_{k}$, the Hamming weight of the column is in $[2^{k-1} - 1, 2^{k-1} + 1]$.
\end{lemma}

\begin{remark}
\label{rem:matrix-construction}
 The statement of this lemma is, in some sense, the best that can we hope for as the lemma does not hold if each column is required to be \emph{perfectly balanced}. In fact, the above lemma does not hold even in the setting where each column is required to have weight \emph{exactly} $w$ for some $w<2^{k}-1$: in this case, $1^{2^k-1}$ would not even be in the $\mathbb{Q}$-linear span of the columns of $A_k$.\footnote{Consider a vector $\mathbf{v}\in \mathbb{Q}^{2^k-1}$ the entries of which are $1-1/w$ or $-1/w$, depending on whether the corresponding row in $A_k$ is the all $1$s row or not. The vector $\mathbf{v}$ is orthogonal to the columns of $A_k$ but not the vector $1^{2^k-1}$.} 

Quantitatively, this lemma exhibits a near-tight converse to a lemma of Bafna, Srinivasan, and Sudan~\cite{bafna2017local} who showed that for any $n\times k$ Boolean matrix with an all-$1$s row, and columns that have Hamming weights in the range $[n/2-\sqrt{n}, n/2 + \sqrt{n}]$ and also span the all $1$s column, we must have $k= \tilde{\Omega}(\log n).$
\end{remark}

\begin{proof}
Fix a $k \in \mathbb{N}$. Given a non-negative integer $i < 2^k$, we denote by $\mathrm{bin}(i)$ the Boolean vector that denotes the $k$-bit binary expansion of $i$ (with the first entry being the most significant bit).

\paragraph{\underline{Defining the base matrix}}Let $M$ be a $(2^{k}-1) \times 2k$ matrix with entries in $\Boo$. For all $i \in [2^{k} - 1]$ and $j \in [2k]$, let the $i^{th}$ row and the $j^{th}$ column of $M$ be denoted by $\mathsf{row}^{(i)}$ and $\mathsf{col}^{(j)}$, respectively. The $i^{th}$ row of $M$ is $\mathsf{row}^{(i)} := (\mathrm{bin}(i) \, \mathrm{bin}(i-1))$, i.e. in $\mathsf{row}^{(i)}$, the first $k$ coordinates are $\mathrm{bin}(i)$ and the next $k$ entries are $\mathrm{bin}(i-1)$, where for an integer $i$, $\mathrm{bin}(i)$ denotes the binary representation of $i$.
\begin{align*}
    M = \begin{bmatrix}
        \vdots & \vdots \\
        \mathrm{bin}(i) & \mathrm{bin}(i-1) \\
        \vdots & \vdots
    \end{bmatrix}_{(2^{k}-1) \times 2k}
\end{align*}
Let $\mathbf{w} \in \R^{2k}$ be the following vector:
\begin{align*}
    \mathbf{w} = \paren{2^{k-1},\ldots,2^{1},2^{0}, \; -2^{k-1},\ldots,-2^{1},-2^{0}}
\end{align*}
It is easy to see that for any row $\mathsf{row}^{(i)}$ of $M$, $\langle \mathsf{row}^{(i)}, \mathbf{w} \rangle = i - (i-1) = 1$. Thus, $M \mathbf{w} = 1^{2^{k}-1}$.

\paragraph{\underline{A useful observation}} For any row $\mathsf{row}^{(i)}$, the $k^{th}$ and the $2k^{th}$ entry are distinct, i.e. $\mathsf{row}^{(i)}_{k} \oplus \mathsf{row}^{(i)}_{2k} = 1$, i.e. $\mathsf{col}^{(k)} = 1^{2^{k}-1} - \mathsf{col}^{(2k)}$.

\paragraph{\underline{Modifying the base matrix}}Let $\Tilde{M}$ be a $(2^{k}-1) \times 2k$ matrix and $\Tilde{\mathbf{w}}$ be a column vector of dimension $2k$. Let the $i^{th}$ row and the $j^{th}$ column of $\Tilde{M}$ be denoted by $\widetilde{\mathsf{row}}^{(i)}$ and $\widetilde{\mathsf{col}}^{(j)}$, respectively. $\Tilde{M}$ and $\Tilde{\mathbf{w}}$ are defined as follows:
\begin{align*}
    \widetilde{\mathsf{col}}^{(j)} = \begin{cases}
        1 - \mathsf{col}^{(j)}, & \text{if} \, j \neq k \\
        \mathsf{col}^{(j)}, & \text{if} \, j = k
    \end{cases} \quad \quad
        \Tilde{w}_{j} = \begin{cases}
        w_{j}, & \text{if} \, j \neq k \\
        -w_{j}, & \text{if} \, j = k
    \end{cases}
\end{align*}
It is easy to verify the following: for any $i \in [2^{k}-1]$, $\langle \widetilde{\mathsf{row}}^{(i)} , \Tilde{\mathbf{w}} \rangle = -2$. Thus $\Tilde{M} (-\Tilde{\mathbf{w}}/2) = 1^{2^{k}-1}$.\newline
Note that $\widetilde{\mathsf{col}}^{(k)} = \widetilde{\mathsf{col}}^{(2k)}$. 
The first row of $M$, i.e. $\mathsf{row}^{(1)} = (\mathrm{bin}(1)\mathrm{bin}(0)) = (0,\ldots,0,1,0,\ldots,0)$. The first row of $\Tilde{M}$. i.e. $\widetilde{row}^{(1)} = (1,\ldots,1)$. Since $\Tilde{M} (-\Tilde{\mathbf{w}}/2) = 1^{2^{k}-1}$, this implies that $\sum_{j = 1}^{2k} (-\Tilde{w}_{j}/2) = 1$.

It's also easy to verify that no row other than the first row of $\tilde{M}$ is $(1,1,\ldots,1).$

\paragraph{\underline{Integral coefficients}}We have $-\Tilde{w}_{k}/2 = -\Tilde{w}_{2k}/2 = 1/2$. Consider any row $\widetilde{\mathsf{row}}^{(i)}$ of $\Tilde{M}$. Since $\widetilde{\mathsf{row}}^{(i)}_{k} = \widetilde{\mathsf{row}}^{(i)}_{2k}$, the following equality holds:
\begin{equation}\label{eqn:unique-deg-1-obs}
    \widetilde{\mathsf{row}}^{(i)}_{k} (-\Tilde{w}_{k}/2) \, + \, \widetilde{\mathsf{row}}^{(i)}_{2k} (-\Tilde{w}_{2k}/2) \; = \; \widetilde{\mathsf{row}}^{(i)}_{k} \cdot 1 \, + \, \widetilde{\mathsf{row}}^{(i)}_{2k} \cdot 0
\end{equation}
Let $\mathbf{c} \in \mathbb{Z}^{2k-1}$ be the following vector: $c_{j} = (-\Tilde{w}_{j}/2)$ if $j \neq k$, otherwise $c_{j} = 1$. For any row $\widetilde{\mathsf{row}}^{(i)}$, from \Cref{eqn:unique-deg-1-obs}, $\langle \widetilde{\mathsf{row}}^{(i)}, \mathbf{c} \rangle = \langle \widetilde{\mathsf{row}}^{(i)} , (-\Tilde{\mathbf{w}}/2) \rangle = 1$. Let $A_{k}$ denote the matrix $\Tilde{M}$ after removing the $2k^{th}$ column.
Then $A_k \mathbf{c} = 1^{2^{k}-1}$.\newline
Since $\sum_{j = 1}^{2k} (-\Tilde{w}_{j}/2) = 1$, using \Cref{eqn:unique-deg-1-obs}, we get that $\sum_{j=1}^{2k-1} c_{j} = 1$.\\

\paragraph{\underline{Columns are nearly balanced}}Finally, we will prove that for each column $\widetilde{\mathsf{col}}^{(j)}$ of $A$, the Hamming weight of $\widetilde{\mathsf{col}}^{(j)} \in [2^{k-1} - 1, 2^{k-1} + 1]$. For any $j \in [2k-1]$, the Hamming weight of $\mathsf{col}^{(j)}$ is in $\set{2^{k-1} - 1, 2^{k-1} + 1}$. This is because if $M$ had an additional row $[\mathrm{bin}(0)\mathrm{bin}(2^{k}-1)]$, then each column of $M$ would be exactly balanced, i.e. have Hamming weight of $2^{k-1}$. Then by definition of $\widetilde{\mathsf{col}}^{(j)}$, it follows that Hamming weight of each column of $A_k$ is also in $[2^{k-1} - 1, 2^{k-1} + 1]$.
\end{proof}
Next, we are going to describe a distribution $\mathcal{D}$ on $(\Boo^{m})^{q}$, where $m = 2^{k}-1$ and $q = 2k-1$. We will do this by randomly sampling rows of the matrix $A_{k}$ given by \Cref{lemma:matrix-construction}. This will give us a local correction gadget and finish the proof of \Cref{lemma:decode-1n}.

\begin{proof}[Proof of \Cref{lemma:decode-1n}]
Assume that $q=2k-1.$ To sample $(\mathbf{y}^{(1)},\ldots, \mathbf{y}^{(q)})\sim \mathcal{D}$ over $(\{0,1\}^n)^q$, we sample $n$ rows independenly and uniformly at random from the rows of $A_k$ as constructed in \Cref{lemma:matrix-construction} and define $(y^{(1)}_i,\ldots, y^{(q)}_i)$ to be the $i$th sample for each $i\in [n].$

We now show that $\mathcal{D}$ has the required properties from the statement of \Cref{lemma:decode-1n}.

Let $(c_1,\ldots, c_q) = \mathbf{c}$ be as guaranteed by \Cref{lemma:matrix-construction}.

The first property holds from the properties of $A_k$ and $\mathbf{c}$. For each $i\in [n]$, the vector $(y^{(1)}_i,\ldots, y^{(q)}_i)$ is a row of $A_{k}$, and from \Cref{lemma:matrix-construction}, we know that the inner product of any row of $A_{k}$ and $\mathbf{c}$ is $1$. 
Further, since $1^q$ is also a row of $A_k$, it follows that the entries of $\mathbf{c}$ sum to $1$.

The second property follows from the fact that each column of $A_{k}$ has relative Hamming weight in the range $[\frac{1}{2} - 2^{-k}, \frac{1}{2}+2^{-k}]$. Thus, for any fixed $j\in [q]$ and each $i\in [n]$, we have
\[
\prob{}{y^{(j)}_i = 1} \in \left[\frac{1}{2}-\frac{1}{2^k}, \frac{1}{2}+\frac{1}{2^k}\right].
\]
Since for a fixed $j\in [q]$ the bits $\{y^{(j)}_i\ |\ i\in [n]\}$ are mutually independent, we are now done by the following standard fact (which can easily be proved by, say, following the proof of~\cite[Theorem 5.5, Claim 5.6]{MansourLec}).

\begin{fact}\label{fact:close-to-uniform}
Let $\eta > 0$. Let $\mathcal{D'}$ be a distribution on $\Boo^{n}$ such that for any $\mathbf{y} \sim \mathcal{D'}$, the co-ordinates of $\mathbf{y}$ are independent and for all $i \in [n]$,
\begin{align*}
    1/2 - \eta \leq \Pr[y_{i} = 1] \leq 1/2 + \eta.
\end{align*}
Then $\mathcal{D'}$ is $\bigO(\eta \sqrt{n})$-close to $U_{n}$.
\end{fact}
This concludes the proof of \Cref{lemma:decode-1n}.
\end{proof}
Summarising the proof of \Cref{thm:uniquedeg1-smallerror} - We first showed that it is enough to focus on constructing local correction gadgets (see \Cref{claim:gadget-to-algo}) and we can assume without loss of generality that we want to decode at $1^{n}$. Then we constructed a matrix with nice properties (see \Cref{lemma:matrix-construction}) and defined a distribution for local correction gadget using this matrix.

\subsection{Constant Error Algorithm via Error-Reduction}\label{subsec:error-reduction-small-constant}
In this subsection, we explain the second step towards proving \Cref{thm:uniquedeg1}. We show how to locally correct degree-$1$ polynomials in the regime of \emph{constant} error (one can think of this error to be around $1/1000$). We will do this by reducing the problem to the case of low error (sub-constant error). The results of this section also work for higher-degree polynomials. 

We will show that there is a randomized algorithm $\mathcal{A}^{f}$ that given oracle access to any function $f$ that is $\delta$-close to a degree-$d$ polynomial $P$ (think of $\delta$ as being a small enough constant depending on $d$), has the following property: with high probability over the internal randomness of $\mathcal{A}^{f}$, the function computed by  $\mathcal{A}^{f}$ is $\eta$-close to $P$, where $\eta < \delta$. We state it formally below.\\

\begin{lemma}[Error reduction for constant error]\label{lem:error-reduction-small-constant-main}
Fix any Abelian group $G$ and a positive integer $d$. The following holds for $\delta < 1/2^{\bigO(d)}$ and $K = 2^{\bigO(d)}$ where the $\bigO(\cdot)$ hides a large enough absolute constant.\newline
For any $\eta, \delta$, where $\eta < \delta$, there exists a randomized algorithm $\mathcal{A}$ with the following properties: Let $f: \Boo^{n} \to G$ be a function and let $P: \Boo^{n} \to G$ be a degree-$d$ polynomial such that $\delta(f,P) \leq \delta$, and let $\mathcal{A}^{f}$ denotes that $\mathcal{A}$ has oracle access to $f$, then
\begin{align*}
    \Pr[\delta(\mathcal{A}^{f}, P) > \eta] < 1/10,
\end{align*}
where the above probability is over the internal randomness of $\mathcal{A}^{f}$. Further, for every $\mathbf{x} \in \Boo^{n}$, $\mathcal{A}^{f}$ makes $K^{T}$ queries to $f$ and $T = \bigO\paren{  \log\paren{ \dfrac{\log(1/\eta)}{\log(1/\delta)} }  }  $.\\
\end{lemma}

\noindent
Putting this together with \Cref{thm:uniquedeg1-smallerror}, we immediately get the following algorithmic result, which is the main result of this subsection.\\

\begin{theorem}[Unique local correction algorithm for constant error]
    \label{thm:uniquedecoding-const-error} 
    Fix any Abelian group $G$. The space $\mathcal{P}_1$ of degree-$1$ polynomials has a $(\delta,q)$-local correction algorithm where $\delta > 0$ is a small enough absolute constant and $q = \bigO(\log n\cdot \poly(\log \log n)).$
\end{theorem}

\begin{proof}
Given oracle access to a function $f$ that is $\delta$-close to a degree-$1$ polynomial $P$, \Cref{lem:error-reduction-small-constant-main} (with $\eta  = o(1/\log n)$) shows how to get access to a randomized oracle $\mathcal{A}^f$ that makes $\poly(\log\log n)$ queries to $f$ is $\eta$-close to $P$ except with small probability. We apply the local correction algorithm from \Cref{thm:uniquedeg1-smallerror} with oracle access to $\mathcal{A}^f$, repeating a constant number of times to reduce the error down to $1/10.$ The latter algorithm works for every choice of the internal randomness of $\mathcal{A}^f$ such that $\mathcal{A}^f$ and $P$ are $\eta$-close. This gives us an overall error probability of
\[
\Pr[\delta(\mathcal{A}^f,P) > \eta] + 1/10 \leq 1/10 + 1/10 < 1/4,
\]
as desired. The query complexity of this algorithm is the product of the query complexities of $\mathcal{A}^f$ and the algorithm from \Cref{thm:uniquedeg1-smallerror}.
\end{proof}

In the rest of this subsection, we will prove \Cref{lem:error-reduction-small-constant-main}. The algorithm $\mathcal{A}^{f}$ in \Cref{lem:error-reduction-small-constant-main} will be a recursive algorithm. Each recursive iteration of the algorithm $\mathcal{A}^{f}$ uses the same `base algorithm' $\mathcal{B}$, which will be the core of our error reduction algorithm from small constant error. In the next lemma, we formally state the properties of the base algorithm.\\ 

\begin{lemma}[Base Error Reduction Algorithm]\label{lem:error-reduction-subroutine}
Fix any Abelian group $G$ and a positive integer $d$. The following holds for $K = 2^{O(d)}$.  For any $0 < \gamma < 1$, there exists a randomized algorithm $\mathcal{B}$ with the following properties: Let $g: \Boo^{n} \to G$ be a function and let $P: \Boo^{n} \to G$ be a degree-$d$ polynomial such that $\delta(g,P) \leq \gamma$, and let $\mathcal{B}^{g}$ denotes that $\mathcal{B}$ has oracle access to $g$, then
\begin{align*}
    \mathbb{E}[\delta(\mathcal{B}^{g}, P)] < O(K^2)\cdot \gamma^{1.5}
\end{align*}
where the above expectation is over the internal randomness of $\mathcal{B}$. Further, for every $\mathbf{x} \in \Boo^{n}$, $\mathcal{A}^{g}$ makes $K$ queries to $g$.\\  
\end{lemma}

We defer the construction of the base algorithm and proof of \Cref{lem:error-reduction-subroutine} to the next subsection, \Cref{subsubsec:base-error-reduction}. For now, we assume \Cref{lem:error-reduction-subroutine} and proceed to describe the recursive construction of $\mathcal{A}^{f}$ and prove \Cref{lem:error-reduction-small-constant-main}.

\begin{proof}[Proof of \Cref{lem:error-reduction-small-constant-main}]
Let $\mathcal{B}$ be the algorithm given by \Cref{lem:error-reduction-subroutine}. We define a sequence of algorithms $\mathcal{A}_0^f, \mathcal{A}_1^f, \ldots,$ as follows. 

\begin{algobox}
The algorithm $\mathcal{A}^f_t$ computes a function mapping inputs in $\{0,1\}^n$ along with a uniformly random string from $\{0,1\}^{r_t}$ to a random group element in $G$ ($t$ denotes the number of recursive calls).\\

\noindent
\begin{itemize}
    \item $\mathcal{A}_0^f$ just computes the function $f.$ (In particular, $r_0 = 0$.)
    \item For each $t > 0,$ we inductively define $r_t = r_{t-1}+ r$, where $r$ is the amount of randomness required by the base error reduction algorithm $\mathcal{B}$.\newline
    On input $\mathbf{x}$ and random string $\sigma_t \sim U_{r_t}$, the algorithm $\mathcal{A}_t^f$ algorithm runs the algorithm $\mc{B}$ on $\mathbf{x}$ using the first $r$ bits of $\sigma_t$ as its source of randomness, and with oracle access to $\mathcal{A}_{t-1}^f$ using the remaining $r_{t-1}$ bits of $\sigma_t$ as randomness.\\
\end{itemize}

\noindent
The algorithm $\mathcal{A}^f$ will be $\mathcal{A}_T^f$ for $T = C\cdot \log\paren{ \dfrac{\log(1/\eta)}{\log(1/\delta)} } $ where $C$ is a large enough absolute constant chosen below. 
\end{algobox}

\textbf{Query complexity:} An easy inductive argument shows that $\mathcal{A}^f$ makes at most $K^T$ queries to $f.$

\textbf{Error probability:} We now analyze the error made by the above algorithms. We will argue inductively that for each $t\leq T$ and  $\delta_t := \delta^{(1.1)^t}$, we have 
\begin{equation}
    \label{eq:algoAindn}
    \Pr_{\sigma_t}[\, \underbrace{\delta(\mathcal{A}^{f}_t(\cdot, \sigma_t), P) > \delta_t \, }_{:= \, \mathcal{E}_t}] \; \leq \; \sum_{j=1}^t \frac{1}{100^j} \; < \; \frac{1}{10}.
\end{equation}
In the inductive proof, we will need that $\delta_0 = \delta <2^{-C_1\cdot d}$ for a large enough absolute constant $C_1.$

We now proceed with the induction. The base case ($t = 0$) is trivial as $\delta(\mathcal{A}^{f}_t, P) = \delta_0$ by definition.

Now assume that $t > 1$. We decompose the random string $\sigma_t$ into its first $r$ bits, denoted $\sigma$, and its last $r_{t-1}$ bits, denoted $\sigma_{t-1}.$ We bound the probability in \Cref{eq:algoAindn} as follows. (Note that the event $\mathcal{E}_{t-1}$ below only depends on $\sigma_{t-1}.$)
\begin{equation}
\label{eq:algoAindn-t-1}
\Pr_{\sigma_t}[\mathcal{E}_t] \, \leq \, \Pr_{\sigma_{t-1}}[\mathcal{E}_{t-1}] + \Pr_{\sigma_t}[\mathcal{E}_t\ |\ \neg\mathcal{E}_{t-1}] \, \leq \, \sum_{j=1}^{t-1} \frac{1}{100^j} + \Pr_{\sigma_t}[\mathcal{E}_t\ |\ \neg\mathcal{E}_{t-1}]
\end{equation}
where we used the induction hypothesis for the second inequality. To bound $\Pr_{\sigma_t}[\mathcal{E}_t\ |\ \neg\mathcal{E}_{t-1}]$, fix any choice of $\sigma_{t-1}$ so that $\neg\mathcal{E}_{t-1}$ holds, i.e. so that $\delta(\mathcal{A}^{f}_{t-1}, P) \leq \delta_{t-1}$. By the guarantee on $\mathcal{B}$, i.e. \Cref{lem:error-reduction-subroutine}, we know that
\[
\mathbb{E}_{\sigma}[\delta(\mathcal{A}^{f}_t(\cdot, \sigma_t), P)] < \bigO(K^2)\cdot \gamma^{1.5},
\]
where $\gamma = \delta(\mathcal{A}^f_{t-1}(\cdot, \sigma_{t-1}),P)$. Substituting it above, we get,
\begin{align*}
    \mathbb{E}_{\sigma}[\delta(\mathcal{A}^{f}_t(\cdot, \sigma_t), P)] \, \leq \, \bigO(K^2) \cdot \delta_{t-1}^{1.5} \, \leq \, \delta_{t-1}^{1.25}
\end{align*}
where for the final inequality, we use the fact that
\[
\bigO(K^2)\cdot \delta_{t-1}^{0.25} \leq \bigO(K^2)\cdot \delta_0^{0.25} \leq 1
\]
as long  as $\delta_0 = \delta \leq 2^{-C_1 d}$ for a large enough constant $C_1.$ Continuing the above computation, we see that by Markov's inequality
\[
\Pr_{\sigma}[\mathcal{E}_t] \, \leq \, \frac{\delta_{t-1}^{1.25}}{\delta_t} \; = \; \delta^{\Omega((1.1)^t)} \, \leq \, \frac{1}{100^t}
\]
where the final inequality holds for all $t$ as long as $\delta \leq 2^{-C_1 d}$ for a large enough constant $C_1.$ Since this inequality holds for any choice of $\sigma_{t-1}$ so that $\neg\mathcal{E}_{t-1}$ holds, we can plug this bound into \Cref{eq:algoAindn-t-1} to finish the inductive case of \Cref{eq:algoAindn}.

Setting $T = C\cdot \log\paren{ \dfrac{\log(1/\eta)}{\log(1/\delta)} } $ for a large enough constant $C$, we see that $\delta_T < \eta.$ In this case, \Cref{eq:algoAindn} implies the required bound on the error probability of $\mathcal{A}^f.$
\end{proof}
Thus we have shown so far that given the base algorithm $\mathcal{B}$, we do get an error reduction algorithm from small constant error to error $\bigO(1/\log n)$. Now it remains to describe the base error reduction algorithm.
In the next subsection, we describe the base algorithm $\mathcal{B}$ and prove \Cref{lem:error-reduction-subroutine}.

\subsubsection{The Base Algorithm and its Analysis}\label{subsubsec:base-error-reduction}
In this section, we prove \Cref{lem:error-reduction-subroutine}, which will then complete the proof of the error reduction algorithm from small constant to sub-constant error (see \Cref{lem:error-reduction-small-constant-main}). Before we describe $\mathcal{B}$, we will define an \textit{error reduction gadget}, which is a variant of the local correction gadget defined previously (\Cref{defn:decoding-gadget}). 

\begin{definition}[Error-reduction Gadget for $\mathcal{P}_d$]\label{defn:error-reduction-gadget}
 For $\rho\in (0,1)$, an $(\rho,q)$-error reduction gadget for $\mathcal{P}_d$ is a distribution $\mathcal{D}$ over $(\Boo^{n})^{q}$ satisfying the following two properties:
\begin{enumerate}
    \item There exists $c_{1}, \ldots, c_{q}\in \mathbb{Z}$ such that for any $(\mathbf{y}^{(1)}, \ldots, \mathbf{y}^{(q)}) \in \mathrm{supp}(\mathcal{D})$, the following holds true for each $P\in \mathcal{P}_d$ and each $\mathbf{a}\in \{0,1\}^n$
    \begin{equation}
    \label{eq:error-reduction-gadget}
    P(\mathbf{a}) = c_{1}P(\mathbf{a}\oplus \mathbf{y}^{(1)}) + \ldots + c_{q}P(\mathbf{a}\oplus\mathbf{y}^{(q)}).
\end{equation}

    \item For any $i \in [q]$, the bits of $\mathbf{y}^{(i)}$ are i.i.d. Bernoulli random variables that are $\rho$-close to uniform. Equivalently, each co-ordinate is $1$ with probability $p_i \in [\frac{1-\rho}{2}, \frac{1+\rho}{2}].$
\end{enumerate}
\end{definition}

To describe the base algorithm $\mathcal{B}$ and prove \Cref{lem:error-reduction-subroutine}, we need an error-reduction gadget for $\mathcal{P}_d$, the space of degree-$d$ polynomials over a group $G.$ The next lemma says that there exists an error-reduction gadget for $\mathcal{P}_{d}$, with small number of queries for constant $d$ and $\rho$.

\begin{lemma}[Constructing an error-reduction gadget for $\mathcal{P}_d$]
\label{lem:error-reduction-gadget}
    Fix any Abelian group $G$ and any $\rho >0.$ Then $\mathcal{P}_d$ has a $(\rho,q)$-error-reduction gadget where $q = 2^{O(d/\rho)}.$
\end{lemma}

Assuming the existence of error-reduction gagdet through the above lemma, we first finish the proof of \Cref{lem:error-reduction-subroutine}.  We prove \Cref{lem:error-reduction-gadget} subsequently.

The idea is as follows. In base algorithm, we use the error-reduction gadget to correct the polynomial at a \emph{random point} $\mathbf{a}\in \{0,1\}^n$. This process is likely to give the right answer except with probability $q\gamma$ since, after shifting, each query is now \emph{uniformly} distributed, and hence the chance that any of the queried points is an error point of $g$ is at most $\gamma$. We reduce the error by repeating this process three times and taking a majority vote. To analyze this algorithm, we need to understand the probability that two iterations of this process both evaluate $g$ at an error point. We do this using hypercontractivity (more specifically \Cref{thm:hypercontractivity}).

\begin{proof}[Proof of \Cref{lem:error-reduction-subroutine}]
    Let $\mathcal{D}$ be a $(1/10,q)$-error-reduction gadget as given by \Cref{lem:error-reduction-gadget}. The algorithm $\mathcal{B}$, given oracle access to $g:\{0,1\}^n\rightarrow G$ and $\mathbf{a}\in \{0,1\}^n$, does the following.

\begin{itemize}
    \item Repeat the following three times independently. Sample $(\mathbf{y}^{(1)}, \ldots, \mathbf{y}^{(q)})$ from $\mathcal{D}$ and compute 
    \[
    c_1 g(\mathbf{a} \oplus \mathbf{y}^{(1)}) +  \cdots + c_qg(\mathbf{a} \oplus \mathbf{y}^{(q)})
    \]
    where $c_1,\ldots, c_q$ are the coefficients corresponding to the error-reduction gadget.
    \item Output the plurality among the three group elements $b_1,b_2,b_3$ computed above.
\end{itemize}

The number of queries made by the algorithm is $K = O(q) = 2^{O(d)}$ as claimed. So it only remains to analyze $\delta(\mathcal{B}^g,P)$. From now on, let $\mathbf{a}$ be a uniformly random input in $\{0,1\}^n.$

For $i\in \{1,2,3\},$ let $\mathcal{E}_i$ denote the event that $b_i \neq P(\mathbf{a}).$ We have
\[
\E[\delta(\mathcal{B}^g,P)] \; = \; \Pr[\mathcal{B}^g(\mathbf{a})\neq g(\mathbf{a})] \; \leq \; \Pr[\mathcal{E}_1 \wedge \mathcal{E}_2] \, + \, \Pr[\mathcal{E}_2 \wedge \mathcal{E}_3] \, + \, \Pr[\mathcal{E}_1 \wedge \mathcal{E}_3],
\]
where the above probability is over the randomness of $\mathcal{B}$. It therefore suffices to show that each of the three terms in the final expression above is at most $O(q^2)\cdot \gamma^{1.5}.$

Without loss of generality, consider the event $\mathcal{E}_1\wedge \mathcal{E}_2$. Let $(\mathbf{y}^{(1)}, \ldots, \mathbf{y}^{(q)})$ and $(\mathbf{z}^{(1)}, \ldots, \mathbf{z}^{(q)})$ be the two independent samples from $\mathcal{D}$ in the two corresponding iterations.

Let $T$ denote the set of points where $g$ and $P$ disagree. It follows from \Cref{eq:error-reduction-gadget} that the algorithm correctly computes $P(\mathbf{a})$ in the first iteration as long as none of the queried points lie in the set $T$. A similar statement also holds for the second iteration. This reasoning implies that
\begin{equation}
\label{eq:E1E2}
\Pr[\mathcal{E}_1 \wedge \mathcal{E}_2]\leq \sum_{i,j = 1}^q \Pr[\underbrace{\mathbf{a} \oplus \mathbf{y}^{(i)}}_{\mathbf{u}^{(i)}}\in T \; \wedge \; \underbrace{\mathbf{a}\oplus \mathbf{z}^{(j)}}_{\mathbf{v}^{(j)}}\in T].
\end{equation}

We bound each term in the above sum using hypercontractivity, \Cref{thm:hypercontractivity}. 

Fix $i,j\in [q]$. Note that for every fixing of $\mathbf{y}^{(i)}$, the vector $\mathbf{u}^{(i)}$ is distributed uniformly over $\{0,1\}^n$ (because $\mathbf{a}$ is uniform over $\{0,1\}^n$). In particular, this implies the following:
\begin{itemize}
    \item The random variable $\mathbf{u}^{(i)}$ is uniformly distributed.
    \item The random variables $\mathbf{u}^{(i)}$ and $\mathbf{y}^{(i)}$ are independent.
\end{itemize}
This means that $\mathbf{v}^{(j)}$ which is equal to $(\mathbf{u}^{(i)} \,\oplus \, \mathbf{y}^{(i)}) \, \oplus \, \mathbf{z}^{(j)}$ is drawn from the noise distribution $\mathcal{N}_{\rho}(\mathbf{u}^{(i)})$. Further, the parameter $\rho \leq 1/100$ since the co-ordinates of $\mathbf{y}^{(i)}$ and $\mathbf{z}^{(j)}$ are i.i.d. Bernoulli random variables that are each $1/10$-close to uniform.

Using \Cref{thm:hypercontractivity}, we have
\[
\prob{}{\mathbf{u}^{(i)}\in T \wedge \mathbf{v}^{(j)}\in T} \leq \gamma^{2/1+|\rho|}\leq \gamma^{1.5}.
\]
Plugging this into \Cref{eq:E1E2} implies that $\Pr[\mathcal{E}_{1} \wedge \mathcal{E}_{2}] \leq \bigO(q^{2}) \cdot \gamma^{1.5}$ (union bound over all pairs $(i,j) \in [q] \times [q]$). Therefore, $\Pr[\mathcal{B}^{g}(\mathbf{a}) \neq g(\mathbf{a})] \leq \bigO(q^{2}) \cdot \gamma^{1.5}$ and this concludes the analysis of $\mathcal{B}.$
\end{proof}

So far we have shown that if we have an error-reduction gadget, then we can use it to construct a base algorithm $\mathcal{B}$, which in turn can be used recursively to construct an error-reduction algorithm for small constant error to sub-constant error. We now show how to construct the error-reduction gadget and prove \Cref{lem:error-reduction-gadget}. This requires the following standard claim (implied e.g. by M\"{o}bius inversion) that shows that any degree-$d$ polynomial over $\{0,1\}^n$ (even with group coefficients) can be interpolated from its values on a Hamming ball of radius $d.$ For completeness, we give a short proof.

\begin{lemma}\label{lem:mobius}
Fix $d \in \mathbb{N}$. For any natural number $m \geq d$ and any Hamming ball $B$ of radius $d$,
\[
P(0^m) = \sum_{\mathbf{b}\in B} \alpha_{\mathbf{b}} P(\mathbf{b})
\]
where the $\alpha_{\mathbf{b}}$ are integer coefficients.
\end{lemma}

\begin{proof}
    Assume that 
    \[
    P(\mathbf{x}) = \sum_{I\subseteq [n]: |I|\leq d} c_I \prod_{i\in I} x_i.
    \]
    By M\"{o}bius inversion (see item 1 of \Cref{thm:basic}), we know that
    \[
    c_I = \sum_{J\subseteq I} (-1)^{|I\setminus J|} P(1_J)
    \]
    where $1_J\in \{0,1\}^m$ denotes the indicator vector of set $J.$ Putting the above equalities together gives us 
    \[
    P(\mathbf{x}) = \sum_{|\mathbf{b}|\leq d} \alpha_{\mathbf{b},\mathbf{x}}' P(\mathbf{b})
    \]
    for suitable integer coefficients $\alpha'_{\mathbf{b},\mathbf{x}}.$ 
    
    Now, assume $B$ is the Hamming ball of radius $d$ around the point $\mathbf{c}\in \{0,1\}^m$. Replacing $\mathbf{x}$ by $\mathbf{x}\oplus \mathbf{c}$ in $P$ does not increase the degree of the polynomial (since this only involves negating a subset of the variables). Applying this substitution above yields
    \[
    P(\mathbf{x}\oplus \mathbf{c}) = \sum_{|\mathbf{b}|\leq d} \alpha_{\mathbf{b},\mathbf{x}}' P(\mathbf{b}\oplus \mathbf{c}) = \sum_{\mathbf{b}\in B}\alpha_{\mathbf{b},\mathbf{x}} P(\mathbf{b}).
    \]
    Setting $\mathbf{x} = \mathbf{c}$ yields the statement of the lemma.
\end{proof}

We end this section by completing the proof of \Cref{lem:error-reduction-gadget}.

\begin{proof}[Proof of \Cref{lem:error-reduction-gadget}]
    The idea is to apply \Cref{lem:mobius} on a random subcube, as defined in \Cref{defn:random-embedding}.

    More precisely, for an even integer $k > 2d$ that we will fix below, let $\mathbf{a}\in \{0,1\}^n$ be arbitrary and let $h:[n]\rightarrow [k]$ be chosen uniformly at random. Let $C = C_{\mathbf{a},h}$ be the corresponding subcube of $\{0,1\}^n$. Let $Q(y_1,\ldots,y_k)$ denote $P|_C$, the restriction of $P$ to this subcube.

    Fix a Hamming ball $B$ of radius $d$ in $\{0,1\}^k$ centred at a point $\mathbf{c}$ of weight exactly $k/2.$

    Since $Q$ is a polynomial of degree at most $d$, applying \Cref{lem:mobius} to $Q$ and the ball $B$ yields an equality
    \[
    Q(0^k) = \sum_{\mathbf{b}\in \mathbf{B}} \alpha_{\mathbf{b}} Q(\mathbf{b}).
    \]
    Since $Q$ is a restriction of $P$, the above equality can be rephrased in terms of $P$ as
    \[
    P(x(0^k)) = \sum_{\mathbf{b}\in \mathbf{B}} \alpha_{\mathbf{b}} P(x(\mathbf{b})).
    \]
    From the definition of the cube $C$, it follows that $x(0^k) = \mathbf{a}$ and thus the above gives us an equality of the type desired in an error-reduction gadget (\Cref{eq:error-reduction-gadget}). To finish the proof, we only need to argue that each $x(\mathbf{b})$ has the required distribution.

    Note that for each $\mathbf{b}\in B$, we have
    \[
    x(\mathbf{b}) =  \mathbf{a} \oplus \mathbf{b}_h
    \]
    where $\mathbf{b}_h$ is the random vector in $\{0,1\}^n$ that at co-ordinate $i$ takes the random value $b_{h(i)}$. Since $h$
    is chosen uniformly at random, it follows that the entries of $\mathbf{b}_h$ are independent and the $i$th co-ordinate is a Bernoulli random variable that takes the value $1$ with probability equal to the relative Hamming weight of $\mathbf{b}.$ 

    To conclude the argument, note that $\mathbf{b}$ is at Hamming distance at most $d$ from $\mathbf{c}$, implying that it has relative Hamming weight in the range
    \[
    \left[\frac{1}{2}-\frac{2d}{k},\frac{1}{2}+\frac{2d}{k}\right].
    \]
    Setting $k$ larger than $4d/\rho$ gives us the desired value for the parameter of the Bernoulli distribution.
    
    Finally, the number of queries $q$ made by the error-reduction gadget is dictated by the size of a Hamming ball in $k = O(d/\rho)$ dimensions. Since this is at most $2^k$, it follows that we have a $(\rho,2^{O(d/\rho)})$-error-reduction gadget.
\end{proof}

Summarising the proof of \Cref{lem:error-reduction-small-constant-main} - We first show that given a base algorithm $\mathcal{B}$ (see \Cref{lem:error-reduction-subroutine}), we can use it recursively to construct an error reduction algorithm (see the algorithm in the proof of \Cref{lem:error-reduction-small-constant-main}). Then we show that using an error-reduction gadget, we can design a base algorithm $\mathcal{B}$, where we use hypercontractivity to bound the error of the base algorithm. Finally, we use M\"{o}bius inversion (see \Cref{lem:mobius}) to construct an error-reduction gadget.

\subsection{Error Close to Half the Minimum Distance (Proof of \Cref{thm:uniquedeg1})}\label{subsec:error-reduction-close-radius}
In this subsection, we explain the third step towards proving \Cref{thm:uniquedeg1}. We will show that there is a randomized algorithm $\mathcal{A}^{f}$ that given oracle access to any function $f$ that is $\delta$-close to a low-degree polynomial $P$ (think of $\delta$ to be very close to half the minimum distance, i.e. $1/2^{d+1} - \varepsilon$ for degree $d$ polynomials), has the following property: with high probability over the internal randomness of $\mathcal{A}$, $\mathcal{A}^{f}$ is $\eta$-close to $P$, where $\eta < \delta$. We state it formally below.\\

\begin{lemma}\label{lem:error-reduction-main}
Fix any Abelian group $G$ and a positive integer $d$. For any $\eta, \delta$, where $\eta < \delta$ and $\delta < 1/2^{d+1} - \varepsilon$ for $\varepsilon > 0$, there exists a randomized algorithm $\mathcal{A}$ with the following properties:\newline
Let $f: \Boo^{n} \to G$ be a function and let $P: \Boo^{n} \to G$ be a degree $d$ polynomial such that $\delta(f,P) \leq \delta$, and let $\mathcal{A}^{f}$ denotes that $\mathcal{A}$ has oracle access to $f$,  then
\begin{align*}
    \Pr[\delta(\mathcal{A}^{f}, P) > \eta] < 1/10,
\end{align*}
where the above probability is over the internal randomness of $\mathcal{A}$, and for every $\mathbf{x} \in \Boo^{n}$, $\mathcal{A}^{f}$ makes $2^{k}$ queries to $f$, where $k = \poly(\frac{1}{\varepsilon},\frac{1}{\eta})$.
\end{lemma}

Putting this together with the unique correction algorithm for constant error (\Cref{thm:uniquedecoding-const-error}), we immediately get \Cref{thm:uniquedeg1}. Since the details are almost identical to the proof of \Cref{thm:uniquedecoding-const-error}, we omit the proof.

Now we state the algorithm $\mathcal{A}^{f}$.

\begin{algobox}
\begin{algorithm}[H]

\DontPrintSemicolon

\KwIn{$f$ and $\mathbf{a} \in \Boo^{n}$}

Choose $k = 1/(\varepsilon^{5} \eta^3)$ \;
Sample a uniformly random $h: [n] \to [k]$ \tcp*{$h$ is the internal randomness of $\mathcal{A}^{f}$}
Construct the cube $\mathsf{C} := C_{\mathbf{a},h}$ according to \Cref{defn:random-embedding} \;
Let $\Tilde{f} := f|_{\mathsf{C}}$ \tcp*{$f|_{\mathsf{C}}$ is the restriction of $f$ to the subcube $\mathsf{C}$} 
Query $\tilde{f}$ on all inputs in $\{0,1\}^k$ and use the algoritm from \Cref{thm:non-local-unique} to find the polynomial $\Tilde{P}$ on $\mathsf{C}$ such that $\delta(\Tilde{f}, \Tilde{P}) < 1/2^{d+1}$ \tcp*{$2^{k}$ queries to $f$}
\If{such a polynomial $\Tilde{P}$ is found}{
\Return{$\Tilde{P}(0^{k})$} \tcp*{$x(0^{k}) = \mathbf{a}$}
}
\Else{
\Return{0} \tcp*{An arbitrary value}
}

\caption{Error Reduction Algorithm $\mathcal{A}^{f}$}
\label{algo:error-reduction}
\end{algorithm}
\end{algobox}

We now analyze \Cref{algo:error-reduction} and prove \Cref{lem:error-reduction-main}.

\begin{proof}[Proof of \Cref{lem:error-reduction-main}]
Let $P$ be the degree $d$ polynomial such that $\delta(f,P) \leq 1/2^{d+1}$. The degree of $P$ is at most $d$ when $P$ is restricted to $\mathsf{C} = C_{\mathbf{a},h}$. If $\delta(P|_{\mathsf{C}}, \Tilde{f}) < 1/2^{d+1}$, then $\Tilde{P} = P|_{\mathsf{C}}$. In particular, $\Tilde{P}(x(0^{k})) = P(\mathbf{a})$, i.e. the output of the algorithm is correct.

Equivalently, $\mathcal{A}^{f}(\mathbf{a}) = P(\mathbf{a})$ unless $\delta(P|_{\mathsf{C}}, \Tilde{f}) \geq 1/2^{d+1}$. In the next lemma, we will show that with high probability over random $\mathbf{a}$ and $h$, $\delta((P|_{\mathsf{C}}, \Tilde{f}) < 1/2^{d+1}.$
\begin{lemma}\label{lemma:error-reduction-main}
Sample $\mathbf{a}$ and $h$ uniformly at random, and let $\mathsf{C} = C_{\mathbf{a}, h}$ be the subcube of dimension $k$ as described in \Cref{defn:random-embedding}. Then,
\begin{align*}
    \Pr_{\mathbf{a},h}[\delta(P|_{\mathsf{C}}, \Tilde{f}) \geq 1/2^{d+1}] < \eta/10
\end{align*}
\end{lemma}
We prove \Cref{lemma:error-reduction-main} below. For now, let us assume \Cref{lemma:error-reduction-main} and finish the proof of \Cref{lem:error-reduction-main}. We have,
\begin{gather*}
    \Pr_{\mathbf{a},h}[\delta(P|_{\mathsf{C}}, \Tilde{f}) \geq 1/2^{d+1}] <  \eta/10 \\
    \Rightarrow \mathbb{E}_{h} \, \Pr_{\mathbf{a}}[\delta(P|_{\mathsf{C}}, \Tilde{f}) \geq 1/2^{d+1}] <  \eta/10
\end{gather*}
Note that if we fix $h$, i.e. the internal randomness of $\mathcal{A}^{f}$, then $\delta(\mathcal{A}^{f}, f)$ is at most $\Pr_{\mathbf{a}}[\delta(P|_{\mathsf{C}}, \Tilde{f}) \geq  1/2^{d+1}]$, as the algorithm always outputs $P(\mathbf{a})$ correctly when $\delta(P|_{\mathsf{C}}, \Tilde{f}) < 1/2^{d+1}$ . Then from the above inequality, we have,
\begin{align*}
    \mathbb{E}_{h} \, [\delta(\mathcal{A}^{f}, f)] <  \eta/10 \\
    \Rightarrow \Pr_{h}[\delta(\mathcal{A}^{f}, f) > \eta] \leq 1/10 \tag*{(Markov's Inequality)}
\end{align*}
As commented in \Cref{algo:error-reduction}, for each $\mathbf{a} \in \Boo^{n}$, $\mathcal{A}^{f}$ makes $2^{k}$ queries to $f$.
\end{proof}
Now we give the proof of \Cref{lemma:error-reduction-main}.
\begin{proof}[Proof of \Cref{lemma:error-reduction-main}]
Let $E$ denote the subset of points in $\Boo^{n}$ where $P$ and $f$ disagree, i.e. $E := \setcond{\mathbf{x} \in \Boo^{k}}{f(\mathbf{x}) \neq P(\mathbf{x})}$. We know that $|E|/2^{n} \leq 1/2^{d+1} - \varepsilon$. Applying \Cref{lemma:sampling-subcube}, we get that for $k = \frac{1}{\varepsilon^5\eta^3}$ (we assume without loss of generality that $\varepsilon,\eta$ are small enough for $k$ to satisfy the hypothesis of \Cref{lemma:sampling-subcube})
\begin{align*}
    \Pr\Pr_{\mathbf{a},h}[\delta(P|_{\mathsf{C}}, \Tilde{f}) \geq 1/2^{d+1}] < \eta/10,
\end{align*}
and this completes the proof of \Cref{lemma:error-reduction-main}.
\end{proof}

\section{Combinatorial Bound for List Decoding Linear Polynomials}
In this section, we are going to prove \Cref{thm:comblistdeg1}. Let $G$ be any Abelian group and let $f: \Boo^{n} \to G$ be a polynomial such that $f$ is $(1/2-\varepsilon)$-close to $\mathcal{P}_{1}$. For small enough $\varepsilon$, $(1/2 - \varepsilon)$ is strictly more than the unique decoding radius of $\mathcal{P}_{1}$, which means that there can be several  polynomials in $\mathcal{P}_{1}$ that are $(1/2-\varepsilon)$-close to $f$. We denote the set of these polynomials by $\mathsf{List}_{\varepsilon}^{f}$, defined as follows:
\begin{align*}
    \mathsf{List}_{\varepsilon}^{f} := \setcond{P(\mathbf{x}) \in \mathcal{P}_{1}}{\delta(f,P) \leq 1/2 - \varepsilon}
\end{align*}

Let $L(\varepsilon) = |\mathsf{List}_{\varepsilon}^{f}|$.
In \Cref{thm:comblistdeg1} we show that $\mathsf{List}_{\varepsilon}^{f}$ is a \emph{small} list, i.e. $L(\varepsilon) =\poly(1/\varepsilon)$. We prove \Cref{thm:comblistdeg1} in the following steps:
\begin{itemize}
    \item Step 1: We prove that the list size is always a finite number, even though the underlying group $G$ is not finite (see \Cref{claim:finite-list}).
    \item Step 2: We show that to give an upper bound on $L(\varepsilon)$, we can assume without loss of generality the underlying group is finite (see \Cref{claim:G-finite}).
    \item Step 3: We decompose the group $G$ in two cases, depending on the order of the elements in $G$:
    \begin{itemize}
        \item Case 1: Every element has order a power of $q$ for a prime $q \in \set{2,3}$ (see \Cref{thm:comb-prod-2groups}).
        \item Case 2:  Every element has order a power of $p$ for a prime $p \geq 5$ (see \Cref{thm:comb-prod-pgroups}).
    \end{itemize}
\end{itemize}
We start by describing the first two steps.\\

\noindent
If $G$ is not finite, e.g. $G = \mathbb{R}$, then apriori it is not clear whether $L(\varepsilon)$ is even finite or not. As a warm-up, we first prove that $L(\varepsilon)$ is finite. This result will also be used later in our proofs.\\

\begin{claim}[The list is finite]\label{claim:finite-list}
Let $f: \Boo^{N} \to G$ be a polynomial which is $(1/2-\varepsilon)$-close to $\mathcal{P}_{1}$, for $\varepsilon > 0$. Then, $|\mathsf{List}_{\varepsilon}^{f}| \leq 2^{2^{N}}$.
\end{claim}
\begin{proof}
For every $P \in \mathsf{List}_{\varepsilon}^{f}$, there is a subset of size at least $(1/2 + \varepsilon) \cdot 2^{N}$ on which $P$ and $f$ agree. Two distinct polynomials in $\mathsf{List}_{\varepsilon}$ cannot agree on more than $1/2 \cdot 2^{N}$ points (as $\delta(\mathcal{P}_1) \geq 1/2$). Thus for every subset of size at least $(1/2 + \varepsilon) \cdot 2^{N}$, there exists at most one polynomial $P$ in $\mathsf{List}_{\varepsilon}^{f}$ such that $f$ and $P$ agree on that subset. Hence the number of polynomials in $\mathsf{List}_{\varepsilon}^{f}$ is at most the number of subsets of $2^{N}$ of size $(1/2 + \varepsilon) \cdot 2^{N}$, and the claim follows.
\end{proof}
Thus \Cref{claim:finite-list} shows that $\mathsf{List}_{\varepsilon}^{f}$ is finite, although the upper bound on $L(\varepsilon)$ is doubly-exponential in $n$. In \Cref{thm:comblistdeg1}, we will prove that $L(\varepsilon)$ is independent of $n$, and is a polynomial of $1/\varepsilon$.\newline
Next, we show that

\paragraph{\underline{The underlying group is finite}}We will simplify our situation by showing that we can assume without loss of generality that $G$ is a \textit{finite} Abelian group. This will allow us to decompose $G$ as a finite product of cyclic groups of prime order and argue about combinatorial bound by considering the projection on each of these groups.\\

\begin{claim}\label{claim:G-finite} 

Let $f:\{0,1\}^N \to G$ be a function which is $(1/2-\varepsilon)$-close to $\mathcal{P}_1$, for $\varepsilon>0$. Then there exists a finite group $G'$ and a function $f':\{0,1\}^N \to G'$ such that $|\mathsf{List}^f_\varepsilon|\le |\mathsf{List}^{f'}_\varepsilon|$.
\end{claim}
The idea of the proof is as follows. We use~\Cref{claim:finite-list} to first argue that there exists a \textit{finitely generated} subgroup of $G$ such that all the coefficients of the polynomials in $\textsf{List}^f_\varepsilon$ are in this subgroup. Then to go from a finitely generated subgroup to a finite group, we simply ``truncate'' the group elements by going modulo a large enough number.
\begin{proof}
    We will first prove the above claim for a {\em finitely generated} subgroup $G'' \subseteq G$ and then describe how to find a finite group $G'$ (not necessarily a subgroup) that still meets the above conditions. We define $G''$ as the subgroup generated by the evaluations of polynomials in $\mathsf{List}_{\varepsilon}^{f}$ and $f$, i.e.,
    $$G'':=\langle \{P(\mathbf{x}):\mathbf{x}\in \{0,1\}^N \text{ and } P\in \mathsf{List}^f_\varepsilon \} \cup \{f(\mathbf{x}):\mathbf{x}\in \{0,1\}^N\}\rangle,$$ where $\langle S \rangle$ denotes the subgroup generated by the elements of a subset $S \subseteq G$. We define $f'':\{0,1\}^N \to G''$ as $f''(\mathbf{x}) = f(\mathbf{x})$.\newline
    Let $\mathsf{List}^f_\varepsilon=\{P_1,\dots,P_t\}$ for some integer $t$; here we are using~\Cref{claim:finite-list} which says that the list is of finite size (even for infinite groups). We define $P''_i:\{0,1\}^N \to G''$ as $P''_i(\mathbf{x})=L_i(\mathbf{x})$ for each $P_i \in \mathsf{List}^f_\varepsilon$. Since $P_i(\mathbf{x}) \in G''$ for all $\mathbf{x}\in \{0,1\}^N$, we observe that all the coefficients of $P_i$ are in $G''$. Hence, $P''_i$ is a linear polynomial in $G''$, whose distance from $f''$ is $(1/2-\varepsilon)$, i.e., $P''_i \in \mathsf{List}^{f''}_\varepsilon$. Moreover, $P''_i$ for $i\in [t]$ are all distinct functions. Hence, $|\mathsf{List}^f_\varepsilon|\le |\mathsf{List}^{f''}_\varepsilon|$. 

    Now by the classification of finitely generated Abelian groups, $G'' = \mathbb{Z}^r \times \mathbb{Z}_{r_1} \times \mathbb{Z}_{r_2} \times \dots \times \mathbb{Z}_{r_k}$ for some integers $r, k\ge 0$ and $r_1,\dots,r_k \ge 2$. If $r=0$, we can take $G'=G''$ and that finishes the proof. Otherwise, we let $$M:=2\cdot\max\{\{|P''_i(\mathbf{x})_j|:i\in [t], j\in [r],\mathbf{x}\in\{0,1\}^N\}\cup \{|f''(\mathbf{x})_j|:j\in [r], \mathbf{x}\in \{0,1\}^N\}\} + 1$$ where $a_j\in \mathbb{Z}$ denotes the $j$-th coordinate of $a$, for $a\in G''$ and $j\in [r]$. This choice of $M$ is to ensure that no two distinct elements among the evaluations of $P''_i$'s and $f''$ are equal modulo $M$. We take $G' = \mathbb{Z}_M^r \times \mathbb{Z}_{r_1} \times \mathbb{Z}_{r_2} \times \dots \times \mathbb{Z}_{r_k}$ and define a homomorphism $\phi:G'' \to G'$ by applying the map $x\mapsto x\mod M$ to the first $r$ coordinates of the input and the identity map on the remaining coordinates. 
    Let $f':\{0,1\}^N \to G'$ be defined as $f'(\mathbf{x}) = \phi(f''(\mathbf{x}))$. For $i\in [t]$, if $P''_i(\mathbf{x})=a^{(i)}_{0} + a^{(i)}_{1}x_1 + \ldots + a^{(i)}_{N} x_N$ we define $P_i':\{0,1\}^N \to G'$ as $P_i'(\mathbf{x}) = \phi(a_0^{(i)}) + \phi(a_1^{(1)}) x_1 + \phi(a_2^{(i)}) x_2 + \dots + \phi(a_N^{(i)}) x_N$. As for $\mathbf{x}\in \{0,1\}^N$, $f''(\mathbf{x})=P''_i(\mathbf{x})$ implies that $f'(\mathbf{x})=P_i'(\mathbf{x})$ and since $P'_i$ is a linear polynomials over $G'$, we have that $P_i'\in \mathsf{List}^{f'}_\varepsilon$. Further, since all the initial $r$ coordinates of the coefficients of $P''_i$ are at most $M$ in absolute value to begin with, we get that $P_i' \ne P_j'$ for $i\ne j$. That is, $|\mathsf{List}^{f'}_\varepsilon| \ge t = |\mathsf{List}^{f}_\varepsilon|$.
    
\end{proof}

Hence to obtain an upper bound on $|{\mathsf{List}^f_\varepsilon}|$, it suffices to upper bound $|{\mathsf{List}^{f'}_\varepsilon}|$. Therefore, without loss of generality, for the rest of the proof we will assume that $G$ is a finite Abelian group.\\
Now we describe the third step towards proving \Cref{thm:comblistdeg1}. Using the structure theorem for finite Abelian groups, we know that $G$ can be written as a product of finitely many cyclic $p$-groups\footnote{For a prime $p$, a $p$-group is a group in which every element has order a power of $p$. For a cyclic group, this is just $\mathbb{Z}_{p^k}$ for some non-negative integer $k$.}. We decompose $G$ as follows:
\begin{align*}
    G = G_{1} \times G_{2} \times G_{3},
\end{align*}
where $G_{1}$ is product of $2$-groups, $G_{2}$ is a product of $3$-groups and $G_{3}$ is a product of $p$-groups for $p \geq 5$. Let $f: \Boo^{n} \to G$ be a polynomial, then $f = (f_{1}, f_{2}, f_{3})$, where $f_{i}: \Boo^{n} \to G_{i}$. We will prove a combinatorial list decoding bound for each $f_{i}$, and then the product of these bounds will be an upper bound on $|\mathsf{List}_{\varepsilon}^{f}|$. We have two cases - in the first case, we provide an upper bound for $G_{1}$ and $G_{2}$, and in the second case, we provide an upper bound for $G_{3}$. We state the upper bounds formally below.\\

\begin{theorem}[Combinatorial bound for product of $2$ and $3$-groups]\label{thm:comb-prod-2groups}
Let $G$ be a product of finitely many $q$-groups, where $q \in \set{2,3}$  and let $f: \Boo^{n} \to G$ be any function. Then, $|\mathsf{List}_{\varepsilon}^{f}| \leq \mathrm{poly}(1/\varepsilon)$.
\end{theorem}

\begin{theorem}[Combinatorial Bound for Product of $p$-groups ($p\geq 5$)]\label{thm:comb-prod-pgroups}
Let $G$ be a product of finitely many $p$-groups, where $p \geq 5$ and let $f: \Boo^{n} \to G$ be any function. Then, $|\mathsf{List}_{\varepsilon}^{f}| \leq \mathrm{poly}(1/\varepsilon)$.\\
\end{theorem}

Assuming \Cref{thm:comb-prod-2groups} and \Cref{thm:comb-prod-pgroups}, we immediately get \Cref{thm:comblistdeg1} because if $P = (P_{1},P_{2},P_{3}) \in \mathsf{List_{\varepsilon}^{f}}$, then for each $i \in [3]$, $P_{i}$ must be in $\mathsf{List}_{\varepsilon}^{f_{i}}$. In the next subsection, we prove \Cref{thm:comb-prod-pgroups}  and in the subsection after that, we prove \Cref{thm:comb-prod-2groups}.

\subsection{Combinatorial Bound for Product of $p$-groups ($p\geq 5$)}\label{subsec:5-groups}
In this subsection, we prove a particular case for the third step towards proving \Cref{thm:comblistdeg1}. We will prove \Cref{thm:comb-prod-pgroups}. We start by proving a result on the sparsity of a polynomial, using an anti-concentration lemma (see \Cref{lemma:sparse-pgroups}).

\subsubsection{Sparsity and Anti-concentration}
Recall that a character of a finite Abelian group $G$ is a homomorphism from $G$ to $\mathbb{C^*}$. For a subgroup $H$ of $G$, the characters of $H$ can be extended to obtain characters of $G$ (we refer to \cite{Conrad} for details). 
The following theorem is a well-known fact about the extensions of characters of a subgroup to characters of the group (see \cite[Theorem 3.4]{Conrad}).

\begin{theorem}
\label{thm: extension of characters}
    Let $G$ be a finite Abelian group and $H$ be a subgroup of $G$. Then each character of $H$ can be extended to a character of $G$ in $|G|/|H|$ ways.
\end{theorem}

An immediate corollary of the above theorem is the following.

\begin{corollary}
\label{coro: random character eval}
    Let $G$ be a finite Abelian group and $a \in G$ such that $\text{order}(a) = r \geq 1$. Let $\chi$ be a randomly chosen character of $G$, we have
    \begin{align*}
        \Pr_{\chi}[\chi(a) = e^{\frac{2\pi k i }{r}}] = \frac{1}{r}
    \end{align*}
    for all $k \in \{0, \dots, (r-1)\}$.
\end{corollary}
\begin{proof}
    Let $H = \langle a \rangle$, be the cyclic group generated by $a$. The characters of $H$ are $e^{\frac{2\pi k i}{r}}$ for $k \in \{0, \dots, (r-1)\}$. By Theorem~\ref{thm: extension of characters}, each of these characters has exactly $|G|/r$ many extensions to characters of $G$. The corollary follows since $G$ has $|G|$ many characters.
\end{proof}

Next, we prove an important result regarding linear polynomials over groups that are a product of $p$-groups for $p \geq 5$. This lemma, which reflects certain anti-concentration properties of linear polynomials over such groups, is the only part of the proof of \Cref{thm:comb-prod-pgroups} that uses something about the structure of the group. The following lemma says that if two distinct linear polynomials agree on a large fraction of points, then their respective coefficient vector must be quite similar.\newline

Given a polynomial $P\in \mathcal{P}_1,$ we use $\mathrm{vars}(P)$ to denote the set of variables with non-zero coefficient in $P.$\\

\begin{lemma}[Anti-concentration lemma]\label{lemma:sparse-pgroups}
Let $G$ be a product of finitely many $p$-groups where $p \geq 5$ and let $P_{i}, P_{j} : \Boo^{n} \to G$ be two distinct linear polynomials. If $P_{i}$ and $P_{j}$ agree on at least $(1/4 - 0.001)$-fraction of $\Boo^{n}$, then $|\mathrm{vars}(P_{i}-P_{j})| \leq C_0$, where $C_0 \geq 4500$ is a constant.
\end{lemma}
\begin{proof}
Let $\Tilde{P}(\x) = P_i(\x) - P_j(\x)$ for all $\x \in \{0,1\}^n$. Thus $\Pr_{\x \in \{0,1\}^n} [\Tilde{P}(x) = 0] = 1 - \delta(P_i, P_j) \geq 1/4 - 0.001$. Let $\Tilde{P}(\x) = \sum_{i = 1}^k a_i x_{j_i} + a_0$ where, for all $i \in \{0,\ldots, k\}$, $a_i \in G$ are non-zero elements and $j_i \in [n]$. Our goal is to upper bound $k$ using the hypothesis of the lemma. 

Let $\widehat{G}$ be the group of characters of $G$. Recall~\cite[Theorem 4.1]{Conrad} that $\Tilde{P}(\x) = 0$ then 
\[
\mathbb{E}_{\chi \in \widehat{G}} [\chi(\Tilde{P}(x))] = \left\{\begin{array}{cc}
1 & \text{if $\Tilde{P}(\mathbf{x})=0$,}\\
0 & \text{otherwise.}
\end{array}
\right.
\]
Using this, we have
\begin{align}
    \frac{1}{4} - 0.001 
    &\leq \Pr_{\x \in \{0,1\}^n} [\Tilde{P}(\x) = 0] \nonumber\\
    &= \left| \mathbb{E}_{\chi \in \widehat{G}} \mathbb{E}_{\x \in \{0,1\}^n} \chi(\Tilde{P}(x)) \right| \nonumber \\
    &\leq \mathbb{E}_{\chi \in \widehat{G}} \left| \mathbb{E}_{\x \in \{0,1\}^n}  \chi(\Tilde{P}(x)) \right| \tag{by triangle inequality} \\
    &= \mathbb{E}_{\chi \in \widehat{G}} \left| \mathbb{E}_{\x \in \{0,1\}^n}  \prod_{i=1}^k \chi(a_i x_{j_i}) \right|  \tag{using multiplicativity of $\chi$ and $|\chi(a_0)|=1$} \\
    &= \mathbb{E}_{\chi \in \widehat{G}} \left| \prod_{i=1}^k \left( \frac{1 + \chi(a_i)}{2} \right) \right|  \nonumber \\
    &= \mathbb{E}_{\chi \in \widehat{G}} \left[ \prod_{i=1}^k \left| \frac{1 + \chi(a_i)}{2} \right| \right]  \label{eq: for contradiction p large}.
\end{align}


The argument of a complex number $z \in \mathbb{C}$ of absolute value $|z|=1$, denoted by $\text{arg}(z)$,  is the unique real number $\alpha \in (-\pi,\pi]$ such that $z = e^{i \alpha}$.


We say that $a$ is ``bad'' for $\chi$ if $|\text{arg}(\chi(a))| < (2 \pi)/20$. Thus, if $\text{order}(a) = r \geq 5$, then from \Cref{coro: random character eval} the probability that a random $\chi$ is bad for $a$ is at most $1/5$. 
Since $\Tilde{P}(\x) = \sum_{i=1}^k a_i x_{j_i}+a_0$ where each $a_i \in G$ ($i\in [k]$) has order at least $5$, the expected number of $a_i$'s that are bad for a randomly chosen $\chi$ is at most $k/5$. 

Let $E$ be the event that there are at least $(\beta k)$ $a_i$'s that are not bad for a randomly chosen $\chi$, where $\beta$ is a suitable constant to be fixed later. By Markov's inequality, $\Pr[\overline{E}] \leq 1/(5 (1-\beta))$. We have 
\begin{align*}
    \mathbb{E}_{\chi \in \widehat{G}} \left[ \prod_{i=1}^k \left| \frac{1 + \chi(a_i)}{2} \right| \right]
    &= \mathbb{E}
    \left[ \left( \prod_{i=1}^k \left| \frac{1 + \chi(a_i)}{2} \right| \right) \mid E \right] \Pr[E] + \mathbb{E}
    \left[ \left( \prod_{i=1}^k \left| \frac{1 + \chi(a_i)}{2} \right| \right) \mid \overline{E} \right] \Pr[\overline{E}] \\
    &\leq \mathbb{E}
    \left[ \left( \prod_{i=1}^k \left| \frac{1 + \chi(a_i)}{2} \right| \right) \mid E \right] \Pr[E] + 
    \Pr[\overline{E}] \\
    &\leq \mathbb{E}
    \left[ \left( \prod_{i=1}^k \left| \frac{1 + \chi(a_i)}{2} \right| \right) \mid E \right] +  \frac{1}{5(1 - \beta)},
\end{align*}
where the first inequality uses the fact that $|(1 + \chi(a)/2| \leq 1$ for all $\chi \in \widehat{G}$ and $a \in G$. 
Conditioned on $E$, at least $(\beta k)$ many $a_i$'s satisfy $|\text{arg}(\chi(a_i))| \geq (2\pi )/20$. For any such $a_i$,
\begin{align*}
    \left| \frac{1 + \chi(a_i)}{2} \right| = 1/2 \cdot \sqrt{ \left( 1 + \cos(\text{arg}(a_i)) \right)^2 + \sin(\text{arg}(a_i))^2} = |\cos(\text{arg}(a_i)/2)| \leq \cos(\pi/20).
\end{align*}
Combining with \Cref{eq: for contradiction p large} we have
\begin{align*}
    \frac{1}{4} - 0.001 \leq \mathbb{E}_{\chi \in \widehat{G}} \left[ \prod_{i=1}^k \left| \frac{1 + \chi(a_i)}{2} \right| \right]
    \leq \cos(\pi/20)^{\beta k} + \frac{1}{5(1 - \beta)}.
\end{align*}
Choosing $\beta$ to be $1/9$, we have a contradiction if $k$ is larger than $C_0$.
\end{proof}

In \Cref{lemma:sparse-pgroups}, we proved that if two distinct linear polynomials agree on nearly $1/4$-fraction of the Boolean cube, then their difference is a sparse polynomial. Intuitively, two linear polynomials in $\mathsf{List}_{\varepsilon}^{f}$ also agree on a large fraction because both of them agree with $f$ on nearly $1/2$-fraction of the Boolean cube. We will use this observation along with \Cref{lemma:sparse-pgroups} to reduce to sparse polynomials.

\paragraph{\underline{Many of the polynomials are sparse}}We will now reduce our problem of proving an upper bound on $|\mathsf{List}_{\varepsilon}^{f}|$ to the setting of proving a combinatorial bound when it is also known that the polynomials in the list are `sparse', in the sense that they only have a few non-zero coefficients.\\

\noindent
Suppose $\mathsf{List}_{\varepsilon}^{f} = \set{P_{1}, \ldots, P_{t}}$. We define a graph $\mathcal{G} = (V,E)$ with $|V| = t$. The vertices in $V$ represents the codewords in $\mathsf{List}_{\varepsilon}^{f}$ -- in particular $i^{th}$ vertex corresponds to $P_i$.
Edge $(i,j)$ exists if and only if
$|\mathrm{vars}(P_{i}-P_{j})| \leq C_0$, where $C_0$ is the absolute constant from \Cref{lemma:sparse-pgroups}.
 
In other words, two linear polynomials $P_{i}$ and $P_{j}$ are related via an edge if $P_{i}-P_{j}$ is supported on at a most constant, $C_0$, many variables.

We will show that $\mathcal{G}$ has a vertex of large degree (in terms of the number of vertices $t$ of $\mathcal{G}$) and then instead of upper bounding $t=|\mathsf{List}_{\varepsilon}^{f}|$, we will upper bound the largest degree of $\mathcal{G}$. To be a bit more precise, in \Cref{lemma:clique}, we will show that $\mathcal{G}$ has a vertex of degree $\Omega(t)$. It will then suffice to show that the number of polynomials neighbouring any vertex in $\mathcal{G}$ is at most $\mathrm{poly}(1/\varepsilon)$.

To prove \Cref{lemma:clique}, we will use the following lemma which can be proved by forming the independent set greedily.
\begin{lemma}
\label{lem: large graph degree}
    Any undirected graph on $t$ vertices contains either an independent set of size $t/(\Delta+1)$ if all the vertices are of degree at most $\Delta$.
\end{lemma}

We also need the following lemma, a proof of which can found in~\cite[Lemma 2.1]{jukna2011extremal}.

\begin{lemma}
\label{lemma: pw intersection}
Let $A_1, \dots, A_k$ be sets of cardinality $r$ and let $X = \cup_{i=1}^k A_i$. If $|A_i \cap A_j| \leq t$ for all $i \neq j \in [k]$, then $|X| \geq (r^2 k)/(r + (k-1)t)$.
\end{lemma}

Now we will prove that the graph $\mathcal{G}$ as defined above, has a vertex of degree at least $\Omega(t)$.
\begin{lemma}[$\mathcal{G}$ has a vertex of large degree]\label{lemma:clique}
Let $\mathcal{G} = (V,E)$ be the graph as defined above on $t$ vertices. Then there exists a vertex $v \in V$ of degree $\Omega(t)$.
\end{lemma}
\begin{proof}
We will show that there does not exist an independent set of size $k$ in $\mathcal{G}$, where $k$ is a large enough constant. Then applying \Cref{lem: large graph degree}, we get that there is a vertex of degree at least $(t/k)-1 \ge \Omega(t)$.\newline
Let $v_1, \dots ,v_k \in V$ be any $k$ distinct vertices. We will show that $\set{v_{1},\dots,v_{k}}$ do not form an independent set for a large enough constant $k$. 
Recall that the vertex $v_i$ corresponds to some linear function $P_{j_i} \in \mathsf{List}_{\varepsilon}^{f}$. Let $A'_i = \{x : f(\x) = P_{j_i}(\x)\}$ and thus $|A'_i| \geq (1/2 \cdot 2^n)$. Also, let $A_i \subset A'_i$ such that $|A_i| = (1/2 \cdot 2^n)$. Now we use~\Cref{lemma: pw intersection} to upper bound $k$ assuming for all $i_1, i_2 \in [k]$, $|A_{i_1} \cap A_{i_2}| \leq ((1/4 - 0.001) \cdot 2^n)$. Since $| \cup_{i=1}^k A_i| \leq 2^n$, we have
\begin{align*}
    2^n \geq \frac{(1/2 \cdot 2^n)^2 \cdot k}{(1/2 \cdot 2^n) + (k-1) \cdot ((1/4 - 0.001) \cdot 2^n)}
\end{align*}
which implies that
\begin{align*}
    1/4 + 0.001 \geq k \cdot 0.001.
\end{align*}
Thus choosing $k$ to be constant greater than $4002$ implies that there exists $i_1, i_2 \in [k]$ such that $|A_{i_1} \cap A_{i_2}| > (1/4 - 0.001) \cdot 2^n$.

Since $\delta(P_{j_{i_1}}, P_{j_{i_2}}) \geq 1/2^n \cdot |A_{i_1} \cap A_{i_2}|$ for all $i_1, i_2 \in [k]$, using \Cref{lemma:sparse-pgroups} and the choice of $C_0$ as discussed before, we get 
$|\mathrm{vars}(P_{j_{i_1}}-P_{j_{i_2}})| \leq C_0$. In particular, the corresponding vertices are adjacent in $\mathcal{G}$, contradicting the assumption that $\{v_1,\ldots, v_k\}$ is an independent set.

Thus we have shown no subset of $k$ vertices in $\mathcal{G}$ forms an independent set and this concludes the proof.
\end{proof}

Let $v \in V$ be the vertex $\mathcal{G}$ given by \Cref{lemma:clique} with degree $m \ge \Omega(t)$. As discussed above, to prove a $\mathrm{poly}(1/\varepsilon)$ upper bound on $t = |\mathsf{List}_{\varepsilon}^{f}|$, it suffices to prove a $\mathrm{poly}(1/\varepsilon)$ upper bound on $m$. Let $P_0$ be the linear polynomial corresponding to the vertex $v$ and let $P_1,P_2,\dots,P_m$ be the linear polynomials in $\mathsf{List}_{\varepsilon}^{f}$ that are adjacent to $P_0$ in $\mathcal{G}$. For every $i \in [m]$, we define a degree $1$ polynomial $\Tilde{P}_{i} := P_{i} - P_{1}$, and we define $\Tilde{f} := f - P_{0}$. Note that $\delta(\Tilde{P}_{i}, \Tilde{f}) = \delta(P_{i}, f)$, i.e., $\Tilde{P}_i \in\mathsf{List}_{\varepsilon}^{\Tilde{f}} $. Moreover, by the definition of the graph $\mathcal{G}$, we have that $|\mathrm{vars}(\Tilde{P}_{i})| \leq C_0$ for all $i \in [m]$.

\noindent
Let $\ell = |\bigcup_{i = 1}^{m} \mathrm{vars}(\Tilde{P}_{i})|$. We will say $\ell$ is \textit{small} if $\ell \leq (1/\varepsilon)^{K}$, where $K = 20$, and $\ell$ is \textit{large} otherwise. Depending on which case we are in, we use \Cref{lemma:ell-small} or \Cref{lemma:ell-large} below to show that $m\leq \poly(1/\varepsilon).$ This completes the proof of \Cref{thm:comb-prod-pgroups} as we have already established that $m\ge \Omega(t)$.

\subsubsection{Union of Supports is Small}
In this subsection, we prove that if the number of variables in the support of polynomials in $\Tilde{V}$ is small, then $m = |\Tilde{V}|$ is small too. We prove the following lemma.\\

\begin{lemma}[$\ell$ is small]\label{lemma:ell-small}
Let $f : \Boo^{n} \to G$ be a polynomial and let $\mathcal{A} \subseteq  \mathsf{List}_{\varepsilon}^{f}$ be a set of degree-$1$ polynomials satisfying the following property: There exists an absolute constant $C$ such that for every $P \in \mathcal{A}$, $|\mathrm{vars}(P)| \leq C$.\newline
Let $\ell = |\bigcup_{P \in \mathcal{A}} \mathrm{vars}(P)|$. 
 Then we have $|\mathcal{A}| \leq O(\ell^C\cdot (1/\varepsilon))$.
\end{lemma}

\begin{proof}[Proof of \Cref{lemma:ell-small}]
Let $\mathcal{A} = \set{P_{1},\ldots,P_{m}}$. Choose any set $T \subseteq \bigcup_{i=1}^{m} \mathrm{vars}(P_{i})$ of size $C$ and let $\mathcal{A}_{T} \subseteq \mathcal{A}$ denote the subset of polynomials that are supported only on (a subset of the) variables in $T$. We have,
\begin{align}
    |\mathcal{A}| \leq \sum_{\substack{T\subseteq \cup_{i=1}^{m} \mathrm{vars}(P_{i}):\\ |T| = C}} |\mathcal{A}_{T}|.
\end{align}
Let $M$ be an upper bound on all $|\mathcal{A}_{T}|$'s, thus
\begin{align}
    |\mathcal{A}| \leq \binom{\ell}{C} M = O(\ell^C \cdot M). \label{eq: ub small eq1}
\end{align}
Fix an arbitrary $T \subseteq \bigcup_{i=1}^{m} \mathrm{vars}(P_{i})$ of size $C$. Label the variables indexed by $T$ as $\set{y_{1},\ldots,y_{|T|}}$ and the remaining variables as $\set{z_{1},\ldots,z_{n-|T|}}$. \Cref{claim:Random-subcube} shows that there exists an assignment $\mathbf{z} = \mathbf{b}$ such that $(\varepsilon/10)$-fraction of $\mathcal{A}|_{\mathbf{z} = \mathbf{b}}$ are $\approx (1/2 - \varepsilon)$-close to $f|_{\mathbf{z} = \mathbf{b}}$, where $\mathcal{A}|_{\mathbf{z}}$ denotes that every polynomial in $\mathcal{A}$ is restricted according to $\mathbf{z}$.


\begin{claim}\label{claim:Random-subcube}
There exists an assignment $\mathbf{b} \in \Boo^{n-|T|}$ to the $\mathbf{z}$ such that there is a subset $\mathcal{B} \subseteq \mathcal{A}_{T}$ satisfying the following properties:
\begin{itemize}
    \item For each $P \in \mathcal{B}$, $\delta(P|_{\mathbf{z} = \mathbf{b}}, f|_{\mathbf{z} = \mathbf{b}}) \leq (1/2 - \varepsilon/10)$,
    \item $|\mathcal{B}|$ is at least $\varepsilon/10 \cdot |\mathcal{A}_{T}|$.
\end{itemize}
\end{claim}
Let us defer the proof of \Cref{claim:Random-subcube} for a while and see how to use \Cref{claim:Random-subcube} to finish the proof of \Cref{lemma:ell-small}. Observe that for any polynomial in $\Tilde{P} \in \mathcal{B}$, $\delta(\Tilde{P}, f|_{\mathbf{z} = \mathbf{b}}) \leq (1/2 - \varepsilon/10)$ and $\Tilde{P}, f|_{\mathbf{z} = \mathbf{b}}$ depend only on $C$ variables. Thus by applying \Cref{claim:finite-list} on $f|_{\mathbf{z} = \mathbf{b}}$, we get the following: 
\begin{align*}
    (\varepsilon/10) \cdot |\mathcal{A}_{T}| \leq 2^{2^{C}} \; \Rightarrow \; M = \bigO(1/\varepsilon).
\end{align*}
Along with Equation~\ref{eq: ub small eq1}, this completes the proof of \Cref{lemma:ell-small}.
\end{proof}
Now we give the proof of \Cref{claim:Random-subcube} and this will complete the proof of the case when $\ell$ is small. 
\begin{proof}[Proof of \Cref{claim:Random-subcube}]
Fix any degree-$1$ polynomial $P \in \mathcal{A}_{T}$. Since $\delta(P,f) \leq (1/2 - \varepsilon)$, we have
\begin{align*}
    \Pr_{\mathbf{z}} \, \Pr_{\mathbf{y}}[f(\mathbf{y}, \mathbf{z}) \neq P(\mathbf{y}, \mathbf{z})] \leq \frac{1}{2} - \varepsilon,
\end{align*}
For an assignment $\mathbf{b} \in \Boo^{n - |T|}$ of $\mathbf{z}$, let $\rho_{\mathbf{b}}^{P}$ denote the distance between $P|_{\mathbf{z} = \mathbf{b}}$ and $f|_{\mathbf{z} = \mathbf{b}}$, i.e.
\begin{align*}
    \rho_{\mathbf{b}}^{P} = \Pr_{\mathbf{y}}[\Tilde{f}(\mathbf{y}, \mathbf{b}), \neq P(\mathbf{y}, \mathbf{b})]
\end{align*}
So we get,
\begin{align*}
    \mathbb{E}_{\mathbf{b}}[\rho_{\mathbf{b}}^{P}] \leq \frac{1}{2} - \varepsilon \\
    \Rightarrow \Pr_{\mathbf{b}}[\rho_{\mathbf{b}}^{P} > (1/2 - \varepsilon/10)] \leq \frac{1/2 - \varepsilon}{1/2 - \varepsilon/10} \tag*{(Markov's Inequality)} \\
    \Rightarrow \Pr_{\mathbf{b}}[\rho_{\mathbf{b}}^{P} \leq (1/2 - \varepsilon/10)] \geq \varepsilon/10
\end{align*}
This implies that for $\mathbf{b} \sim U_{n - |T|}$, in expectation at least $\varepsilon/10$ fraction of polynomials in $\mathcal{A}_{T}$ are $(1/2 - \varepsilon/10)$ close to $f|_{\mathbf{z} = \mathbf{b}}$. This implies the existence of an assignment $\mathbf{b} \in \Boo^{n-|T|}$ and a subset $\mathcal{B}$ as claimed.
\end{proof}

\subsubsection{Union of Supports is Large}
In this subsection, we prove the complementary case to \Cref{lemma:ell-small}.\\

\begin{lemma}[$\ell$ is large]\label{lemma:ell-large} 
Let $f : \Boo^{n} \to G$ be a polynomial and let $\mathcal{A} \subseteq  \mathsf{List}_{\varepsilon}^{f}$ be a set of degree $1$ polynomials such that $\ell = |\bigcup_{P\in \mathcal{A}} \mathrm{vars}(P)|$. If $|\mathrm{vars}(P)| \leq C_0$ for all $P \in \mathcal{A}$ and $\ell \geq (1/\varepsilon)^{20}$ (i.e. $\ell$ is large), then $|\mathcal{A}| = O((1/\varepsilon)^{20 C_0})$.
\end{lemma}

\begin{proof}
Let $\mathcal{A} = \set{P_{1},\ldots,P_{m}}$ where for contradiction we assume that $m \geq (C_0/\varepsilon)^{20C_0}$. Let $B$ denote the following subset of variables:
\begin{align*}
    B := \setcond{i \in \bigcup_{j=1}^{m} \mathrm{vars}(P_{j}) \; }{\; \# j \, \text{ such that } \, \mathrm{vars}(P_{j}) \ni i \text{~is at least~} \varepsilon^{10} \cdot m}.
\end{align*} 
Since $|\mathrm{vars}(P_{j})| \leq C_0$ for all $j \in [m]$, $|B|$ has at most $C_0/\varepsilon^{10}$ variables. 
We call a polynomial $P_j$ an \emph{ignore polynomial} if $\mathrm{vars}(P_j) \subset B$. The number of ignore polynomials is at most $O((C_0/\varepsilon^{10})^{C_0} \cdot (1/\varepsilon))$ using \Cref{lemma:ell-small}. 
Let $\mathcal{A}_{0}$ denote the set $\mathcal{A}$ obtained after removing the ignore polynomials, and we have $|\mathcal{A}_{0}| 
\geq m/2$.
 \\

To prove the claim, we construct a set $\mathcal{A}' \subseteq \mathcal{A}_{0}$ satisfying the following properties:
\begin{itemize}
    \item Each polynomial in $\mathcal{A}'$ depends on a variable outside $B$, i.e. for any $Q \in \mathcal{A}'$, $\mathrm{vars}(Q)$ is not a proper subset of $B$.
    \item For any two distinct polynomials in $\mathcal{A}'$, their pairwise intersection lies inside $B$, i.e. if $Q_{i}, Q_{j} \in \mathcal{A}'$ are two distinct polynomials, then,
    \begin{align}
        (\mathrm{vars}(Q_{i}) \setminus B) \cap (\mathrm{vars}(Q_{j}) \setminus B) = \emptyset. \label{eqn: fix vars in bad set}
    \end{align}
    In other words, if the variables in $B$ are fixed to an assignment, then the resulting polynomials in $\mathcal{A}'$ are supported on pairwise disjoint sets of variables.
    \item $|\mathcal{A}'| = \Omega(1/\varepsilon^{10})$.
\end{itemize}
We will construct $\mathcal{A}'$ iteratively, initially $\mathcal{A}' = \emptyset$. Consider any polynomial $Q_{1} \in \mathcal{A}_{0}$ and let $B_{1}$ denote the set $\mathrm{vars}(Q_{1}) \setminus B$. By definition of $B$, each variable in $B_{1}$ occurs in variable sets of at most $\varepsilon^{10} \cdot m$ polynomials and $|B_{1}| \leq C_0$. Thus there are at most $(\varepsilon^{10} \cdot m) \cdot C_0$ many polynomials in $\mathcal{A}_{0}$ containing a variable in $B_1$; remove the polynomials from $\mathcal{A}_0$ and also update $\mathcal{A}' = \mathcal{A}' \cup \{Q_1\}$. Observe that the size of the resulting $\mathcal{A}_{0}$ is at least $m/2 - (C_0 \cdot \varepsilon^{10} \cdot m)$. Thus, in a similar manner and using the assumption that $m \geq \poly(1/\varepsilon)$ for suitably large polynomial of $(1/\varepsilon)$, we can choose $Q_2$ from $\mathcal{A}_0$ in the next iteration, and so on to obtain $\mathcal{A}' = \{Q_1, Q_2, \dots, Q_r\}$ where $r = \Omega(1/\varepsilon^{10})$.

Now choose $\mathbf{x} \in \Boo^n$ uniformly at random and consider
\begin{align}
    \Pr_{\x}\left[\exists i^{\star} \in [r] : \frac{\sum_{i \neq i^{\star}, i=1}^{r} 1[Q_{i^{\star}}(\x) = Q_i(\x)]}{r} \geq (1/2 + \varepsilon/10 -1/r) \right]. \label{eqn: not many agree on a random point}
\end{align}
The above equation denotes the probability that at a random $\x$, some $Q_{i^{\star}}$ agrees with a large fraction of other $Q(i)$'s. 
Going forward, the proof strategy is to upper bound and lower bound the above probability to get a contradiction.

First for the upper bound, observe that for a random $\mathbf{x} \in \Boo^{n}$ conditioned on $\x|_B = \mathbf{a} \in \Boo^{|B|}$, by \Cref{eqn: fix vars in bad set}, $Q_{i}(\x)$'s are independent random variables as they depend on disjoint set of variables. In particular, for a fixed $j \in [r]$ and $i \neq j \in [r]$, $\Pr_{\x}[Q_i(\x) = Q_j(\x) \mid \x|_{B = a}] \leq 1/2$ (by the Schwartz-Zippel Lemma) for all $a \in \Boo^{|B|}$, and thus by Chernoff bound,
\begin{align*}
    \Pr_{\x}\left[\frac{\sum_{i \neq j, i=1}^{r}1[Q_j{(\x)} = Q_i(\x)]}{r} \geq (1/2 + \varepsilon/10 - 1/r) \right] \leq \exp(-\varepsilon^2 \cdot r).
\end{align*}
Thus by union bound, the probability in \Cref{eqn: not many agree on a random point} is upper bounded by $(r \cdot \exp(-\varepsilon^2 \cdot r)) = O(\exp(-(\varepsilon^2 \cdot r)/2)) = O(\exp(-1/\varepsilon^8))$.

Next, we give a lower bound for the probability in \Cref{eqn: not many agree on a random point}.
Since each $Q_i \in \mathcal{A}'$, for $i \in [r]$, agrees with $f$ on at least $(1/2 + \varepsilon)$ fraction of $\Boo^n$, we have
\begin{align*}
    \Pr_{i \in [r], \x \in \Boo^{n}} [Q_i(\x) \neq f(\x)] \leq (1/2 - \varepsilon).
\end{align*}
From an argument similar to that of \Cref{claim:Random-subcube} we have that on at least $(\varepsilon/10)$ fraction of $\x \in \Boo^{n}$, $\Pr_i[Q_i(\x) \neq f(\x)] \leq (1/2 - \varepsilon/10)$. On such an $\x$, there exists an $i^{\star} \in [r]$ such that $Q_{i^{\star}}$ agrees with at least $(1/2 + \varepsilon/10)r - 1$ of $Q_1, \dots, Q_r$. In other words,
\begin{align*}
    \Pr_{\x}\left[\exists i^{\star} \in [r] : \frac{\sum_{i \neq i^{\star}, i=1}^{r}1[Q_{i^{\star}}(\x) = Q_i(\x)]}{r} \geq (1/2 + \varepsilon/10 - 1/r) \right] \geq \varepsilon/10,
\end{align*}
which contradicts the upper bound.
\end{proof}

Summarizing the proof of \Cref{thm:comb-prod-pgroups} - We reduced our problem to upper bounding the list size of sparse polynomials (see \Cref{lemma:clique}). We then have two cases, depending on the size of the union of variables in the support. In the case where this is small (see \Cref{lemma:ell-small}), we upper bound the list size by fixing some variables and using a union bound. For the second case (see \Cref{lemma:ell-large}), we show that we can treat the polynomials as independent random variables and then use concentration inequalities.

\subsection{Combinatorial Bound for $2$ and $3$-groups}\label{subsec:comb-prod-2}
In this subsection, we will prove a particular case for the third step towards proving \Cref{thm:comblistdeg1}. We will prove \Cref{thm:comb-prod-2groups}. The proof uses similar techniques as used in \cite{DinurGKS-ECCC} to prove combinatorial bound for group homomorphisms. We start by first recalling a definition from \cite{DinurGKS-ECCC} of a set system with some nice properties.

\paragraph{Special intersecting set systems.} The next definition is about \textit{special intersecting sets}. Let $S_{1},\ldots,S_{t}$ be subsets of universe $X$. For a set $S \subseteq X$, let $\mu(S)$ denote the density of $S$, i.e. $\mu(S) := |S|/|X|$. For a subset of indices $I \subseteq [t]$, let $S_{I}$ denote the common intersection of subsets corresponding to indices in $I$, i.e. $S_{I} := \bigcap_{i \in I} S_{i}$.
\begin{definition}[Special intersecting sets \cite{DinurGKS-ECCC}]\label{def:special-intersection-sets}
Let $\rho, \tau \in (0,1]$ be two numbers such that $\tau \leq \rho$ and $c$ is a constant. Sets $S_{1},\ldots,S_{t}$ are said to be $(\rho, \tau, c)$-special intersecting sets if they satisfy the following conditions:
\begin{enumerate}
    \item (Each subset is dense) For every subset $S_{i}$, $\mu(S_{i})$ is at least $\rho$.
    \item (The pairwise intersection is small) For any two distinct subsets $S_{i}$ and $S_{j}$, $\mu(S_{i} \cap S_{j})$ is at most $\rho$.
    \item Let $\mu(S_{i}) = \rho + \alpha_{i}$. Then, $\sum_{i = 1}^{t} \alpha_{i}^{c} \leq 1$.
    \item (Sunflower-structure) For subsets $I,J \in [t]$ where $J \subseteq I$ and $|J| \geq 2$, if $\mu(S_{I})$ is strictly more than the threshold $\tau$, then the common intersection $S_{I}$ is equal to $S_{J}$.
\end{enumerate}
\end{definition}

The following lemma is a crucial lemma in our technical results, which essentially says that if we define a ``potential'' function on the density of subsets and the subsets form certain special intersection sets, then we can give an upper bound on the potential of the union of subsets. In particular, if the density of each subset is ``large'', then we can give an upper bound on the number of subsets.\\

\begin{lemma}\label{lemma:special-intersection-bound}
(Theorem 3.2 of \cite{DinurGKS-ECCC}). Fix any constant $C$. Then there exists a constant $D$ (depending on $C$) satisfying the following: Suppose $S_{1},\ldots,S_{t}$ are $(\rho, \rho^{2}, C)$-special intersection sets for $\rho > 0$ and $\mu(S_{i}) = \rho + \alpha_{i}$. Let $\mu \paren{\bigcup_{i} S_{i}} = \rho + \alpha$. Then,
\begin{align*}
    \sum_{i=1}^{t} \alpha_{i}^{D} \leq \alpha^{D} \\
\end{align*}
\end{lemma}

Recall that the Johnson bound (see e.g.\cite[Chapter 7]{GRS-codingbook}) allows us to bound the list-decodability of codes based on their distance. In what follows, we will need an analytic extension of this bound over $\mathbb{Z}_2$ and a similar (but incomparable) statement over $\mathbb{Z}_3$.\\

\begin{lemma}[Extended Johnson bound]\label{lemma:johnson-3}
Let $q \in \set{2,3}$. There exists an absolute constant $C > 0$ so that the following holds. Let $f:\{0,1\}^n\rightarrow \mathbb{Z}_q$ be any function and let $\Phi_{1},\ldots, \Phi_{t}: \Boo^{n} \to \mathbb{Z}_{q}$ be distinct degree-$1$ polynomials such that for every $i \in [t]$,  $f$ and $\Phi_{i}$ agree on at least $1/2 + \alpha_{i}$ fraction of $\Boo^{n}$. Then $\sum_{i=1}^{t} \alpha_{i}^{C} \leq 1$.\\
\end{lemma}

Note that the above statement immediately implies that the space $\mathcal{P}_1$ is $(\frac{1}{2}-\varepsilon, \poly(1/\varepsilon))$-list decodable over $\mathbb{Z}_2$ and $\mathbb{Z}_3$ since the number of indices $i$ for which $\alpha_i \ge \varepsilon$ can be at most $(1/\varepsilon)^C$ as the sum $\sum_{i=1}^t \alpha_i^C$ has to be at most 1. As we will see below, \Cref{lemma:johnson-3} is essentially equivalent to a statement bounding the number of polynomials with agreement at least $1/2+\varepsilon$. We will need the analytic formulation, however, to apply the proof ideas of \cite{DinurGKS-ECCC}. In the next subsection we will prove \Cref{lemma:johnson-3} and in the subsequent section, we will use \Cref{lemma:johnson-3} and special intersecting set systems to prove \Cref{thm:comb-prod-2groups}.

\subsubsection{Proof of \Cref{lemma:johnson-3} (Extended Johnson Bound)}
Here we prove \Cref{lemma:johnson-3}. When $q = 2$, this follows immediately from the standard ``binary Johnson bound'' (see, e.g.~\cite[Appendix A.1]{DinurGKS-ECCC}). In order to prove it for $q = 3$, we will first prove a combinatorial bound in the special case that the underlying group is $\mathbb{Z}_{3}$ and then \Cref{lemma:johnson-3} will follow from it.\\
We first show that for any $\varepsilon> 0 $, the number of $\Phi_i$'s for which $\alpha_i \ge \varepsilon$ is bounded by $\poly(1/\varepsilon)$.\\

\begin{claim}[Combinatorial bound for $\mathbb{Z}_{3}$]\label{claim:johnson-3-count}
    Let $\varepsilon > 0$ and $f: \Boo^{n} \to \mathbb{Z}_{3}$ be any function. Then, $|\mathsf{List}_{\varepsilon}^{f}| \le \mathrm{poly}(1/\varepsilon)$.\\
\end{claim}

Before proving \Cref{claim:johnson-3-count}, let's see how \Cref{claim:johnson-3-count} implies \Cref{lemma:johnson-3}.
\begin{proof}[Proof of \Cref{lemma:johnson-3}]
For any $j \in \mathbb{N}$, let $B_{j}$ represent the following subset of $\mathsf{List}_{\varepsilon}^{f}$: 
\begin{align*}
    B_j = \setcond{i\in [t]}{\frac{1}{2^{j+1}} < \alpha_i \le \frac{1}{2^{j}}},
\end{align*}
i.e. we partition $\mathsf{List}_{\varepsilon}^{f}$ depending on how high the agreement is with $f$. Let $c\in \mathbb{N}$ be a constant such that $t\le\paren{1/\varepsilon}^c$ in the conclusion of \Cref{claim:johnson-3-count}. Then taking $\varepsilon = 1/2^{j+1}$ and applying \Cref{claim:johnson-3-count}, we get that the size of $B_j$ is at most $\paren{2^{j+1}}^c$ for all $j\in\mathbb{N}$. Therefore,
    \begin{align*}
        \sum_{i=1}^t \alpha_i^{3c}  = \sum_{j\ge 1} \sum_{i\in B_j} \alpha_i^{3c}
         \le \sum_{j\ge 1} \abs{B_j} \cdot \paren{\frac{1}{2^{j}}}^{3c}
         \le \sum_{j\ge 1} \paren{\dfrac{1}{2^{j}}}^{c}  \le 1.
    \end{align*}
Setting $C = 3c$, we get that $\sum_{i=1}^{t} \alpha_{i}^{C} \leq 1$, and this completes the proof of \Cref{lemma:johnson-3} (assuming \Cref{claim:johnson-3-count}).
\end{proof}

\paragraph{}In the rest of the section, we will prove \Cref{claim:johnson-3-count}. We will first prove the following lemma, which is based on the proof of the binary Johnson bound mentioned above, but also uses some anti-concentration properties over $\mathbb{Z}_3.$

\begin{lemma}
\label{lemma:for-johnson-3}
    Let $f:\{0,1\}^n \to \mathbb{Z}_3$ be an arbitrary function and $\Phi_1,\Phi_2,\dots, \Phi_t:\Boo^n \to \mathbb{Z}_3$ be distinct linear polynomials satisfying the following properties:
    \begin{enumerate}
        \item For all $i\in [t]$, $\delta(f,\Phi_i) \le 1/2$.
        \item For all $i\ne j\in [t]$, $|\mathrm{vars}(\Phi_i-\Phi_j)| \ge 6$.
    \end{enumerate}
    Then $t\le 31$ i.e., a constant.
\end{lemma}

\begin{proof}
    Let $\omega \in \mathbb{C}$ be a primitive cube root of unity. Let $\mathbf{u} \in \mathbb{C}^{2^n}$ be defined as $u_\mathbf{x} = \omega^{f(\mathbf{x})}$ and for $i \in [t]$, define $\mathbf{v}^{(i)} \in \mathbb{C}^{2^n}$ as $v^{(i)}_\mathbf{x} = \omega^{\Phi_i(\mathbf{x})}$. As $f$ and $\Phi_i$ agree on at least $1/2$ fraction of points, we have
    \begin{align*}
        \Real\paren{\innerprod{\mathbf{u}, \mathbf{v}^{(i)}}}
         & = \sum_{\mathbf{x}} \Real \paren{\omega^{f(\mathbf{x})-\Phi_i(\mathbf{x})}}\\
        & \ge 2^{n-1} \cdot {\paren{1-1/2}} = 2^{n-2},
    \end{align*}
    where the last inequality uses the fact that $f(\mathbf{x})-\Phi_i(\mathbf{x})=0$ for at least $2^{n-1}$ choices of $\mathbf{x}$.
    For arbitrary $i\ne j\in [t]$, let $\Psi(\mathbf{x}):=\Phi_i(\mathbf{x}) - \Phi_j(\mathbf{x})$ be equal to $a_1 x_1 + a_2 x_2 + \dots + a_n x_n + a_0$ where $a_i \in \mathbb{Z}_3$ for $i\in[n]$ and the number of indices $i$ with $a_i\ne 0$ is at least $6$, by assumption. We now show that the vectors $\mathbf{v}^{(i)}$ and $\mathbf{v}^{(j)}$ are ``almost'' orthogonal:
    \begin{align*}
        \abs{\innerprod{\mathbf{v}^{(i)}, \mathbf{v}^{(j)}}} & = \abs{\sum_{\mathbf{x}} \omega^{\Phi_i(\mathbf{x})} \cdot {\omega^{-\Phi_j(\mathbf{x})}}}\\
        & = 2^n \cdot \abs{\E_{\mathbf{x}}[\omega^{a_1x_1 +a_2x_2 + \dots + a_nx_n + a_0}]}\\
        & = 2^n \cdot \prod_{i\in [n]} \abs{\E_{x_i}[\omega^{a_i x_i}]}\\
        & = 2^n \cdot \prod_{i\in [n]} \abs{\frac{1+\omega^{a_i}}{2}}\\
        & = 2^n \cdot \paren{\frac{1}{2}}^{\abs{\{i\in[n]:a_i \ne 0\}}} \le 2^{n-5}.
    \end{align*}
    Combined with the fact that $\mathbf{u}$ has a ``large'' component along $\mathbf{v}^{(i)}$ for every $i\in [t]$ i.e., $\Real\paren{\innerprod{\mathbf{u}, \mathbf{v}^{(i)}}} \ge 2^{n-2}$, this leads us to conclude that $t$ is a constant. To see this, let $\mathbf{w}^{(i)}:=\frac{1}{\sqrt{2^n}}\paren{\mathbf{v}^{(i)} - \mathbf{u}/4}$ for $i\in [t]$. Then we have
    \begin{align*}
    \innerprod{\mathbf{w}^{(i)},\mathbf{w}^{(i)}} = \frac{1}{2^n}\paren{2^n + 2^n/16 - \Real\paren{\innerprod{\mathbf{u}, \mathbf{v}^{(i)}}+\innerprod{\mathbf{v}^{(i)},\mathbf{u}}}/4} \le 1+1/16-1/8=15/16, 
    \end{align*} and
    \begin{align*}
        \Real\paren{\innerprod{\mathbf{w}^{(i)}, \mathbf{w}^{(j)}}} & = \frac{1}{2^n}\paren{\Real\paren{\innerprod{\mathbf{v}^{(i)}, \mathbf{v}^{(j)}}} + 2^n/16 - \Real\paren{\innerprod{\mathbf{u}, \mathbf{v}^{(i)}}+\innerprod{\mathbf{v}^{(i)}, \mathbf{u}}}/4}\\
        & \le 1/32 + 1/16 - 1/8 = -1/32.
    \end{align*}

    Thus, we have
    $$0 \le \innerprod{ \sum_{i=1}^t \mathbf{w}^{(i)}, \sum_{i=1}^t \mathbf{w}^{(i)}} = \sum_{i=1}^t \innerprod{\mathbf{w}^{(i)}, \mathbf{w}^{(i)}} + \sum_{i\ne j} \Real\paren{\innerprod{\mathbf{w}^{(i)}, \mathbf{w}^{(j)}}} \le 15t/16 - (t^2-t)/32.$$ Therefore $t\le 31$.
    
\end{proof}

Now we finish the proof of \Cref{claim:johnson-3-count} using \Cref{lemma:for-johnson-3}.
\begin{proof}[Proof of \Cref{claim:johnson-3-count}]
The proof idea is to follow the same initial strategy as for groups with order at least 5 from \Cref{subsec:5-groups}. To recall the setup, we have $\mathsf{List}_{\varepsilon}^{f} = \set{L_{1}, \ldots, L_{t}}$. We construct an undirected graph $\mathcal{G} = (V,E)$ with $|V| = t$: the $i^{th}$ vertex in $V$ corresponds to the polynomial $L_{i}$. We add an edge $(i,j)$ in $E$
iff $|\mathrm{vars}{(L_i - L_j)}| \leq 5$.

Note that \Cref{lemma:for-johnson-3} implies that there is no independent set size greater than $31$ in $\mathcal{G}$. Hence, by applying~\Cref{lem: large graph degree}, we conclude that there exists a vertex of degree at least $\Omega(t)$ in $\mathcal{G}$. Once we have such a vertex, we proceed in the same manner as in the proof of \Cref{thm:comb-prod-pgroups} to finally get an upper bound of $t\le \poly(1/\varepsilon)$ on the number of linear polynomials in $\mathsf{List}^{f}_\varepsilon$. 
\end{proof}

\subsubsection{Proof of Combinatorial Bound for a Product of $2$ and $3$-groups}
In this subsection, we will prove \Cref{thm:comb-prod-2groups}. We will use $q$ for $2$ or $3$ in the rest of the subsection. To prove \Cref{thm:comb-prod-2groups}, we need the following lemma, which is crucial to upper bound the list size using special-intersecting sets.\\

\begin{lemma}\label{lemma:induction-lemma-2-groups}
The following holds true for any finite $q$-group $G.$ 
 Let $f: \Boo^{n} \to G$ be any function. Let $\{P_1,\ldots, P_t\}$ be the set of polynomials in $\mathcal{P}_{1}(\Boo^{n}, G)$ that are $1/2$-close to $f.$ Then, if $\delta(f,P_i) = \frac{1}{2} - \alpha_i$ for each $i\in [t]$, we have
 \[
 \sum_{i=1}^t \alpha_i^D \leq 1
 \]
 for some absolute constant $D > 0.$ In particular, the number of $i$ such that $\delta(f,P_i)\leq \frac{1}{2}-\varepsilon$ is at most $(1/\varepsilon)^D.$\\
\end{lemma}
Note that \Cref{thm:comb-prod-2groups} follows immediately from \Cref{lemma:induction-lemma-2-groups} because a finite product of $q$-groups is a $q$-group. In the remaining part of this subsection, we are going to prove \Cref{lemma:induction-lemma-2-groups}.
\begin{proof}[Proof of \Cref{lemma:induction-lemma-2-groups}]
We will prove it via induction on the size of $G$. The constant $D$ is chosen as follows. Let $C_q$ be the constant in the extended Johnson bound (\Cref{lemma:johnson-3}), and let $D$ be the constant obtained from \Cref{lemma:special-intersection-bound} in the case of $(1/2,1/4, C_q)$-intersecting sets. 

\textbf{Base Case ($|G|=1$):} In this case, the lemma follows trivially.

\textbf{Induction Step:} Let $h \in G$ be an element of order $q$ in $G$ (the existence of such an element is guaranteed by Cauchy's Theorem for finite Abelian groups\footnote{Cauchy's theorem states that in a finite Abelian group $G$ where $|G|$ is divisible by a prime $p$, there exists an element of order $p$.}), and let $H = \langle h \rangle$ denote the subgroup generated by $h$.

Let $f':\{0,1\}^n\rightarrow G/H$ defined by
\[
f'(\mathbf{x}) = f(\mathbf{x})\pmod{H}
\]
Let $\{P_{1},\ldots, P_{t}\}$ be the set of linear polynomials in $\mathcal{P}_{1}(\Boo^{n}, G/H)$ that are $1/2$-close to $f'$. Assuming that $\delta(P_i,f') = \frac{1}{2}-\beta_i,$ we have by the induction hypothesis
\begin{equation}
\label{eq:indDGKS}
\sum_{i=1}^t \beta_i^D \leq 1.
\end{equation}

We now consider the polynomials in $\mathcal{P}_{1}(\Boo^{n}, G)$ that are $1/2$-close to $f.$ Given such a polynomial $Q$, we say that $Q$ \emph{extends} $P_i$ if $P_i = Q \pmod{H}$. Each such $Q$ extends a \emph{unique} $P_i$ ($i\in [t]$). 

Fix $P_i$ and assume that $Q_1,\ldots, Q_\ell$ are the polynomials in $\mathcal{P}_{1}(\Boo^{n}, G)$ that are $1/2$-close to $f$ and extend $P_i.$ Assuming that $\delta(f,Q_j) = \frac{1}{2}-\alpha_j$, we show that
\begin{equation}
\label{eq:indDGKS2}
\sum_{j=1}^\ell \alpha_j^D \leq \beta_i^D.
\end{equation}
Assuming \Cref{eq:indDGKS2}, we sum over all $i\in [t]$, and then using \Cref{eq:indDGKS}, we get the inductive statement.

So it suffices to prove \Cref{eq:indDGKS2}, and we will do this using properties of special intersection sets, in particular, \Cref{lemma:special-intersection-bound}. Fix $P_i$ and $Q_1,\ldots, Q_\ell$ as above for the rest of the proof. 

Let $S$ denote the agreement set of $P_i$ and $f'$, i.e.
\begin{align*}
    S := \setcond{\mathbf{x} \in \Boo^{n}}{ P_i(\mathbf{x}) = f'(\mathbf{x}) },
\end{align*}
where $\mu(S) = (1/2 + \beta_i)$. Similarly, let $S_{j}$ ($j\in [\ell]$) denote the agreement set of $Q_{j}$ and $f$, where $\mu(S) = (1/2 + \alpha_{j})$. Note that for every $j \in [\ell]$, $S_{j} \subseteq S$ since $Q_{j}$ extends $P_i$, which implies that $\cup_{j \in [\ell]} S_{j} \subseteq S$.

The core for the proof of \Cref{eq:indDGKS2} is to show that the sets $S_{1},\ldots, S_\ell$ form a special intersecting set family inside the universe $X = \{0,1\}^n$. We then use \Cref{lemma:special-intersection-bound}, and we upper bound the number of extensions. In particular, we prove the following claim.

\begin{claim}[Agreement sets are special intersecting sets]\label{claim:extensions-special-intersecting}
 The sets $S_{1},\ldots,S_{\ell}$ as defined above form a $(1/2,1/4,C_q)$-special intersecting sets, where $C_q$ is the constant from \Cref{lemma:johnson-3}. 
\end{claim}

Once we prove \Cref{claim:extensions-special-intersecting}, we can then apply \Cref{lemma:special-intersection-bound} on the sets $S_{1},\ldots,S_{\ell}$. Using the observation that $\cup_{j \in [\ell]} S_{j} \subseteq S$, we immediately get \Cref{eq:indDGKS2}, which finishes the proof of \Cref{lemma:induction-lemma-2-groups}.
\end{proof}
\begin{proof}[Proof of \Cref{claim:extensions-special-intersecting}]
We verify the four properties from the definition of special intersecting sets.

\begin{enumerate}
    \item For each $i \in [\ell]$, $\mu(S_{i}) \geq 1/2$. This is  true since each $Q_i$ is $1/2$-close to $f.$
    
    \item For any two distinct $i,j \in [\ell]$, $\mu(S_{i}, S_{j}) \leq 1/2$. This follows from the Schwartz-Zippel Lemma (\Cref{thm:basic}).
    
    \item Note that $H = \{0,h,\ldots, (q-1)h\}.$ We choose a set of coset representatives $c_1,\ldots, c_M$ ($M = |G|/|H|$) for $H$ in $G$. Now, given any $g\in G,$ we can write it uniquely as $c_p + s\cdot h$ where $p\in [M]$ and $s\in \{0,\ldots, q-1\}.$ 
    
    In particular, using this decomposition at each input $\mathbf{x}\in \{0,1\}^n$, we can write 
    \[
    f(\mathbf{x}) = \Tilde{f}(\mathbf{x}) + f'(\mathbf{x})\cdot h
    \]
    where $\Tilde{f}(\mathbf{x})$ is a coset representative and $f'(\mathbf{x})\in \{0,1,\ldots,q-1\}$ which we identify with $\mathbb{Z}_q.$

    Similarly, given an polynomial $Q(\mathbf{x}) = a_0 + \sum_{k=1}^n a_kx_k\in \mathcal{P}_1(\{0,1\}^n,G)$, we apply the above decomposition to each of its coefficients to write
    \[
    Q(\mathbf{x}) = \underbrace{\left(c_{p_0} + \sum_{i=1}^n c_{p_i}x_i\right)}_{\Tilde{Q}(\mathbf{x})} + \underbrace{\left(s_0 + \sum_{i=1}^n s_i x_i\right)}_{Q'(\mathbf{x})}\cdot h
    \]
    We treat the polynomial $Q'(\mathbf{x})$ as a polynomial over the group $\mathbb{Z}_q$. (This makes sense as the order of $h$ is $q$.)

    Returning to the polynomials $Q_1,\ldots, Q_\ell$, we note that if $Q_j(\mathbf{x}) = f(\mathbf{x}),$ then it must be true that $Q_j'(\mathbf{x}) = f'(\mathbf{x})$, implying that each $S_j$ is contained in $S_j' := \{\mathbf{x}\in \{0,1\}^n\ |\ Q_j'(\mathbf{x}) = f'(\mathbf{x})\}$. If we assume that $|S_j'| = \frac{1}{2} + \alpha_j'$, then we have
    \[
    \sum_{j=1}^\ell \alpha_j^{C_q}\leq \sum_{j=1}^\ell (\alpha_j')^{C_q} \leq 1
    \]
    where the final inequality is the extended Johnson bound (\Cref{lemma:johnson-3}).

    \item Let $I \subseteq [\ell]$ be a subset with $|I| \geq 3$ such that $\mu(S_{I}) > 1/4$. Let $T_{I}$ denote the following set
    \begin{align*}
        T_{I} = \setcond{\mathbf{x} \in \Boo^{n}}{Q_{j}(\mathbf{x}) = Q_{k}(\mathbf{x})\ \forall j,k\in I}
    \end{align*}
Observe that $S_{I} \subseteq T_{I}$, and hence we have $\mu(T_{I}) > 1/4$.
    
    We now note that the polynomials $Q_1,\ldots, Q_\ell$ are all equal modulo $H,$ implying that $\Tilde{Q}_1 = \cdots = \Tilde{Q}_\ell.$ In particular, we see that for $j\neq k\in I,$
    \[
    Q_j(\mathbf{x}) = Q_k(\mathbf{x}) \Longleftrightarrow Q_j'(\mathbf{x}) = Q_k'(\mathbf{x})
    \]
    where the latter equality is an equality of polynomials over $\mathbb{Z}_q.$ In other words, $T_{I}$ is the solution set in $\Boo^{n}$ to the following system of linear equations over $\mathbb{Z}_{q}$:
    \begin{align*}
        Q'_{jk}(\mathbf{x}) := Q_j'(\mathbf{x}) - Q_k'(\mathbf{x}) = 0, \quad \text{ for all } \ j,k \in I.
    \end{align*}
    Applying \Cref{claim:soln-dim-linear}, we see that the set of polynomials $\setcond{Q'_{jk}}{j,k\in I}$ are all integer multiples of a single linear polynomial. Call this polynomial $R(\mathbf{x})$.

    We are now ready to prove property 4. Fix any $J = \{j,k\}\subseteq I$. If $\mathbf{x}\in S_J,$ then $Q'_{jk}(\mathbf{x})=0$, implying that $R(\mathbf{x}) = 0.$ This implies that all the polynomials $Q_{r}$ ($r\in I$) take the same value at this point $\mathbf{x}.$ Moreover, since $\mathbf{x}\in S_J$, we have $Q_j(\mathbf{x}) = f(\mathbf{x})$, implying that $Q_r(\mathbf{x}) = f(\mathbf{x})$ for each $r\in I.$ This shows that $S_J\subseteq S_I.$ Since $S_I\subseteq S_J$ trivially, we have $S_I = S_J$ and we have thus shown property 4. 
\end{enumerate}
This shows that $S_{1},\ldots,S_{\ell}$ form a $(1/2,1/4,C_q)$-special intersecting sets.
\end{proof}
Now all that remains is to prove the following claim.
\begin{claim}[Dimension of system of linear equations]\label{claim:soln-dim-linear}
Let $q$ be a prime and let $\set{L_{i}(\mathbf{x}) = 0}_{i = 1}^{m}$ be a set of linear constraints over $\mathbb{Z}_{q}$. If the fraction of solutions of this set in $\Boo^{n}$ is $> 1/2^{r}$ i.e.
\begin{align*}
    |\setcond{\mathbf{x} \in \Boo^{n}}{L_{i}(\mathbf{x}) = 0, \text{ for all } \; i \in [m]}| > \dfrac{1}{2^{r}} \cdot 2^{n},
\end{align*}
then the dimension\footnote{Since $q$ is prime, we can measure the dimension of this set in the standard linear-algebraic sense.} of $\set{L_{i}(\mathbf{x})}_{i=1}^{m}$ is $< r$.
\end{claim}
\begin{proof}
Suppose $\dim\paren{\set{L_{i}(\mathbf{x})}_{i=1}^{m}} = r$. Let $M \in \mathbb{Z}_{q}^{m \times n}$ denote the coefficient matrix of the above system of equations, i.e. the $i^{th}$ row of $M$ denotes the coefficients of the variables in $L_{i}(\mathbf{x})$. Since the dimension is $r$, we know that $\dim(M)$ is either $r-1$ or $r$. 

If $\dim(M) = r-1$, then the system of equations $L_i(\mathbf{x}) = 0$ ($i\in [m]$) has no solution, implying that the claim is trivially true. So we assume that $\dim(M) = r.$

By doing elementary row operations, we can assume without loss of generality that the top leftmost $r \times r$ sub-matrix of $M$ is the $r$-dimensional identity matrix $I_{r}$. Let the new linear polynomials (after the elementary row transformations) be $L_{1}',\ldots, L_{m}'$. We then have the following property: For $i \in [r]$, $x_{i} \in \mathrm{vars}(L_{i}')$ and $x_{i} \notin \mathrm{vars}(L_{j}')$ for all $j \in [r]$ and $j \neq i$.

Now for any assignment of $x_{r+1},\ldots,x_{n} \in \Boo^{n-r}$, there exists at most one assignment of $x_{1},\ldots,x_{r} \in \Boo^{r}$ that satisfies the constraints. Thus the number of solutions in $\Boo^{n}$ is at most $2^{n-r}$. The claim follows.
\end{proof}

\section{Local List Correction Algorithm}
In this section, we prove \Cref{thm:listdecoding}, i.e. we construct a local list correction algorithm for $\mathcal{P}_{1}$. Our algorithm first constructs a list of deterministic oracles that are close to the polynomials in the list and then uses our local correction algorithm (see \Cref{thm:uniquedeg1}) on those oracles.

Let $G$ be an Abelian group. Let $f: \Boo^{n} \to G$ be any function. Let $\mathsf{List}_{\varepsilon}^{f}$ denote the set of degree $1$ polynomials that are $(1/2 - \varepsilon)$-close to $f$, and let $L(\varepsilon) = |\mathsf{List}_{\varepsilon}^{f}|$. Recall from \Cref{thm:comblistdeg1} that $L(\varepsilon) = \poly(1/\varepsilon) = \bigO_\varepsilon(1).$

We state the main construction of our local list-correction algorithm.

\begin{theorem}[Approximating oracles]\label{thm:approx-oracles-list-decoding}
Fix $n \in \mathbb{N}$ and $\varepsilon > 0$. There exists an algorithm $\mathcal{A}_{1}^{f}$ that makes at most $\bigO_{\varepsilon}(1)$ oracle queries and outputs deterministic algorithms $\psi_{1},\ldots, \psi_{L'}$ satisfying the following property: with probability at least $3/4$, for every polynomial $P \in \mathsf{List}_{\varepsilon}^{f}$, there exists a $j \in [L']$ such that $\delta(\psi_{j}, P) \leq 1/100$, and moreover, for every $\mathbf{x} \in \Boo^{n}$, $\psi_{j}$ computes $P(\mathbf{x})$ by making at most $\bigO_{\varepsilon}(1)$ oracle queries to $f$. Here $L' = \bigO(L(\varepsilon/2)\log L(\varepsilon/2)) = \bigO_\varepsilon(1)$.
\end{theorem}

Let us first see why the construction of approximating oracles is enough to prove \Cref{thm:listdecoding}.
\begin{proof}[Proof of \Cref{thm:listdecoding}]
We first run the algorithm given by \Cref{thm:approx-oracles-list-decoding} and it outputs $\psi_{1},\ldots, \psi_{L'}$. Next we run our local correction algorithm for $\mathcal{P}_{1}$ (see \Cref{thm:uniquedeg1} and \Cref{sec:deg-1-decoding}) with $\psi_{1}, \ldots, \psi_{L'}$ as oracles, and these algorithms will be $\phi_{1}, \ldots, \phi_{L'}$. This completes the description of the local list correction algorithm $\mathcal{A}^{f}$ for $\mathcal{P}_{1}$, and the bound on correctness probability follows from the correctness probability of \Cref{thm:uniquedeg1} and \Cref{thm:approx-oracles-list-decoding}.\newline
The algorithm $\mathcal{A}_{1}$ makes $\bigO_{\varepsilon}(1)$ queries to $f$ as stated in \Cref{thm:approx-oracles-list-decoding}, and then each $\phi_{i}$ makes $\bigO_{\varepsilon}(1) \cdot \Tilde{\bigO}(\log n) = \Tilde{\bigO}_\varepsilon(\log n)$ queries to $f$.
\end{proof}
In the rest of the section, we prove \Cref{thm:approx-oracles-list-decoding}.

\subsection{Overview of the Algorithms}
In this section, we give a bird's-eye view of the algorithms $\mathcal{A}_{1}$ and $\psi_{1},\ldots, \psi_{L'}$. As discussed in the proof overview, our algorithm is inspired by the list decoding algorithm for multivariate low-degree polynomials of \cite{STV-list-decoding}. We expand the discussion from the proof overview below. To use the same notation from the proof overview, let $S := \mathsf{List}_{\varepsilon}^{f}$.
\begin{enumerate}
    \item \textbf{Getting the advice :} This is a pre-processing step for the local list correction algorithm. The algorithm $\mathcal{A}_{1}^{f}$ queries $f$ on a random subcube $\mathsf{C}$ (see \Cref{defn:random-embedding}) of dimension $k = \bigO_\varepsilon(1)$. $\mathcal{A}_{1}^{f}$ will query on all of $\mathsf{C}$ and return all degree $d$ polynomials that are almost $(1/2 - \varepsilon)$-close to $f$ on $\mathsf{C}$. Denote the set of polynomials by $\Tilde{S} := \set{Q_{1},\ldots,Q_{L'}}$. Using a consequence of hypercontractivity and Chebyshev's inequality (see \Cref{lemma:sampling-subcube}), we will show that for any $P \in \mathsf{List}_{\varepsilon}^{f}$, with high probability (over the randomness of $\mathsf{C}$), there exists a $j \in [L']$ such that $P|_{\mathsf{C}} = Q_{j}$. Then $\psi_{j}$ will use $Q_{j}$ as advice. The details are described in \Cref{algo:list-decoding}.

    \item The algorithm $\psi_j$ works as follows.
    \begin{enumerate}
        \item \textbf{Computing a list of values:} For any input $\mathbf{b}$, $\psi_j$ will construct a subcube $\mathsf{C'}$ of dimension $2k$ containing $\mathsf{C}$ and $\mathbf{b}$. $\psi_{j}$ will query $f$ on $\mathsf{C'}$ and find all polynomials that are almost $(1/2 - \varepsilon/2)$-close to $f$ on $\mathsf{C'}$. Denote the set of polynomials by $S' := \set{R_{1},\ldots, R_{L'}}$. When the input $\mathbf{b}$ is random, we can show as in the previous step that, with high probability, for any $P \in \mathsf{List}_{\varepsilon}^{f}$, there exists a $i \in [L']$ such that $P|_{\mathsf{C'}} = R_{i}$.

        \item \textbf{Filtering the correct value using the advice:} $\psi_{j}$ has a list of polynomials $S'$, and it has to decide which of the polynomials in $S$ is equal to the restriction $P|_\mathsf{C'}$. To find this polynomial, $\psi_{j}$ will use the advice from Step 1 and check which of the polynomials from Step 2 is equal to the advice. Having found the correct polynomial $R_j$, the algorithm outputs the value of $R_j$ at the point $\mathbf{b}.$ The details of $\psi_j$ are described in \Cref{algo:high-agreement}.
        
        There is a subtlety here: To check whether a polynomial from $S'$ is equal to the advice or not, $\psi_{j}$ will restrict the polynomials in $S'$ to $\mathsf{C}$. Because of the underlying random process, this involves partitioning the variables $y_{1},\ldots,y_{2k}$ uniformly and randomly into pairs and identifying them. Under this random pairing, a polynomial $R_{j'} \in S'$ which is not equal to $P|_{\mathsf{C'}}$ may become equal to the advice. We will carefully upper bound the probability of this event in most cases by a combinatorial argument. The only bad case for this argument is when the difference polynomial $D:= R_{j'}-P|_{\mathsf{C}'}$ is of the form $\alpha\cdot (y_1+\cdots+y_{2k})$ where $\alpha\in G$ is an element of order $2.$ By setting $k$ to be even, we ensure that in this case $R_{j'}$ and $P|_{\mathsf{C}'}$ evaluate to the same value at $\mathbf{b}$ and hence it does not matter which of the polynomials is chosen to obtain the final output. 
    \end{enumerate}
\end{enumerate}

\noindent
Before we describe our algorithms, we need to describe a sub-routine - given an embedding of a subcube $\mathsf{C}$ and a point $\mathbf{b}$, we would like to find a small random subcube $\mathsf{C}'$ such that $\mathsf{C}$ is contained in $\mathsf{C}'$ and $\mathsf{C}'$ also contains $\mathbf{b}$.

\begin{definition}[Subcube spanned by $\mathsf{C}$ and $\mathbf{a}$]\label{def:bigger-subcube}
    Let $\mathsf{C} = C_{\mathbf{a}, h}$ be an embedding of a subcube of dimension $k$ (see \Cref{defn:random-embedding}). For any point $\mathbf{b} \in \Boo^{n}$, let $\mathbf{v} := \mathbf{a} \oplus \mathbf{b}$. Pick a uniformly random permutation $\sigma:[2k]\rightarrow [2k]$. Define a hash function $h' : [n] \to [2k]$ as follows: For all $i \in [n]$,
\begin{align*}
    h'(i) = \begin{cases}
        \sigma(j), & \text{if } \, h(i) = j \, \text{ and } \, v_{i} = 0 \\
        \sigma(j+k), & \text{if } \, h(i) = j \, \text{ and } \, v_{i} = 1.
    \end{cases}
\end{align*}
In other words, the partition of $[n]$ given by $h'$ is a refinement of the partition given by $h$, where the refinement depends on whether $\mathbf{b}$ and $\mathbf{a}$ agree on a coordinate.\newline
For every $\mathbf{z} \in \Boo^{2k}$, $x(\mathbf{z})$ is defined as follows:
\begin{align*}
    x(\mathbf{z})_{i} = z_{h'(i)} \oplus a_{i}.
\end{align*}
$\mathsf{C}^{\mathbf{b}}$ is the set of points $x(\mathbf{z})$ for all $\mathbf{z} \in \Boo^{2k}$, i.e. $\mathsf{C}^{\mathbf{b}} := \setcond{x(\mathbf{z})}{\mathbf{z} \in \Boo^{2k}}$.
\end{definition}
It is easy to verify from the definition that $\mathsf{C} \subset \mathsf{C}^{\mathbf{b}}$, and also $\mathbf{b} \in \mathsf{C}^{\mathbf{b}}$. In particular, say we define $\mathbf{w} \in \Boo^{2k}$ as follows: for $j \in [k]$, $w_{\sigma(j)} = 0$ and $w_{\sigma(j+k)} = 1$. Then $x(\mathbf{w}) = \mathbf{b}$.

\begin{observation}\label{obs:Cb-is-random}
Let $h$ be a random hash function from $[n]$ to $[k]$ and $\mathbf{a} \sim \Boo^{n}$. Then for a random $\mathbf{b} \sim \Boo^{n}$, $h'$ as defined above is a random hash function from $[n]$ to $[2k]$ - this follows because for a random $\mathbf{b} \sim \Boo^{n}$, $\mathbf{v}$ is uniformly distributed in $\Boo^{n}$ and independent of $\mathbf{a}$. This means that $\mathsf{C}^{\mathbf{b}}$ as defined above is a random embedding of a subcube of dimension $2k$ (see \Cref{sec:prelims} just after \Cref{defn:random-embedding}).

Furthermore, conditioned on the choice of $\mathsf{C}'$ (i.e. the choice of $\mathbf{a},h'$), the subcube $\mathsf{C}$ may be described as follows: we partition the variables $z_1,\ldots, z_{2k}$ into pairs uniformly at random and identify the variables in each pair. 

Finally, note that $\mathbf{b} = x(\mathbf{w})$ for some $\mathbf{w}$ of Hamming weight exactly $k$. 
\end{observation}

\subsection{The Algorithms}

Let $\mathsf{C}$ be a subcube of dimension $k$ and $Q: \{0,1\}^k \to G$  a degree-$1$ polynomial, which we consider to be a function on $\mathsf{C}$. We will use $\psi_{\mathsf{C},\sigma, Q}$ to denote a \emph{deterministic}  algorithm that has the description of $\mathsf{C}$, a permutation $\sigma:[2k]\rightarrow [2k],$ and evaluation of $Q$ on $\mathsf{C}$ hardwired inside it\footnote{In the final algorithm, $\mathsf{C}$ will be a random subcube of dimension $\mathrm{poly}(1/\varepsilon)$ and $Q$ with high probability be equal to $P|_{\mathsf{C}}$, for some $P \in \mathsf{List}_{\varepsilon}^{f}$}. The description of algorithm $\psi_{\mathsf{C},\sigma, Q}$ follows.

\begin{algobox}
\begin{algorithm}[H]
\caption{Approximating algorithm $\psi_{\mathsf{C},\sigma,Q}$}
\label{algo:high-agreement}
\DontPrintSemicolon

\KwIn{Oracle access to $f$, $\mathbf{b} \in \Boo^{n}$}

Let $\mathsf{C}'$ be a bigger subcube containing $\mathsf{C}$ and $\mathbf{b}$ constructed using $\sigma$ as in \Cref{def:bigger-subcube} \;
Let $\mathbf{w} \in \Boo^{2k}$ such that $x(\mathbf{w}) = \mathbf{b}$ \tcp*{$|\mathbf{w}| = k$}
Query $f$ on $\mathsf{C}'$ \tcp*{Number of queries is $2^{2k}$}
Use the algorithm in \Cref{thm:non-local-list} to find all polynomials $R_{1},\ldots, R_{L''} \in \mathcal{P}_{1}(\Boo^{2k}, G)$ that are $\paren{\frac{1}{2} - \frac{\varepsilon}{2}}$-close to $f|_{\mathsf{C'}}$ \tcp*{$L'' \leq L(\varepsilon/2)$} \label{line:brute-force}
\If{there exists an $i \in [L'']$ such that $R_{i}|_{\mathsf{C}} = Q$}{
pick any such $i$ and \Return{$R_{i}(\mathbf{w})$}
}
\Else{
\Return{$0$}\tcp*{An arbitrary value}
}

\end{algorithm}
\end{algobox}

Now we can state the algorithm $\mathcal{A}_{1}$.

\begin{algobox}
\begin{algorithm}[H]
\caption{Algorithm $\mathcal{A}_{1}$}
\label{algo:list-decoding}

\DontPrintSemicolon

\KwIn{Oracle access to $f$}

Choose $k \leftarrow 2\cdot L(\varepsilon/2)^3\cdot \lceil 1/\varepsilon^5\rceil$\tcp*{$k$ is even (will be crucial later)}
Set $\ell \leftarrow \log L(\varepsilon)$\;
$T \leftarrow \emptyset$\;
\Repeat{$\ell$ times}{Sample random $\mathbf{a} \sim U_{n}$ and $h$ as a random hash function from $[n]$ to $[k]$ \;
Construct the subcube $\mathsf{C} := C_{\mathbf{a}, h}$ according to \Cref{defn:random-embedding}\;
Query $f$ on $\mathsf{C}$ \tcp*{Number of queries is $2^{k}$}
Use the algorithm in \Cref{thm:non-local-list} to find all polynomials $Q_{1},\ldots,Q_{L'} \in \mathcal{P}_{1}(\Boo^{k}, G)$ that are $\paren{\frac{1}{2} - \frac{\varepsilon}{2}}$-close to $f|_{\mathsf{C}}$ \tcp*{$L' \leq L(\varepsilon/2)$} \label{line:brute-force-2}
$T \leftarrow T \cup \set{Q_{1},\ldots,Q_{L'}}$
}
Pick a uniformly random permutation $\sigma:[2k]\rightarrow [2k]$\;
\Return{$\psi_{\mathsf{C},\sigma, Q_{1}},\ldots, \psi_{\mathsf{C},\sigma,Q_{t}}$ for all $Q_{i} \in T$} 

\end{algorithm}
\end{algobox}

\subsection{Analysis of the Algorithms}
In this section, we will prove \Cref{thm:approx-oracles-list-decoding} by analyzing the query complexity and the error probability.
\paragraph{Query complexity:}The algorithm $\mathcal{A}_{1}$ makes $2^{k} = \bigO_\varepsilon(1)$ queries to $f$ to output approximating oracles $\psi_{1},\ldots,\psi_{t}$. Each approximating oracle $\psi_{i}$ makes $2^{2k} = \bigO_\varepsilon(1)$ queries to return the evaluation at a point $\mathbf{b}$. 

\paragraph{Correctness:}We want to show that with probability $\geq 3/4$, for every $P \in \mathsf{List}_{\varepsilon}^{f}$, there exists an oracle $\psi_{\mathsf{C}, \sigma,Q_{j}}$ such that $\delta(\psi_{\mathsf{C}, \sigma,Q_{j}}, P) \leq 1/100$. We first show that in a single iteration for \Cref{algo:list-decoding} the following holds: for every polynomial $P \in \mathsf{List}_{\varepsilon}^{f}$, with probability at least $9/10$, there exists a $1/100$-close approximating oracle $\psi_{j}$. We prove this is in \Cref{lemma:list-correction-error}. Since we repeat this $\ell$ times, the probability that there is no $1/100$-close approximating oracle for $P$ is at most $1/10^{\ell}$. By a union bound for all polynomials $P \in \mathsf{List}_{\varepsilon}^{f}$, we get the desired correctness probability in \Cref{thm:approx-oracles-list-decoding}. Since each iteration produces a list of size at most $L(\varepsilon/2),$ overall we obtain a list of size $\bigO(L(\varepsilon/2)\cdot \log L(\varepsilon))$ as claimed.

\begin{lemma}[Correctness of Local List Correction]\label{lemma:list-correction-error}
Fix any polynomial $P \in \mathsf{List}_{\varepsilon}^{f}$.  Then the probability (over the randomness of 
\Cref{algo:list-decoding}) that there does not exist a $j$ such that $\delta(\psi_{\mathsf{C},\sigma, Q_{j}}, P) \leq 1/100$ is at most $1/10$.
\end{lemma}
\begin{proof}
Let $\mathcal{E}_{P}$ denote the event that there does not exist a $j$ such that $\delta(\psi_{\mathsf{C},\sigma, Q_{j}}, P) \leq 1/100$. We want to bound the probability of event $\mathcal{E}_P.$ We will show that 
\begin{equation}
\label{eq:EP}
    \E_{\mathbf{a},h,\sigma}[\min_{j}\delta(\psi_{\mathsf{C},\sigma,Q_j}, P)] = \E_{\mathbf{a},h,\sigma}[\min_{j}\Pr_{\mathbf{b}}[\psi_{\mathsf{C},\sigma,Q_j}(\mathbf{b})\neq P(\mathbf{b})]] \leq 1/1000
\end{equation}
from which the lemma follows via an application of Markov's inequality.

Define the following auxiliary events, depending on the choice of $\mathbf{a},h$ and $\sigma$, along with the choice of a random input $\mathbf{b}.$
\begin{enumerate}
    \item Event $\mathcal{E}_{1,P}$ (only depends on $\mathbf{a},h$): In \Cref{algo:list-decoding}, there does not exist a polynomial $Q_{j}$ such that $Q_{j} = P|_{\mathsf{C}}$.

    \item Event $\mathcal{E}_{2,P}$: In \Cref{algo:high-agreement}, there does not exist a polynomial $R_{i}$ such that $R_{i} = P|_{\mathsf{C'}}$.

    \item Event $\mathcal{E}_{3,P}$: In \Cref{algo:high-agreement}, there exist two polynomials $R_{i_1}$ and $R_{i_2}$ such that  $R_{i_1}(\mathbf{w})\neq R_{i_2}(\mathbf{w})$ but $R_{i_1}|_\mathsf{C} = R_{i_2}|_\mathsf{C}$. Here $\mathbf{w}$ is, as defined in \Cref{algo:high-agreement}, the point in $\{0,1\}^{2k}$ of Hamming weight $k$ such that $x(\mathbf{w}) = \mathbf{b}.$
\end{enumerate}

To see how these events are useful in analyzing \Cref{eq:EP}, we proceed as follows. For $\mathbf{a}, h$ such that the event $\mathcal{E}_{1,P}$ does not occur, we can fix a $j^*\leq L'$ such that $P|_{\mathsf{C}} = Q_{j^*}.$  Thus, we have
\begin{equation}
\label{eq:E1P}
\E_{\mathbf{a},h,\sigma}[\min_{j}\Pr_{\mathbf{b}}[\psi_{\mathsf{C},\sigma,Q_j}(\mathbf{b})\neq P(\mathbf{b})]\leq \Pr_{\mathbf{a},h}[\mathcal{E}_{1,P}] + \E_{\mathbf{a},h,\sigma}[\mathbf{1}_{\neg \mathcal{E}_{1,P}}\cdot \Pr_{\mathbf{b}}[\psi_{\mathsf{C},\sigma,Q_{j^*}}(\mathbf{b})\neq P(\mathbf{b}) ]]
\end{equation}
Fix any $\mathbf{a}, h$ such that the event $\mathcal{E}_{1,P}$ does not occur. Further, if the event $\mathcal{E}_{2,P}$ does not occur, then there is an $i^*\leq L''$ such that $P|_{\mathsf{C}'} = R_{i^*}.$ In particular, $R_{i^*}|_\mathsf{C} = P|_{\mathsf{C}} = Q_{j^*}.$ 

Finally, if event $\mathcal{E}_{3,P}$ also does not occur, then there is no $i\neq i^*$ such that $R_{i^*}(\mathbf{w}) \neq R_i(\mathbf{w})$ but $R_{i}|_{\mathsf{C}} = R_{i^*}|_\mathsf{C}.$ In particular, the only possible output of the algorithm $\psi_{\mathsf{C},\sigma,Q_{j^*}}$ on input $\mathbf{w}$ is $R_{i^*}(\mathbf{w}) = P(x(\mathbf{w})) = P(\mathbf{b}).$

We have thus shown that 
\[
\E_{\mathbf{a},h,\sigma}[\mathbf{1}_{\neg \mathcal{E}_{1,P}}\cdot \Pr_{\mathbf{b}}[\psi_{\mathsf{C},\sigma,Q_{j^*}}(\mathbf{b})\neq P(\mathbf{b}) ]] \leq 
\Pr_{\mathbf{a},h,\sigma,\mathbf{b}}[\mathcal{E}_{2,P}\vee \mathcal{E}_{3,P}] \leq \Pr_{\mathbf{a},h,\sigma,\mathbf{b}}[\mathcal{E}_{2,P}] + \Pr_{\mathbf{a},h,\sigma,\mathbf{b}}[\mathcal{E}_{3,P}].
\]
Plugging the above into \Cref{eq:E1P}, we get
\begin{equation}
    \label{eq:E1P2P3P}
    \E_{\mathbf{a},h,\sigma}[\min_{j}\Pr_{\mathbf{b}}[\psi_{\mathsf{C},\sigma,Q_j}(\mathbf{b})\neq P(\mathbf{b})]\leq \Pr_{\mathbf{a},h}[\mathcal{E}_{1,P}]+\Pr_{\mathbf{a},h,\sigma,\mathbf{b}}[\mathcal{E}_{2,P}] + \Pr_{\mathbf{a},h,\sigma,\mathbf{b}}[\mathcal{E}_{3,P}].
\end{equation}
So it now suffices to bound the probabilities of the events $\mathcal{E}_{1,P}, \mathcal{E}_{2,P}$ and $\mathcal{E}_{3,P}$, which we do in the following two claims. 

\begin{claim}\label{claim:l-restriction-far}
 $\Pr_{\mathbf{a},h}[\mathcal{E}_{1,P}], \Pr_{\mathbf{a},h,\sigma,\mathbf{b}}[\mathcal{E}_{2,P}] \leq 1/10000.$
\end{claim}

\begin{claim}\label{claim:E3P}
$\Pr_{\mathbf{a},h,\sigma, \mathbf{b}}[\mathcal{E}_{3,P}]\leq 1/10000.$
\end{claim}

Substituting the above bounds into \Cref{eq:E1P2P3P}, we get \Cref{eq:EP}, implying the statement of the lemma. So it suffices to prove \Cref{claim:l-restriction-far} and \Cref{claim:E3P}.

\begin{proof}[Proof of \Cref{claim:l-restriction-far}]
Recall that $\delta(P,f) \leq (1/2 - \varepsilon).$ Equivalently, the set of points $T$ where $f$ and $P$ differ has density at most $(1/2-\varepsilon)$ in $\{0,1\}^n.$ For a cube $\mathsf{C}$, the non-existence of $Q_{j}$ such that $Q_{j} = P|_{\mathsf{C}}$ is equivalent to $\delta(P|_{\mathsf{C}}, f|_{\mathsf{C}}) > (1/2 - \varepsilon/2)$.

For a random subcube $\mathsf{C}$ with this distribution, using \Cref{lemma:sampling-subcube} for $T$ as above, we get that for $k\geq 1/\varepsilon^5$,
\begin{align*}
    \Pr_{\mathsf{C}}\brac{\delta(P|_{\mathsf{C}}, f|_{\mathsf{C}}) > \frac{1}{2}-\frac{\varepsilon}{2}} \leq 1/10000.
\end{align*}
(Here, we are assuming, without loss of generality that $\varepsilon$ is less than or equal to a small enough constant so that any $k \geq 1/\varepsilon^5$ satisfies the hypothesis of \Cref{lemma:sampling-subcube}.) Hence $\Pr[\mathcal{E}_{1,P}] \leq 1/10000$.

Using \Cref{obs:Cb-is-random}, we know that $\mathsf{C'}$ is also a random subcube of dimension $2k$ (drawn from a similar distribution). Proceeding as above, we get the stated upper bound on $\Pr[\mathcal{E}_{2,P}]$. 
\end{proof}

\begin{proof}[Proof of \Cref{claim:E3P}]
To bound this probability, we first note that the polynomials $R_1,\ldots, R_{L''}$ are determined by the subcube $\mathsf{C}'.$ We condition on a fixed choice of $\mathsf{C}'.$ Recall that, as noted in \Cref{obs:Cb-is-random}, the subcube $\mathsf{C}$ is obtained from $\mathsf{C}'$ by partitioning the $2k$ variables in $\mathsf{C}'$ into $k$ pairs uniformly at random and identifying the variables in each pair. 

We denote by $y_1,\ldots, y_{2k}$ the variables of a polynomial $R$ defined on $\mathsf{C}'.$

Fix any pair of degree-$1$ polynomials $R_{i_1}$ and $R_{i_2}$ that are found in \Cref{algo:high-agreement} such that $R_{i_1}(\mathbf{w}) \neq R_{i_2}(\mathbf{w}).$ Let $D_{i_1,i_2} = R_{i_1}-R_{i_2}$, which evaluates to a  non-zero value at the point $\mathbf{w}.$ We consider the event  $D_{i_1,i_2}|_\mathsf{C} = 0.$ We will show that
\begin{equation}
    \label{eq:Di1i2}
    \Pr[D_{i_1,i_2}|_\mathsf{C} = 0] \leq \frac{\log k}{k}.
\end{equation}
Assuming the above, we can use a union bound over all pairs $i_1,i_2\in [L'']$ such that $R_{i_1}(\mathbf{w})\neq R_{i_2}(\mathbf{w})$ to get
\[
\prob{}{\mathcal{E}_{3,P}}\leq (L'')^2 \cdot \frac{\log k}{k}\leq L(\varepsilon/2)^2\cdot \frac{\log k}{k} \leq \frac{1}{10000}
\]
where the last inequality follows from our choice of $k$ (we assume without loss of generality that $\varepsilon$ is less than a small enough absolute constant). This shows that \Cref{eq:Di1i2} implies the claim.

We now show \Cref{eq:Di1i2}. Assume that
\begin{align*}
    D_{i_1,i_2}(y_{1}, \ldots, y_{2k}) = \alpha_{0} + \sum_{i=1}^{2k} \alpha_{i} y_{i}
\end{align*}
where $\alpha_0,\ldots, \alpha_{2k}\in G.$ 

If $\alpha_{0} \neq 0$, then clearly $D_{i_1,i_2}|_{\mathsf{C}} \neq 0$ since pairing and identifying variables does not change the constant term. So we can assume without loss of generality that $\alpha_{0} = 0$.

 Let $\alpha$ denote the plurality of the coefficients $\alpha_{1},\ldots,\alpha_{2k}$ of $D_{i_1,i_2}$ (breaking ties arbitrarily). Let $W \subseteq [2k]$ index the subset of coefficients that are $\alpha$, i.e. $W := \setcond{j \in [2k]}{\alpha_{j} = \alpha}$. We have the following two cases depending on the order of $\alpha$.
\begin{enumerate}
    \item $\alpha\neq -\alpha$: We have two further cases depending on the size of $W$.
    \begin{itemize}
        \item $|W| \geq \log k$: Note that for the polynomial $D_{i_1,i_2}|_{\mathsf{C}}$ to vanish, it must be the case that each of the variables indexed by $W$ is mapped to a variable with coefficient $-\alpha \neq \alpha.$ Since $-\alpha$ is \emph{not} the plurality, it follows that at most half the variables have coefficient $-\alpha.$ Hence, we have
        \begin{align*}
            \Pr[D_{i_1,i_2}|_\mathsf{C} = 0]&\leq \Pr[\forall \; j \in W:\  \text{$y_j$ is paired with a variable with coefficient $-\alpha$}]\\ 
            &\leq (1/2)^{\log k} = 1/k,
        \end{align*}
        which implies \Cref{eq:Di1i2} in this case.

        \item $|W| < \log k$: Since $\alpha$ is the most frequently appearing coefficient, this implies that no coefficient appears more than $\log k$ times. Specifically, this also applies to $-\alpha$. Hence, for any $j\in W$, the probability that $y_j$ is mapped to a variable with coefficient $-\alpha$ is at most $(\log k)/k,$ implying \Cref{eq:Di1i2} in this case.
    \end{itemize}

    \item $\alpha = -\alpha$ (i.e. $2\alpha = 0$): We have two further sub-cases depending on the size of $\overline{W}$, the complement of $W.$
    \begin{itemize}
        \item $|\overline{W}| \geq \log k$: For the polynomial $D_{i_1,i_2}|_{\mathsf{C}}$ to vanish, each variable indexed by a $j\in \overline{W}$ with coefficient $\alpha_j\neq \alpha$ must be paired with another variable with coefficient $-\alpha_j$. Since $\alpha_j\neq \alpha,$ we also see that $-\alpha_j\neq \alpha$ and hence $-\alpha_j$ is the coefficient of at most half the variables. Thus, we can upper bound the probability that $D_{i_1,i_2}|_\mathsf{C}$ vanishes as follows.
        \begin{align*}
            \Pr_{\sigma}[\forall \; i \in \overline{W} \, : \text{$y_j$ is paired with a variable with coefficient $-\alpha_j$} ] \leq (1/2)^{\log k} = 1/k.
        \end{align*}

        \item $|\overline{W}| < \log k$: In this case, we note that $\overline{W}\neq \emptyset$, and the reason is as follows. $D_{i_1,i_2}(\mathbf{w})\neq 0$ by assumption. On the other hand, we know that $|\mathbf{w}| = k$ and we have chosen $k$ to be even (see \Cref{algo:list-decoding}). Thus $\alpha\cdot (w_1+\cdots + w_{2k}) = 0$. Hence, we cannot have $D_{i_1,i_2} = \alpha\cdot (y_1+\cdots+ y_{2k}).$ Thus, there must be some variable $y_j$ whose coefficient is not $\alpha.$ 
        
        Fix any such $j\in \overline{W}$. For the polynomial $D_{i_1,i_2}|_\mathsf{C}$ to vanish, the variable $y_j$ must be paired with another variable with coefficient $-\alpha_j$, and any such variable is also indexed by an element of $\overline{W}.$ Since $|\overline{W}| < \log k,$ the probability of this is at most $(\log k)/k,$ implying the claimed probability bound.
    \end{itemize}
\end{enumerate}
We have thus shown \Cref{eq:Di1i2}, implying the proof of the claim.
\end{proof}
As discussed above, we have proved \Cref{claim:l-restriction-far} and \Cref{claim:E3P} and substituting them in \Cref{eq:E1P2P3P}, we get the desired bound, and this concludes the correctness of the local list correction algorithm.
\end{proof}

\bibliographystyle{alpha}
\bibliography{references}

\appendix

\section{Non-Local Algorithms for Decoding Low-Degree Polynomials}
\label{app:non-local}

In this section, we prove the following results regarding unique and list decoding algorithms for $\mathcal{P}_d$ over an arbitrary Abelian group $G.$ We assume throughout that group operations (addition, inverse etc.) and comparing group elements can be done in constant time.

\begin{theorem}[essentially due to Reed~\cite{Reed}]
\label{thm:non-local-unique}
Fix any Abelian group $G.$ There is a $\poly(2^n)$-time algorithm that, given oracle access to a function $f:\{0,1\}^n\rightarrow G$ produces the unique polynomial $P\in \mathcal{P}_d$ such that $\delta(f,P) < 1/2^{d+1}$, assuming that such a $P$ exists.
\end{theorem}

\begin{theorem}
\label{thm:non-local-list}
Fix any Abelian group $G$ and degree parameter $d.$ There is a $\poly(2^{n^{d+1}})$-time algorithm that, given oracle access to a function $f:\{0,1\}^n\rightarrow G$ produces a list of all polynomials $P\in \mathcal{P}_d$ such that $\delta(f,P) < 1/2^{d}$.
\end{theorem}

\begin{remark}
We will use \Cref{thm:non-local-list} in the setting for $d = 1$ and output polynomials $P$ such that when $\delta(f, P) > (1/2 - \varepsilon)$, for some $\varepsilon > 0$. Our algorithm will output a list $\mathcal{L} \subset \mathcal{P}_{1}$ of size $2^{\bigO(n^{2})}$. $\mathcal{L}$ may contain polynomials that are not $(1/2 - \varepsilon)$-close to $f$ and we would prune $\mathcal{L}$ to remove these polynomials. To do this, we simply compute $\delta(f,R)$ for every $R \in \mathcal{L}$, and remove $R$ if $\delta(f,R) > (1/2-\varepsilon)$. This can be done in time $\bigO(2^{n} \cdot |\mathcal{L}|) = \bigO(2^{\bigO(n^{2})})$, and this adds up to the time stated in \Cref{thm:non-local-list}.
\end{remark}

\subsection{Proof Sketch of \Cref{thm:non-local-unique}}

We only give a sketch here, because the algorithm and proof of correctness are almost identical to the Majority-logic decoding algorithm of Reed~\cite{Reed} (see also~\cite[Chapter 14]{GRS-codingbook}).

We need the following lemma, which follows immediately from M\"{o}bius Inversion (item 1 in \Cref{thm:basic}). 

\begin{lemma}
    \label{lem:find-coeff}
    Fix $P\in \mathcal{P}_d(\{0,1\}^n,G)$ and $I\subseteq [n]$ of size $[d].$ Let $\mathbf{a}\in \{0,1\}^{n-d}$ be any assignment to the variables outside $I.$ Then, the coefficient $c_I$ of $\prod_{i\in I}x_i$ in $P$ is given by 
    \begin{equation}
    \label{eq:reed-mobius}
    c_I = \sum_{J\subseteq I} (-1)^{|I\setminus J|} P(1_J \circ \mathbf{a})
    \end{equation}
    where $1_J \circ \mathbf{a}$ denotes the input $\mathbf{b}\in \{0,1\}^n$ that agrees with the indicator vector of $J$ on co-ordinates inside $I$ and the fixed input $\mathbf{a}$ on co-ordinates outside $I.$
\end{lemma}

\begin{proof}
    This follows directly from item 1 of \Cref{thm:basic} applied to the restriction of $P$ obtained by setting the variables outside $I$ according to the values assigned by $\mathbf{a}.$ Note that this restriction does not change the coefficient of the monomial $\prod_{i\in I}x_i$ (though it can change other coefficients). Hence, the coefficient of the restricted polynomial is equal to $c_I.$
\end{proof}

We can now sketch Reed's Majority-logic algorithm in this setting. Assume that we are given $f$ such that $\delta(f,P) < 1/2^{d+1}$. For each $I\subseteq [n]$ of size at most $d$, let $c_I$ denote the coefficient of the monomial $\prod_{i\in I}x_i$ in $P.$

\paragraph{Finding $c_I$ for $|I| = d$.} Fix any $I$ of size $[d].$ For each setting $\mathbf{a}\in \{0,1\}^{n-d}$, compute 
\begin{equation}
\label{eq:reed-find-coeff}
c_{I,\mathbf{a}} = \sum_{J\subseteq I} (-1)^{|I\setminus J|} f(1_J \circ \mathbf{a})
\end{equation}
where $1_J\circ \mathbf{a}$ is as defined in the statement of \Cref{lem:find-coeff}. Among these $2^{n-d}$ many group elements, output the most commonly occurring one.

\paragraph{Correctness.} Since $\delta(f,P) < 1/2^{d+1},$ it follows that for strictly more than half the possibilities for $\mathbf{a}\in \{0,1\}^{n-d},$ the function $f$ agrees with $P$ on all inputs in the subcube $C_\mathbf{a}$ obtained by setting variables outside $I$ according to $\mathbf{a}$  (not to be confused with the kinds of subcubes defined in \Cref{defn:random-embedding}). This implies that $c_{I,\mathbf{a}} = c_I$ for all such $\mathbf{a}.$ Hence, the mostly commonly occurring value among the $c_{I,\mathbf{a}}$ ($\mathbf{a}\in \{0,1\}^{n-d}$) is the right coefficient $c_I$.

\paragraph{Recursion to find other coefficients.} After finding all the coefficients of monomials of degree $d,$ we replace the function $f$ with $f'$ where
\[
f'(\mathbf{x}) = f(\mathbf{x}) - \sum_{|I| = d} c_I \prod_{i\in I} x_i.
\]
We then apply a recursive procedure to $f'$ with $d$ replaced by $d-1$. 

\paragraph{Correctness.} Define analogously a  $P'\in \mathcal{P}_{d-1}(\{0,1\}^n,G)$ by
\[
P'(\mathbf{x}) = P(\mathbf{x}) - \sum_{|I| = d} c_I \prod_{i\in I} x_i = \sum_{|I| < d} c_I \prod_{i\in I} x_i.
\]
Note that $\delta(f',P') = \delta(f,P) < 1/2^{d+1} < 1/2^d.$ Hence, by induction, the recursive procedure correctly finds $c_I$ for $|I| < d.$

\paragraph{Running time.} The running time is easily analyzed to be $(d+1)\cdot 2^{O(n)} = 2^{O(n)}.$

\subsection{Proof Sketch of \Cref{thm:non-local-list}}

The algorithm is similar to the Majority-logic algorithm above, except that in the first step, we now find a large list of polynomials.

Assume that we are given $f:\{0,1\}^n\rightarrow G$. Below, $P$ will denote any degree $d$ polynomial such that $\delta(f,P) < 1/2^{d}$. The coefficients $c_I$ of $P$ are as defined in the previous section.

We describe the algorithm in analogy with the algorithm from the previous section. The difference is that in the first step, we find a \emph{list} of homogeneous polynomials of degree $d$ such that one of them is exactly the homogeneous component of the polynomial $P.$

\paragraph{Finding $c_I$ for $|I| = d$.} For $I$ of size $d$ and $\mathbf{a}\in \{0,1\}^{n-d}$, define $c_{I,\mathbf{a}}$ as in \Cref{eq:reed-find-coeff}. 

Now, define $N := \binom{n}{d}$ and for each tuple $\overline{\mathbf{a}}=(\mathbf{a}^{(1)},\ldots, \mathbf{a}^{(N)})\in (\{0,1\}^{n-d})^N$, compute the polynomial
\[
P_{\overline{\mathbf{a}}}(\mathbf{x}) = 
\sum_{j=1}^N c_{I_j,\mathbf{a}^{(j)}} \prod_{i\in I_j} x_i
\]
where $I_1,\ldots, I_N$ is some ordering of all subsets of $[n]$ of size $d$.

Let $\mathcal{L}_d$ denote the set of all such polynomials.

\paragraph{Correctness.} Since $\delta(f,P) < 1/2^d$, it follows that for each $I$ such that $|I|=d$, there is at least one $\mathbf{a}\in \{0,1\}^{n-d}$ such that $f$ agrees with $P$ on all inputs in the subcube $C_{\mathbf{a}}$ (as defined above). This implies that there is at least one choice for the tuple $\overline{\mathbf{a}}$ such that $P_{\overline{\mathbf{a}}}$ is exactly the homogeneous component of $P$ of degree $d$.

\paragraph{Coefficients of smaller degree.} For each polynomial $P_{\overline{\mathbf{a}}}\in \mathcal{L}_d$, we define $f_{\overline{\mathbf{a}}}$ by
\[
f_{\overline{\mathbf{a}}}(\mathbf{x}) = f(\mathbf{x}) - P_{\overline{\mathbf{a}}}(\mathbf{x})
\]
We now use the algorithm from \Cref{thm:non-local-unique} on $f_{\overline{\mathbf{a}}}$ to find the unique polynomial $Q_{\overline{\mathbf{a}}}$ of degree at most $d-1$ that is at distance less than $1/2^{d}$ from $f_{\overline{\mathbf{a}}}$. Finally, we output the list of polynomials $\mathcal{L}$ where 
\[
\mathcal{L} = \{P_{\overline{\mathbf{a}}} + Q_{\overline{\mathbf{a}}}\ |\ \overline{\mathbf{a}}\}.\]

\paragraph{Correctness.} Define $P'$ of degree $d-1$ as in the previous section. For any $\overline{\mathbf{a}}$ such that $P_{\overline{\mathbf{a}}}$ is equal to the homogeneous component of degree $d$ in $P,$ we see that $\delta(f_{\overline{\mathbf{a}}}, P') = \delta(f,P) < 1/2^{d}.$ By \Cref{thm:non-local-unique}, we see that the polynomial $Q_{\overline{\mathbf{a}}}$ is equal to $P'.$ This implies that the list $\mathcal{L}$ contains the polynomial $P.$ Since $P$ was an arbitrary degree $d$ polynomial such that $\delta(f,P) < 1/2^{d},$ the list $\mathcal{L}$ contains all such polynomials.

\paragraph{Running time.} Note that the number of $\overline{\mathbf{a}}$ is at most $(2^n)^N\leq 2^{n^{d+1}}$. As the above algorithm runs in time $2^{O(n)}$ for each choice of $\overline{\mathbf{a}}$, the overall running time is $2^{O(n^{d+1})}.$

\section{Improved Local Correction over Reals}
\label{app:improved-uniquedeg1-smallerror}

In this section, we improve~\Cref{thm:uniquedeg1-smallerror} by presenting an upper bound of $q=\bigO(\log n/\log\log n)$ on the number of queries to correct degree 1 polynomials from $\{0,1\}^n$ to $\mathbb{R}$ when the relative error is $\bigO(\log\log n/\log n)$. 

\begin{theorem}
    \label{thm:improved-uniquedeg1-smallerror}
    Let $\mathcal{P}_{1}$ be the set of degree $1$ polynomials from $\Boo^{n}$ to $\mathbb{R}$. Then for any $\delta < \bigO(\log\log n/\log n)$, $\mathcal{P}_{1}$ has a $(\delta, q)$-local correcting algorithm where $q =\bigO(\log n/\log\log n)$.
\end{theorem}

\begin{proof}
    The proof approach is similar to that of~\Cref{thm:uniquedeg1-smallerror}, except that we use a different correction gadget in~\Cref{lemma:decode-1n}: we will show that we can find $c_1,\dots,c_q \in \mathbb{R}$ and a distribution $\mathcal{D}$ over $(\{0,1\}^n)^q$ such that 
    \begin{itemize}
        \item $c_1 + \dots + c_q = 1$ and for all $i\in [n]$ and any sample $({\bf{y}}^{(1)},\dots,{\bf{y}}^{(q)})$ in the support of $\mathcal{D}$, we have $c_1 y^{(1)}_i + \dots + c_q y_i^{(q)}=1$.
        \item For each $j\in [q]$, ${\bf{y}}^{(j)}$ is $\varepsilon$-close to $U_n$ for some $\varepsilon= \sqrt{n}/q^{\Omega{(q)}}$ (Note that $\varepsilon$ was $\sqrt{n}/2^{\Omega{(q)}}$ in~\Cref{lemma:decode-1n}).
    \end{itemize}
    Once we prove the existence of such $c_j$'s and $\mathcal{D}$, arguing along the same lines as in~\Cref{thm:uniquedeg1-smallerror} allows us to conclude that there is a $(\delta,q)$-local correcting algorithm as long as $\varepsilon=\sqrt{n}/q^{\Omega(q)} \le 1/n$. Here we note that although the argument in~\Cref{thm:uniquedeg1-smallerror} was assuming $c_j$'s were integers, it still continues to hold when they are arbitrary real numbers as we assume that the underlying group is $\mathbb{R}$ in this section. Hence, we can take $q=\bigO(\log n/ \log\log n)$ with a sufficiently large constant factor.

    Returning to proving the existence of $\mathcal{D}$ (and $c_j$'s), we define the distribution $\mathcal{D}$ over $\{0,1\}^{n\times q}$ based on a probability distribution $p$ over $[q]$ (we let $p_j := \Pr[p=j]$ for $j\in [q]$) and a matrix $M \in \{0,1\}^{q\times q}$ (we denote the $i$-th row of $M$ by $M_i$, the $j$-th column by $M^{(j)}$ and the $(i,j)$-th entry by $M_{ij}$) as follows:
    \begin{itemize}
        \item For any $i\in [n]$, the $i$-th row of a sample from $\mathcal{D}$ is taken to be $M_j$, where $j\sim p$ is chosen independently across different $i$.
    \end{itemize}
    With this definition of $\mathcal{D}$, we note that it suffices if $M_1 = (1,\dots,1)$ and $\innerprod{M_j,{\bf c}}=1$ for all $j\in [q]$ in order to satisfy Item 1, where ${\bf c} := (c_1,\dots,c_q)$. Item 2 shall be satisfied with a careful choice of $M$ and $p$.  It will be more convenient to work in $\{\pm 1\}$ notation instead of $\{0,1\}$: let $N_{ij}:=1-2M_{ij} \in \{\pm 1\}$ for $i,j\in[q]$. Then the above conditions can be equivalently written as $N_1 = (-1,\dots,-1)$ and $\innerprod{N_j,{\bf c}}= \sum_j c_j - 2\innerprod{M_j,{\bf c}}=-1$ for all $j\in [q]$. Note that a solution for ${\bf c}$ satisfying Item 1 always exists if $N$ is non-singular as it simply amounts to solving $q$ linearly independent equations over $q$ variables. 
    
    To satisfy Item 2, we take $N$ to be an {\em anti-Hadamard} or {\em ill-conditioned} matrix as constructed by Alon and Vu~\cite{AlonVu} with some additional properties (Items 3 and 4 below):
    \begin{enumerate}
        \item $N$ is non-singular,
        \item the least singular value of $N$ is at most $1/q^{\Omega{(q)}}$, i.e., there exists a non-zero vector $v\in \mathbb{R}^q$ such that $\|N^T v\|_2/\|v\|_2 \le 1/q^{\Omega{(q)}}$,
        \item the vector $v$ above is a probability distribution, i.e., $v_j \ge 0$ for all $j$ and $\sum_{j=1}^q v_j =1$, and 
        \item the first row of $N$ is $N_1=(-1,\dots,-1)$.
    \end{enumerate}
    Although the construction in~\cite{AlonVu}
does not directly guarantee Items 3 and 4, we can easily achieve them by flipping the sign of appropriate columns and rows of $N$ and then by flipping and scaling the entries of $v$ (these changes do not affect the least singular value). We now define the probability distribution $p$ over $[q]$ to simply be $p_j=v_j$. This finishes the description of the matrix $N$ and the distribution $\mathcal{D}$. It remains to be shown for $({\bf y}^{(1)},\dots,{\bf y}^{(q)}) \sim \mathcal{D}$ we have that ${\bf y}^{(j)}$ is $\varepsilon$-close to $U_n$ for some $\varepsilon=\sqrt{n}/q^{\Omega(q)}$. Since the entries of ${\bf y}^{(j)}$ are mutually independent and the rows of $\mathcal{D}$ are sampled from the rows of $M$, by~\Cref{fact:close-to-uniform} it suffices to show that for all $j\in [q]$, we have  $\abs{\E_{j'\sim p}[M_{j'j}]-1/2} \le 1/q^{\Omega(q)}$. We prove this below:
\begin{align*}
    \paren{\E_{j'\sim p}[M_{j'j}]-1/2}^2 & = {\E_{j'\sim p}[N_{j'j}]}^2/4\tag{by definition of $N$}\\
    & = \paren{\sum_{j'=1}^q p_{j'}N_{j'j}}^2/4
     = {\innerprod{N^{(j)},v}}^2/4\tag{as $p_{j'}=v_{j'}$ by definition}\\
    & \le \sum_{j''=1}^q \innerprod{N^{(j'')},v}^2/4
     = v^\top NN^\top v/4
     = \|N^\top v\|_2^2/4\\
    & \le \|v\|_2^2/q^{\Omega{(q)}} 
     \le 1/q^{\Omega{(q)}}.\tag{using Item 2 and $\|v\|_2^2 \le q$}
\end{align*}
\end{proof}

\addtocontents{toc}{\protect\setcounter{tocdepth}{1}}
	
\end{document}